\documentclass[a4paper,11pt]{article}
\usepackage{fullpage}
\usepackage{libertine}

\usepackage{booktabs} 
\usepackage[ruled]{algorithm2e} 

\SetAlFnt{\small}
\SetAlCapFnt{\small}
\SetAlCapNameFnt{\small}
\SetAlCapHSkip{0pt}
\IncMargin{-\parindent}

\usepackage{amssymb,amsmath,amsfonts,amstext,amsthm}
\usepackage{thmtools} 
\usepackage{mathtools} 

\usepackage{arydshln}

\usepackage[dvipsnames]{xcolor}         
\definecolor{pku-red}{RGB}{139,0,18}
\usepackage[backref,colorlinks,citecolor=blue,linkcolor=black,urlcolor=pku-red,bookmarks=false]{hyperref}
\usepackage[nameinlink]{cleveref}

\renewcommand\url[1]{#1}
\usepackage{enumitem}

\usepackage{pgfplots}
\pgfplotsset{compat=newest}
\usepgfplotslibrary{ternary}
\usepackage{tikz}
\usetikzlibrary{patterns}

\usepackage{bbm}
\usepackage{natbib}%
\setcitestyle{sort,authoryear,square,citesep={;},aysep={,},yysep={;}}

\newcommand{\wt}{\widetilde}

\newcommand{\calA}{\mathcal{A}}

\newcommand{\calG}{\mathcal{G}}

\newcommand{\calM}{\mathcal{M}}
\newcommand{\calN}{\mathcal{N}}

\newcommand{\va}{\mathbf{a}}
\newcommand{\vb}{\mathbf{b}}

\newcommand{\ee}{\mathbf{e}}
\newcommand{\xx}{\mathbf{x}}
\newcommand{\yy}{\mathbf{y}}
\newcommand{\zz}{\mathbf{z}}
\newcommand{\bv}{\mathbf{v}}

\renewcommand{\gg}{\mathbf{g}}

\newtheorem{proposition}{Proposition}[section]
\newtheorem{lemma}[proposition]{Lemma}
\newtheorem{theorem}[proposition]{Theorem}
\newtheorem*{theorem*}{Theorem}

\newtheorem{observation}[proposition]{Observation}
\newtheorem{definition}{Definition}

\newtheorem{claim}{Claim}

\newtheorem*{maintheorem}{Main Theorem}

\DeclareMathOperator{\closure}{cl}
\DeclareMathOperator{\BR}{BR}
\DeclareMathOperator*{\argmax}{argmax}
\DeclareMathOperator*{\argmin}{argmin}

\DeclareMathOperator{\rank}{rank}
\newcommand{\rankp}[1]{\rank\left(#1\right)}

\newcommand{\ER}{\textnormal{ER}}
\newcommand{\SSE}{\textnormal{SSE}}
\newcommand{\AER}{\calA_{\ER}}
\newcommand{\ASSE}{\calA_{\SSE}}

\newcommand{\ots}{{\{1,2\}}}

\Crefname{equation}{Eq.}{Eqs.}

\newcommand{\detailsprovidedin}[1]{[cf. #1]}

\usepackage{mdframed}
\newmdenv[skipabove=2mm, skipbelow=2mm, backgroundcolor=black!10]{boxedtext*}
\newmdenv[skipabove=2mm, skipbelow=2mm, rightline=true, leftline=true]{boxedtext}

\newcommand\NoCaseChange[1]{#1}
\newcommand\email[1]{#1}

\title{Learning to Manipulate a Commitment Optimizer}

\author{Yurong Chen \\Peking University \\\email{chenyurong@pku.edu.cn} \and  Xiaotie Deng \\Peking University \\ \email{xiaotie@pku.edu.cn} \and Jiarui Gan$^*$\\ University of Oxford \\\email{jiarui.gan@cs.ox.ac.uk} \and Yuhao Li \\Columbia University \\ \email{yuhaoli@cs.columbia.edu}}

\date{}

\begin{document}

\maketitle

{\let\thefootnote\relax\footnotetext{$^*$Corresponding author.}}

\begin{abstract}
It is shown in recent studies that in a Stackelberg game the follower can manipulate the leader by deviating from their true best-response behavior. Such manipulations are computationally tractable and can be highly beneficial for the follower. Meanwhile, they may result in significant payoff losses for the leader, sometimes completely defeating their first-mover advantage. A warning to commitment optimizers, the risk these findings indicate appears to be alleviated to some extent by a strict information advantage the manipulations rely on. That is, the follower knows the full information about both players' payoffs whereas the leader only knows their own payoffs. In this paper, we study the manipulation problem with this information advantage relaxed. We consider the scenario where the follower is not given any information about the leader's payoffs to begin with but has to {\em learn} to manipulate by interacting with the leader. The follower can gather necessary information by querying the leader's optimal commitments against contrived best-response behaviors. Our results indicate that the information advantage is {\em not} entirely indispensable to the follower's manipulations: the follower can learn the optimal way to manipulate in polynomial time with polynomially many queries of the leader's optimal commitment.
\end{abstract}

\section{Introduction}

Strategy commitment is a useful tactic in many game-theoretic scenarios.
In anticipation that the other player, i.e., the follower, will respond optimally, a commitment optimizer, i.e., the leader, picks a strategy that maximizes their own payoff.
The interaction between the leader and the follower is often modeled and known as a {\em Stackelberg game} \citep{von1934marktform,von2010leadership}. The equilibrium of the game, called the {\em Stackelberg equilibrium}, captures the leader's optimal commitment.  
It is well-known that the leader's first-mover position comes with a payoff benefit: an optimal commitment always yields a higher payoff (often strictly higher) for the leader than what they obtain in any Nash equilibrium of the same game \citep{von2010leadership}. 

Yet, this first-mover advantage is not without a caveat.
The computation of an optimal commitment relies crucially on the follower's payoff information. This strong reliance offers the follower a means to manipulate the leader's commitment.
A series of recent studies considered this issue and investigated how a follower can induce an equilibrium different from the original one by misreporting their payoffs \citep{gan2019imitative, gan2019manipulating, nguyen2019imitative, birmpas2021optimally, chen2022optimal}. It is shown that finding an optimal manipulation is computationally tractable and there is a large space of outcomes that are realizable through manipulations. 
In the worst case, a manipulation may completely defeat the first-mover advantage of the leader and cause a significant payoff loss.\footnote{\citet{birmpas2021optimally} showed that any outcome can be induced as a Stackleberg equilibrium as long as it offers the leader at least the maximin value of the game---a lower bound of the leader's Nash payoff, in contrast to the upper bound promised by an optimal commitment when the follower behaves truthfully.}
Moreover, even if the leader is well aware of the possibility of such manipulations, they face an NP-hard problem to compute an optimal mechanism to counteract \citep{gan2019imitative}.

A warning to commitment optimizers, the risk these findings indicate may appear to be alleviated to some extent by a strict information advantage required by the manipulations. That is, the follower knows the full information about both players' payoffs, whereas the leader only knows their own.
This may not be the case in practice. 
In this paper, we consider a setting with this information asymmetry relaxed, where neither the leader nor the follower knows any payoff information of the opponent to begin with. 
The follower has to {\em learn} to manipulate by interacting with the leader and can gather necessary information by querying the leader's optimal commitments against contrived payoff functions. We are interested in understanding whether the follower can {\em efficiently} learn to solve the optimal manipulation problem. 

More specifically, we assume that the follower has query access to the leader's optimal commitment: there is an equilibrium oracle which answers whether a certain Stackelberg equilibrium can be induced by a given (fake) payoff function of the follower. 
The query access resembles an information exchange process during the course of interaction: the follower (mis)reports a payoff function to the leader, and the leader reacts by committing optimally with respect to the report; the optimal commitment is then observed by the follower. 
Payoff reporting may take the form of direct information exchange. 
For example, in online platforms, users (follower) set up their profiles and the platform (leader) offers personalized recommendations or pricing based on the information provided.\footnote{\url{https://medium.com/swlh/why-is-your-friend-getting-a-cheaper-uber-fare-than-you-ai-and-the-new-frontier-in-dynamic-pricing-2b7d908deed0}}
Additionally, it can also be realized via another layer of active learning from the leader's side, e.g., in green security games, the defender (leader) learns the optimal patrolling strategy by interacting with poachers (followers) \citep{fang2016green}.  

Our main result, presented as follows, is an affirmative answer to the question asked above.

\begin{maintheorem}[informal]
With access to an equilibrium oracle, a follower can learn an optimal payoff function to misreport in \textbf{polynomial time} and via \textbf{polynomially many queries}. 
The payoff function induces a Stackelberg equilibrium that maximizes the follower's (real) payoff among all inducible equilbria (i.e., equilibria that can be induced by some fake payoff matrix of the follower).
\end{maintheorem}

The result indicates that the information advantage is {\em not} entirely indispensable to the follower's manipulation, so it may not be safe to take it as a protection against manipulations.
Indeed, the issue is fairly widespread. The flourish of e-commerce and other online platforms, including various financial activities on crypto-currencies \citep{chen2022absnft}, offers many testbeds and realistic application areas for the problem we study.
In these domains, the leader interacts individually with a large number of different followers and aims to commit optimally against each of them.
A follower who intends to manipulate can easily forge pseudonym identities in a short period of time at low cost and without being detected by the leader 
(e.g., by creating multiple accounts for a web-based service such as crypto-currency).
Adding to the fact, the manipulations we consider are {\em imitative}---a term coined by \citet{gan2019imitative} meaning that the follower keeps behaving in accordance with the reported payoff function.
Such manipulations are almost impossible for the leader to detect in many cases.
Due caution is needed when one seeks to exploit the power of commitment in these domains.

Our approach to deriving the main result consists of two main components: (1) learning the gradient information of the leader's payoff function, and 
(2) constructing strategically equivalent games with this information to compute the follower's best manipulation strategy.
The latter requires computing a payoff matrix of the follower to induce her maximin value in the original game. Even in the full information setting, this problem requires a non-trivial approach in order to derive an efficient solution \cite{birmpas2021optimally}. In our partial information setting, additional difficulties come from the fact that we cannot fully recover all gradient information of the leader's payoff function. We note that the equilibrium oracle does not directly reveal the leader's payoff information. 
Therefore, to acquire useful information from the oracle, it requires carefully designed queries that expose strategy profiles of interest as Stackelberg equilibria. This task becomes more challenging, as the numbers of players' actions increase. 

\subsection{Related Work}

Our work directly relates to the recent line of work on follower deception in Stackelberg games as we mentioned above.
This line of work is motivated by an active learning approach to finding an optimal strategy to commit to \citep{letchford2009learning, balcan2015committment, blum2015learning, roth2016watch, peng2019learning}, which asks whether the leader's optimal commitment can be learned with query access to the follower's best response. 
\citet{gan2019imitative} first pointed out that this approach leads to an untruthful mechanism that can be manipulated by the following imitating responses as if they have a different payoff function. 
They also showed that designing an optimal mechanism to counteract the follower's manipulation is in general NP-hard even to find an approximate solution. The hardness contrasts the tractability of the computation of an optimal commitment in the full information setting, as shown in an early work of \citet{conitzer2006computing}.
\citet{gan2019manipulating} considered a security game scenario and showed that, under mild assumptions, in a Stackelberg security game the follower's optimal manipulation is always to misreport payoffs that make the game zero-sum; the leader gains only their maximin payoff as a result. This result does not hold true in general bi-matrix Stackelberg games, but the later result of \citet{birmpas2021optimally} indeed also revealed a general connection between payoff manipulations in Stackelberg games and the leader's maximin payoff. As we will discuss in detail later, this connection is also one of the cornerstones of our main result.
\citet{nguyen2019imitative} studied a similar security game scenario and also considered the follower's manipulation strategy when their payoff reporting is restricted in a ball of the true payoff.
More recently, \citet{chen2022optimal} further extended the line of work to extensive-form games. In another line of work \citet{kolumbus2022how, kolumbus2022auctions} studied the Nash equilibrium of the meta-game when multiple players attempt to manipulate simultaneously.

Besides the above line of work, in a more recent study, \citet{haghtalab2022learning} approached manipulations in Stackelberg games via a repeated game model and interpreted the follower's manipulations as a strategic behavior resulting from their far-sightedness. A framework is proposed in this study to derive efficient learning algorithm for a more patient leader (who does not discount future rewards) to induce truthful best responses from the less patient non-myopic follower (who discounts future rewards). Our setting is analogous to the opposite, where a more patient follower plays with a less patient leader. The patient follower explores the leader's payoff structure until having learned the optimal manipulation and behaves accordingly afterwards to exploit the leader.
We also note that similar interactions against non-myopic bidders have also been extensively explored in the online auction literature~\cite{amin2013learning, amin2014repeated, mohri2014optimal, liu2018learning, abernethy2019learning, golrezaei2021dynamic}.

\section{Preliminaries}
\label{sec:preliminaries}

Stackelberg games are a standard framework for studying strategy commitment in game theory.
In a Stackelberg game, a leader commits to a strategy and a follower best responds to this commitment.
We consider general bi-matrix games in this paper. A bi-matrix game $\calG = (u^L,u^F)$ is given by two matrices $u^L,u^F \in \mathbb{R}^{m \times n}$, which specify the leader's and the follower's payoffs, respectively.
The leader has $m$ actions (i.e., pure strategies) at their disposal, each corresponding to a row of the payoff matrices; 
and the follower has $n$ actions, each corresponding to a column. 
The entries $u^L (i,j)$ and $u^F(i,j)$ are the payoffs of the leader and the follower when a pair $(i,j) \in [m] \times [n]$ of actions is played.\footnote{For any positive integer $x$ we write $[x] \coloneqq \{1,\dots, x\}$.}

In the mixed strategy setting, the leader can further randomize their actions and the resulting distribution over the actions is called a mixed strategy, e.g., $\xx \in \Delta_m \coloneqq \left\{\xx \in \mathbb{R}_{\ge0}^m : \sum_{i \in [m]} x_i = 1 \right\}$.
Slightly abusing notation, we denote by $u^L(\xx,j) = \sum_{i \in [m]} x_i \cdot u^L(i,j)$ the expected payoff of the leader for a strategy profile $(\xx,j)$, and similarly denote by $u^F(\xx,j) = \sum_{i \in [m]} x_i \cdot u^F(i,j)$ the expected payoff of the follower. 
We will refer to a payoff matrix and the corresponding payoff function interchangeably throughout this paper.

A (pure) best response of the follower is then given by $j \in \BR(\xx)$, where 
\[
\BR(\xx) \coloneqq \argmax_{j \in [n]} u^F(\xx, j)
\]
is called the follower's {\em best response set}, or a {\em BR-correspondence} as a function $\BR: \Delta_{m} \to 2^{[n]}$. 
In most cases, it is without loss of generality to consider only pure strategy responses, because there always exists an optimal strategy that is pure. Hence, unless otherwise specified, all best responses are pure strategies throughout.
It will also be useful to define the inverse function of $\BR$: for any $j \in [n]$,
\[
\BR^{-1}(j) \coloneqq \left\{\xx \in \Delta_m : j \in \BR(\xx) \right\}
\]
is the set of leader strategies that incentivize the follower to best-respond $j$.

The optimal commitments of the leader are captured by the following optimization, with the assumption that the follower breaks ties by picking a $j \in \BR(\xx)$ in favor of the leader when there are multiple best responses in $\BR(\xx)$:
\begin{align}
\label{eq:sse}
	(\xx, j) \in \argmax_{\xx' \in \Delta_{m},\ j' \in \BR(\xx') } u^L (\xx', j').
\end{align}
The strategy profile $( \xx, j )$ is called a {\em strong Stackelberg equilibrium} (SSE). 
Alternatively, using the inverse function of $\BR$ gives the following equivalent definition of an SSE:
\begin{align}
\label{eq:sse-inv}
	(\xx, j) \in \argmax_{\xx' \in \BR^{-1}(j'),\ j' \in [n]} u^L (\xx', j').
\end{align}
SSE is the most widely used solution concept in the literature on Stackelberg games. The optimistic tie-breaking assumption adopted by it is justified by noting that this tie-breaking behavior can often be induced by an infinitesimal perturbation in the leader's strategy \citep{von2004leadership}. 

\begin{definition}[SSE and SSE response]
\label{def:sse}
A strategy profile $(\xx,j) \in \Delta_m \times [n]$ is said to be an SSE of a game $\calG = (u^L,u^F)$ if and only if \Cref{eq:sse} (or equivalently, \Cref{eq:sse-inv}) holds. 
An action $j$ of the follower is called an {\em SSE response} of $\calG$ if and only if $(\xx, j)$ is an SSE of $\calG$ for some $\xx \in \Delta_m$.
\end{definition}

\subsection{Equilibrium Manipulation via Payoff Misreporting}

According to the above definition, the optimal commitment of the leader, or the SSE, is a function of the follower's payoff matrix.
When this payoff information is private to the follower, the follower has a chance to manipulate the leader's commitment by reporting a fake payoff matrix.
At a high-level, to find out the optimal way to manipulate amounts to solving the following optimization problem:
\begin{subequations}
\label{eq:opt-deception}
\begin{align}
\max_{\tilde{u}^F, \xx, j} \quad & u^F(\xx, j) \tag{\ref{eq:opt-deception}} \\
\text{s.t.} \quad
& (\xx, j ) \text{ is an SSE of } \widetilde{\calG} = (u^L, \tilde{u}^F). \label{eq:opt-deception-const}
\end{align}
\end{subequations}
\citet{birmpas2021optimally} showed that this problem is tractable in the full information setting (i.e., when $u^L$ is known to the follower) and presented an elegant characterization of strategy profiles that satisfy~\Cref{eq:opt-deception-const}.
This characterization, summarized in \Cref{thm:maximin-characterization} below, forms a foundation of our technical results.
We follow the terminology by \citet{birmpas2021optimally} and define the inducibility as follows.
Note that \Cref{eq:opt-deception-const} also means that the manipulation is {imitative}:
the follower is required to respond according to the reported payoff $\tilde{u}^F$, and this ``best response'' decides their equilibrium utility.

\begin{definition}[Inducibility]
A payoff matrix $\tilde{u}^F \in \mathbb{R}^{m \times n}$ {\em induces} a strategy profile $(\xx,j)$ (with respect to $u^L$) if and only if $(\xx,j)$ is an SSE of $(u^L, \tilde{u}^F)$.
Moreover, $(\xx,j)$ is said to be {\em inducible} (with respect to $u^L$) if and only if it is induced by some $\tilde{u}^F$.
\end{definition}

\begin{theorem}[\citealt{birmpas2021optimally}]
\label{thm:maximin-characterization}
A strategy profile $(\xx,j)$ is inducible if and only if 
\[
u^L(\xx,j)\geq M_{[n]} \coloneqq \max_{\yy \in \Delta_m}\min_{k\in [n]}u^L(\yy,k).
\]
Moreover, a payoff matrix $\tilde{u}^F$ that induces $(\xx,j)$ can be computed in polynomial time.
\end{theorem}

The above result establishes an interesting connection between the follower's optimal manipulation and the leader's maximin value $M_{[n]}$. 
Intuitively, to induce $(\xx,j)$, the follower can respond in a way that is completely adversarial against the leader (whereby the leader only gets $M_{[n]}$) unless the leader plays $\xx$.
Hereafter, we extend the notation $M_{[n]}$ to every subset $S \subseteq [n]$ of the follower's actions, and define
\begin{align}
\label{eq:def-M}
M_S \coloneqq \max_{\xx \in \Delta_m} \min_{j \in S} u^L(\xx, j)\quad 
\text{ and}\quad
\calM_S \coloneqq \argmax_{\xx \in \Delta_m} \min_{j \in S} u^L(\xx, j),
\end{align}
which are the leader's maximin payoff and the set of maximin strategies when the follower's responses are restricted in $S$.
These two notations will be frequently used throughout the paper.

\subsection{Main Problem: Learning to Manipulate with SSE Oracle}

We consider the learning version of the follower's optimal manipulation problem, where {\em the leader's payoff matrix is unknown} to the follower (i.e., unknown to us). 
Instead, the follower only has query access to an {\em SSE oracle}, denoted $\ASSE$, which is able to answer whether a given solution $(\tilde{u}^F, \xx, j)$ to the above optimization problem satisfies \Cref{eq:opt-deception-const} or not---in other words, whether $(\xx,j)$ is induced by $\tilde{u}^F$.\footnote{The reason that we also specify $(\xx, j)$ as part of the input to $\ASSE$ is to sidestep the tricky case where there are multiple SSEs. Our results do not apply to the setting where we do not have the power to specify a particular SSE. We leave this setting as an interesting open problem. See our discussion in \Cref{sec:conclusion}.}
Conceptually, when the follower can interact repeatedly with the leader, they can try reporting different payoff matrices $\tilde{u}^F$ and observe the leader's optimal commitment against these matrices. 
The SSE oracle abstract this process.
Note that in our model the leader is unaware of the fact that the follower keeps changing their payoffs throughout the process. Rather, the leader thinks that they are interacting with different followers and, as we discussed earlier, this applies to scenarios where the follower can easily forge a large number of fake identities to elicit the leader's payoff information.
To put it differently, the leader commits to playing the optimal commitment against every reported payoff matrix.

\begin{definition}[SSE oracle]
\label{def:SSE-oracle}
Given a matrix $\tilde{u}^F$ and a strategy profile $(\xx, j) \in \Delta_m \times [n]$, the {\em SSE oracle}, denoted $\ASSE$, outputs whether $\tilde{u}^F$ induces $(\xx, j)$ or not.
\end{definition}

We consider the time and query complexity with respect to the bit-size of the leader's payoff matrix and assume that its entries are rational numbers.
Throughout, when we say that something can be computed efficiently, we mean that it can be computed with polynomially many queries to $\ASSE$ \textit{and} in polynomial time.
Our algorithms will make frequent use of binary search to find the exact value of an unknown rational number. We note that this can be done {\em exactly} in time polynomial in the bit-size of the target number, using a standard approach that searches on the Stern-brocot tree~\citep{graham1989concrete}.

\section{Warm-up and Approach Overview}

In this section, we provide an overview of our approach to solving the problem defined above.
We start with a warm-up, with two examples showing how the SSE oracle can be utilized to obtain some basic information that will be useful throughout the paper.

\subsection{Warm-up with $\ASSE$}

The oracle $\ASSE$ does not directly reveal any payoff values, but it exposes payoff information about SSEs. For example, if through the oracle we can confirm that two strategy profiles $(\xx,j)$ and $(\xx', j')$ are both SSEs, then we know that $u^L(\xx, j) = u^L(\xx', j')$. 
Hence, a main approach to obtaining information via $\ASSE$ is by designing payoff matrices that induce strategy profiles of interest as SSEs. 
In particular, in order for a profile $(\xx, j)$ to be an SSE, a necessary condition according to \eqref{eq:sse-inv} is that $\xx \in \argmax_{\xx' \in \wt{\BR}^{-1}(j)} u^L(\xx', j)$ with respect to some follower action $j$ and the BR-correspondence $\wt{\BR}^{-1}(j)$ of a fake payoff matrix. 
Using this necessary condition and $\ASSE$, we can easily, and efficiently, identify the following handy information; we  henceforth assume that they are known in the remainder of the paper.

\begin{observation}
\label{asn:M-j}
With query access to $\ASSE$, for every $j \in [n]$, the set 
$I_j \coloneqq \argmax_{i \in [m]} u^L(i, j)$
of the leader's (pure) best responses against $j$ can be computed efficiently.
Hence, $\calM_{\{j\}}$, which is the convex hull of $I_j$, can also be computed efficiently.
\end{observation}

Specifically, to decide whether $i \in I_j$ for a pure strategy $i \in [m]$, we can construct a matrix $\tilde{u}^F$ such that $\tilde{u}^F(i,j)=0$ and $\tilde{u}^F(i, j')=-1$ for all $j' \neq j$. This way $j$ is the strictly dominant strategy of the follower, 
so we have $\wt{\BR}^{-1}(j) = \Delta_m$ and $\wt{\BR}^{-1}(j') = \emptyset$ for all $j' \neq j$.
We then query $\ASSE$ to check whether $\tilde{u}^F$ induces $(i,j)$ as an SSE.
A ``yes'' answer implies that 
\[
u^L(i,j) = \max_{\xx \in \wt{\BR}^{-1}(j)} u^L(\xx,j) = \max_{\xx \in \Delta_m} u^L(\xx,j) \ge \max_{i' \in [m]} u^L(i',j).
\]
So $i \in I_j$.
A ``no'' answer implies the existence of $\xx \in \wt{\BR}^{-1}(j) = \Delta_m$ such that $u^L(i,j) < u^L(\xx,j)$. 
It follows that
$u^L(i,j) < u^L(\xx,j) \le \max_{i' \in [m]} u^L(i',j)$.
Hence, $i \notin I_j$.

\begin{observation}
\label{asn:relation-jik}
For every $i \in [m]$ and $j,k \in [n]$, the relation (i.e., $>$, $=$, or $<$) between $M_{\{j\}}$ and $u^L(i,k)$ can be decided efficiently with query access to $\ASSE$. 
\end{observation}

To decide the above relation, we slightly modify $\tilde{u}^F$ defined for deriving \Cref{asn:M-j}; we let $\tilde{u}^F(i,k)=0$ and keep all other entries the same. This gives the following BR-correspondence of $\tilde{u}^F$:
\begin{align*}
\wt{\BR}^{-1}(k) = \{i\},\quad
\wt{\BR}^{-1}(j) = \Delta_m,\quad
\text{and } \
\wt{\BR}^{-1}(j') = \emptyset \text{ for all } j' \in [n] \setminus \{j,k\}.
\end{align*}
Consequently, the best strategies to induce responses $j$ and $k$ 
yield payoffs $M_{\{j\}}$ and $u^L(i,k)$, respectively, for the leader. We can then query $\ASSE$ to check if $(i',j)$ and $(i,k)$ are SSEs for some $i' \in I_j$ to decide the relation between these two payoffs.

\subsection{Approach Overview}
\label{sec:approach-overview}

Now we present an overview of our approach.
Given the characterization by \Cref{thm:maximin-characterization}, the problem we want to solve, formulated as Problem~\ref{eq:opt-deception}, boils down to solving the following linear program (LP) for every $j\in[n]$ and selecting the one with the maximum optimal value. 
\begin{subequations}
\label{lp:maximin-inducibility}
\begin{align}
\max_{\xx \in \Delta_m} \quad & u^F(\xx, j) \tag{\ref{lp:maximin-inducibility}} \\
\text{s.t.} \quad
& u^L(\xx, j )\geq M_{[n]} \quad \text{\em (i.e., $(\xx,j)$ is inducible)}
\label{eq:maximin-inducibility-cons1}
\end{align}
\end{subequations}

With only access to $\ASSE$, we cannot hope to learn the exact payoff function $u^L$ to solve the above LP.
Hence, the idea is to learn a strategically equivalent representation of $u^L$, that is, a vector $\va_j = (a_{j1}, \dots, a_{j{m}}) \in \mathbb{Q}^{m}$ such that
\begin{equation}
\label{eq:aj}
u^L(\xx, j) \equiv \gamma_j \cdot \va_j \cdot \xx + \beta_j
\quad \text{ for some } \gamma_j \in \mathbb{R}_{>0} \text{ and } \beta_j \in \mathbb{R},
\end{equation}
where $\va_j \cdot \xx$ denotes the inner product of $\va_j$ and $\xx$.
Intuitively, $\va_j$ indicates the direction of the gradient of $u^L(\cdot, j)$.
As we demonstrate in \Cref{lmm:opt-to-inducibility}, knowing $\va_j$ alone (without $\gamma_j$ and $\beta_j$) allows us to reduce the LP to an {\em inducibility problem}: {\em decide whether a given strategy profile $(\xx, j)$ is inducible or not}.
To solve the inducibility problem means comparing $u^L(\xx, j)$ with the maximin payoff $M_{[n]}$ and this requires constructing a $\tilde{u}^F$ that induces either $(\xx, j)$ or a strategy profile that gives the maximin payoff.

In more detail, these procedures are summarized in \Cref{fig:main-approach}, which also include a special treatment of a degenerate case (i.e., when $M_{\{j\}} = M_{[n]}$) that prevents us from even learning $\va_j$.
The detailed implementations of Steps 1 and 2 are presented in the next sections: learning $\va_j$ in \Cref{sec:learningaj}, and learning $J$ and $\xx^*$ in 
\Cref{sec:learn-j-xxs}.
The reason that LP~\eqref{lp:maximin-inducibility} can be solved efficiently in Step~3 is given by \Cref{thm:equ-game}. 

\begin{figure}[t]
\begin{boxedtext*}
\begin{itemize}[leftmargin=*,itemsep=1mm]
\item[1.]
For each $j \in [n]$, decide if $M_{\{j\}} = M_{[n]}$.
If $M_{\{j\}} \neq M_{[n]}$, learn a vector $\va_j$ that satisfies \Cref{eq:aj}.
\detailsprovidedin{\Cref{sec:learningaj}}

\item[2.]
Identify a subset $J \subseteq [n]$ along with a leader strategy $\xx^* \in \Delta_m$, such that 
\begin{align}
\label{eq:J-M}
u^L(\xx^*,j) = M_J = M_{[n]}, \quad \text{ for all } j \in J.
\end{align}
\detailsprovidedin{\Cref{sec:learn-j-xxs,sec:first-reference-pair,sec:second-ref}}

\item[3.] 
Solve LP~\eqref{lp:maximin-inducibility} for every $j \in [n]$ with information obtained in the above two steps, hence also solving Problem~\eqref{eq:opt-deception}.
\detailsprovidedin{\Cref{thm:equ-game}}
\end{itemize}
\end{boxedtext*}
\vspace{-3mm}
\caption{Main approach overview. 
\label{fig:main-approach}}
\end{figure}

\begin{lemma}
\label{lmm:opt-to-inducibility}
Suppose that there is an efficient algorithm that decides correctly whether any given strategy profile $(\yy, k) \in \Delta_m \times [n]$ is inducible or not.
Then given $\va_j$ satisfying \eqref{eq:aj}, LP~\eqref{lp:maximin-inducibility} can be solved in polynomial time with query access to $\ASSE$.
The same can be achieved without $\va_j$ if $M_{\{j\}} = M_{[n]}$.
\end{lemma}

\begin{proof}
In the case where $M_{\{j\}} = M_{[n]}$, we have $u^L(\xx, j) \ge M_{[n]}$ if and only if $\xx \in \calM_{\{j\}}$. 
By \Cref{asn:M-j}, LP~\eqref{lp:maximin-inducibility} is then equivalent to solving $\max_{\xx \in \Delta_m} u^F(\xx, j)$, subject to $x_i = 0$ for all $i \notin I_j$, which is an LP with all parameters known. Hence, LP~\eqref{lp:maximin-inducibility} can be solved efficiently.

In the case where $M_{\{j\}} \neq M_{[n]}$, we are given $\va_j$.
We can rewrite the constraint \Cref{eq:maximin-inducibility-cons1} as $\va_j \cdot \xx \ge d_j^*$, where $d_j^* = (M_{[n]} - \beta_j)/\gamma_j$.
Note that $d_j^*$ is still unknown since $\beta_j$ and $\gamma_j$ are unknown. 
However, now that there is an algorithm to solve the inducibility problem, we can learn $d_j^*$ using binary search as follows. For any given $d$, pick arbitrary $\xx \in \Delta_m$ such that $\va_j \cdot \xx = d$. If $(\xx, j)$ is inducible then we know that $u^L(\xx, j) \ge M_{[n]}$ by \Cref{thm:maximin-characterization} and hence, $d \ge d_j^*$; otherwise, we know that $d < d_j^*$.  
Knowing $d_j^*$, we can then solve LP~\eqref{lp:maximin-inducibility} efficiently.
\end{proof}

\begin{restatable}{theorem}{thmequgame}
\label{thm:equ-game}
Suppose that the following elements are given: a vector $\va_j$ satisfying \eqref{eq:aj} for every $j \in [n]$ such that $M_{\{j\}} \neq M_{[n]}$; moreover, $J \subseteq [n]$ and $\xx^* \in \Delta_m$ satisfying \eqref{eq:J-M}. 
Then for every $j \in [n]$, LP~\eqref{lp:maximin-inducibility} can be solved in polynomial time with query access to $\ASSE$.
\end{restatable}

\begin{proof}[Proof sketch]
By \Cref{lmm:opt-to-inducibility}, LP~\eqref{lp:maximin-inducibility} reduces to  deciding the inducibility of any given strategy profile $(\yy, k)$.
We show that this inducibility problem can be efficiently solved.
Specifically, we argue that the following algorithm produces an {\em inducibility witness} $\tilde{u}^F$ for $(\yy, k)$ in polynomial time.
That is, $(\yy, k)$ is inducible (with respect to $u^L$) if and only if it is an SSE in $(u^L, \tilde{u}^F)$.
Hence, by querying $\ASSE$ to check whether $(\yy, k)$ is an SSE of $(u^L, \tilde{u}^F)$, we can decide the inducibility of  $(\yy, k)$. The stated result then follows.

\begin{boxedtext}
Construct an inducibility witness $\tilde{u}^F$ of $(\yy, k)$:
\begin{itemize}[leftmargin=*,itemsep=1mm]
\item[1.] 
For every $j$ such that $M_{\{j\}} = M_{[n]}$, since $\va_j$ is not given, let $\va_j$ be a vector in $\mathbb{R}^{m}$ such that: $a_{ji} = 1$ if $i \in \calM_{\{j\}}$, and $a_{ji} = 0$ otherwise.

\item[2.] 
Construct a payoff matrix $\tilde{u}^L \in \mathbb{R}^{m \times n}$ corresponding to the following payoff function:
\begin{equation}
\label{eq:equ-uL}
\tilde{u}^L(\xx, j) = 
\begin{cases}
\va_j \cdot \xx - \va_j \cdot \xx^*, & \text{ if } j \in J \setminus\{k\}; \\
\va_k \cdot \xx - \va_k \cdot \yy, & \text{ if } j = k; \\
1, & \text{ otherwise. }
\end{cases}
\end{equation}

\item[3.] 

Decide if $(\yy, k)$ is inducible with respect to $\tilde{u}^L$:
\begin{itemize}
\item 
If it is inducible, output a payoff matrix $\tilde{u}^F \in \mathbb{R}^{m \times n}$ such that $(\yy, k)$ is an SSE in $(\tilde{u}^L, \tilde{u}^F)$;

\item
Otherwise, output an arbitrary $\tilde{u}^F \in \mathbb{R}^{m \times n}$.
\end{itemize}

\end{itemize}
\end{boxedtext}

Namely, the above algorithm constructs a ``surrogate'' matrix $\tilde{u}^L$ using the given information $\va_j$, $J$, and $\xx^*$.
Then a witness $\tilde{u}^F$ is produced using $\tilde{u}^L$. 
The polynomial run-time of the algorithm is readily seen. In particular, Step~3 can be done in polynomial time according to \Cref{thm:maximin-characterization}. (Note that all parameters of $\tilde{u}^L$ is known, so this is a full-information setting.) 
It can also be proven that $(\yy, k)$ is inducible if and only if it is an SSE in $(u^L, \tilde{u}^F)$, where $\tilde{u}^F$ is the matrix produced by the above algorithm.
The details can be found in \Cref{sec:proof-of-thm:equ-game}.
\end{proof}

\section{Learning $\va_j$ - Directions towards which Leader's Utility Increases}
\label{sec:learningaj}

\begin{figure}[t]
	\center
	\begin{tikzpicture}
	
	\begin{ternaryaxis}[
	ternary limits relative=false,
	width= 45mm,
 	xtick=\empty,
 	ytick=\empty,
 	ztick=\empty,
	area style,
	grid=none,
	clip=false,
	xmin=0, xmax=1.0,
	ymin=0, ymax=1.0,
	zmin=0, zmax=1.0,
    axis on top,
	]
	
	\addplot3 coordinates {
	    (1, 0, 0)
		(0.32, 0.68, 0)
		(0.45, 0, 0.55)
        (1, 0, 0)
	};
	
	\addplot3 coordinates {
		(0, 1, 0)
		(0.6, 0.4, 0)
        (0.7, 0, 0.3)
		(0, 0, 1)
	};

 	\draw[black, thick] (0.57, 0.53, -0.1) -- (0.72, -0.1, 0.38);

	\node at (0.76, 0.12, 0.12) {\small $C$};
	\node at (0.36, 0.32, 0.32) {\footnotesize $\widetilde{\BR}^{-1}(n)$};
    
    \addplot3[only marks, color=BrickRed] coordinates{
        (0.6, 0.4, 0)
        (0.7, 0, 0.3)
    };
    
    \addplot3[only marks, color=black] coordinates{
        (1, 0, 0)
    };
	
	\node [black, rotate=14, anchor=west] at (0.72, -0.1, 0.38) {\small $\va_{n} \cdot \xx = d$};
	
	\end{ternaryaxis}
	
	\end{tikzpicture}
    \caption{
    A BR-correspondence $\widetilde{\BR}$ constructed for learning $\va_n$.
    The triangle represents the strategy space $\Delta_m$ of the leader.
    The region labeled $C$ is $\bigcup_{j \in [n-1]} \widetilde{\BR}^{-1}(j)$, which contains strategies inducing responses $j\neq n$ of the follower.
    The black dot at the top represents $\calM_{\{n\}}$ in this example. 
    The two red dots on the boundary of $\widetilde{\BR}^{-1}(n)$ are two critical points defining the boundary hyperplane.
	\label{fig:aj}}
\end{figure}

For notational simplicity, we present an algorithm for learning $\va_n$ instead of $\va_j$. The cases with other $\va_j$'s are analogous and can be handled by appropriate relabeling.

To learn $\va_n$, the high-level idea is to construct a BR-correspondence $\widetilde{\BR}$, such that the boundary of $\widetilde{\BR}^{-1}(n)$ is aligned with a hyperplane with norm vector $\va_n$.
Since we are searching ``in the dark'' and only have access to $\ASSE$, 
we scan through possible positions of the boundary in the hope of a position where all the points on the boundary are SSE strategies.
Take \Cref{fig:aj} as an example, we want to adjust the two red points defining the boundary to a position where both points are SSE strategies (which form SSEs along with response $n$ of the follower). The boundary then aligns with the contour of $u^L(\cdot,n)$ and its norm vector aligns with $\va_n$. 

The above idea is formalized in Lemma~\ref{lmm:ref-points-conditions} below, which also considers degenerate cases where $\va_n$ is parallel to a facet of $\Delta_m$ (e.g., when $\va_n$ is parallel to an edge of the simplex in \Cref{fig:aj}). Special treatment is needed for such degenerate cases, as it shall be clear in the sequel.
For ease of description, we reorder the leader's actions according to \Cref{asn:uL-large-small} throughout this section.
(The payoff information needed for the reordering is known because of \Cref{asn:M-j}.)
Specifically, in the case where $m_1=1$, we can simply let $\va_n = \mathbf{0}$ without further learning it, so we can assume that $m_1 > 1$.

\begin{observation}
\label{asn:uL-large-small}
Without loss of generality, we can assume that
$u^L(m_1,n)= u^L(m_1+1,n) = \cdots= u^L(m,n)= M_{\{n\}}$, $1 < m_1 \le m$, and $u^L(k,n) < M_{\{n\}}$ for all $k = 1,\dots, m_1-1$.
\end{observation}

\begin{restatable}{lemma}{lmmrefpointsconditions}
\label{lmm:ref-points-conditions}
Suppose that the following properties hold for strategies $\xx^1,\dots, \xx^{m_1-1} \in\Delta_m$:
\begin{itemize}[itemsep=1mm]
\item[(a)] $u^L(\xx^1, n)=\dots=u^L(\xx^{m_1-1},n)$;
\item[(b)] $x_i^i>0$ for all $i\in[m_1-1]$; and
\item[(c)] $1-x_i^i = \sum_{k=m_1}^{m} x_k^i$ for all $i\in[m_1-1]$.
\end{itemize}
Let $\vb =\left(-\frac{1}{x_{1}^{1}},-\frac{1}{x_{2}^{2}},\cdots,-\frac{1}{x_{m_1-1}^{m_1-1}},0,\cdots,0 \right)$. 
Then there exists $\gamma \in \mathbb{R}_{>0}$ and $\beta \in \mathbb{R}$ such that $u^L(\xx, n) = \gamma \cdot \vb \cdot \xx+\beta$ for all $\xx \in \Delta_m$.
\end{restatable} 

\begin{proof}
Recall that by \Cref{asn:uL-large-small}, $u^L(m_1,n)=\cdots=u^L(m,n)$. 
Hence, 
\[u^L(\xx,n) \equiv \sum_{i=1}^m u^L(i,n) \cdot x_i \equiv \sum_{i=1}^{m_1-1} \left(u^L(i,n)-u^L(m,n) \right)\cdot x_i+u^L(m,n).\]
Given Property~(c), for each $i\in[m_1-1]$, we have 
$u^L(\xx^i,n)= \left(u^L(i,n)-u^L(m,n) \right) \cdot x_i^i+u^L(m,n)$, so Property~(a) implies that:
\[
\left (u^L(1,n)-u^L(m,n) \right) \cdot x_1^1=\cdots= \left( u^L(m_1-1,n)-u^L(m,n) \right) \cdot x_{m_1-1}^{m_1-1}.
\]
Let $\gamma=- \left(u^L(1,n)-u^L(m,n) \right) \cdot x_1^1$ and $\beta=u^L(m,n)$.
Hence, $\gamma > 0$, and since $x_i^i > 0$ by Property~(b), we get that $u^L(\xx, n) \equiv \gamma \cdot \vb \cdot \xx + \beta$.
\end{proof}

Following \Cref{lmm:ref-points-conditions}, we aim to find $m_1 - 1$ critical points $\xx^1,\dots,\xx^{m_1-1}$ with the listed properties. In  \Cref{fig:aj}, These are the red points defining the boundary hyperplane shared by area $C$.
To ensure Property~(a), the critical points need to form SSEs with action $n$. 
Ideally, we can design the BR-correspondence in a way such that the follower responds to strategies in $C$ (as in \Cref{fig:aj}) with a bad action to the leader, so that the best strategy of the leader in $C$ does not outperform the critical points. 
Sometimes this requires using more than one follower response to ``cover'' $C$.
We introduce the following useful concept called {\em maximin-cover}, or {\em cover} for short.

\begin{definition}[Cover]
\label{def:cover}
A payoff matrix $\mu \in \mathbb{R}^{m \times n}$ of the follower is said to be a {\em maximin-cover} (or {\em cover}) of $S \subseteq [n]$ if and only if
\begin{equation}
\label{eq:cover}
\max_{\xx \in \calM_S} \max_{j \in \wt{\BR}(\xx)} u^L(\xx, j) < M_S,
\end{equation}
where $\wt{\BR}$ denotes the BR-correspondence of $\mu$.
It is said to be a {\em proper cover} of $S$ if it holds in addition that $S \cap \wt{\BR}(\xx) = \emptyset$ for all $\xx \in \Delta_m$.
\end{definition}

We  arrange the BR-correspondence in region $C$ according to a cover of $\{n\}$, in an attempt to make the leader's maximum attainable payoff in $\widetilde{\BR}^{-1}(n)$ surpass that in $C$.
To see how this could work, consider moving the hyperplane $\va_n \cdot \xx = d$ in \Cref{fig:aj} towards the top. As it approaches the vertex at the top, $\widetilde{\BR}^{-1}(n)$ and $C$ will approach $\Delta_m$ and $\calM_{\{n\}}$, respectively.
The maximum attainable payoffs in $\widetilde{\BR}^{-1}(n)$ and $C$ then approaches 
$\max_{\xx \in \Delta_m} u^L(\xx, n) = M_{\{n\}}$ and $\max_{\xx \in \calM_{\{n\}}, j \in \wt{\BR}(\xx)} u^L(\xx, j)$, respectively, which correspond to the right and left sides of \Cref{def:cover} (with $S = \{n\}$).
Hence, given \Cref{def:cover}, Property~(a) can be achieved when the hyperplane is placed sufficiently close to the top.
The additional requirement that a cover is {\em proper} is useful as we do not want any strategies in $C$ to induce $n$, which will eventually alter the boundary of $\widetilde{\BR}^{-1}(n)$. 
It turns out that this requirement is actually not strictly more demanding: according to \Cref{lmm:proper-cover}, any cover can be efficiently converted into a proper cover. Therefore, in the remainder of the paper, we simply refer to a proper cover as a cover.
Moreover, a simple characterization of a cover is given in \Cref{lmm:cover-iff}: the existence of a cover requires a gap between $M_S$ and the maximin value $M_{[n]}$.

\begin{restatable}{lemma}{lmmpropercover}
\label{lmm:proper-cover}
Given a cover of set $S \subseteq [n]$, a proper cover of $S$ can be constructed in polynomial time.
\end{restatable}

\begin{restatable}{lemma}{lmmcoveriff}
\label{lmm:cover-iff}
For any $S\subseteq[n]$, a cover of $S$ exists if and only if $M_S > M_{[n]}$. 
\end{restatable}

Next, we first demonstrate in \Cref{sec:find-xxs} that given a cover of $\{n\}$ a set of strategies $\xx^1, \dots, \xx^{m_1 - 1}$ satisfying the properties in \Cref{lmm:ref-points-conditions} can be computed efficiently, and hence we learn $\va_j$.
In the special case where $\{n\}$ does not admit a cover, \Cref{lmm:cover-iff} implies that $M_{\{n\}}=M_{[n]}$. This means that $J = \{n\}$ and an arbitrary $\xx^* \in \calM_{\{n\}}$ already form a tuple satisfying \eqref{eq:J-M}, so by \Cref{thm:equ-game}, we are done without learning $\va_n$.
We demonstrate how to compute (or decide the existence of) a cover in \Cref{sec:computing-cover}.
In summary, the approach to learning $\va_n$ is given in \Cref{fig:learning-an}.

\begin{figure}[t]
\begin{boxedtext*}
\begin{itemize}[leftmargin=*,itemsep=1mm]
\item[1.]
Compute a cover $\mu$ of $\{n\}$. 
\detailsprovidedin{\Cref{sec:computing-cover}}

\begin{itemize}
\item If no cover of $\{n\}$ exists, claim that $M_{\{n\}}=M_{[n]}$; Pick $J = \{n\}$ and an arbitrary $\xx^* \in \calM_{\{n\}}$, and go to Step~3 in \Cref{fig:main-approach}.
\end{itemize}

\item[2.] 
Use $\mu$ to compute strategies $\xx^1, \dots, \xx^{m_1-1}$ satisfying the properties in \Cref{lmm:ref-points-conditions};
Output $\va_n = \left(-\frac{1}{x_{1}^{1}},-\frac{1}{x_{2}^{2}},\cdots,-\frac{1}{x_{m_1-1}^{m_1-1}},0,\cdots,0 \right)$.
\detailsprovidedin{\Cref{sec:find-xxs}}

\end{itemize}
\end{boxedtext*}
\vspace{-3mm}
\caption{Summary of the approach to learning $\va_n$. (Relabeling $n$ to $j$ gives the procedures to learning each $\va_j$.) 
\label{fig:learning-an}}
\end{figure}

\subsection{Computing $\xx^1, \dots, \xx^{m_1-1}$}
\label{sec:find-xxs}

\begin{figure}[t]
    \centering
    \includegraphics[width=0.45\textwidth]{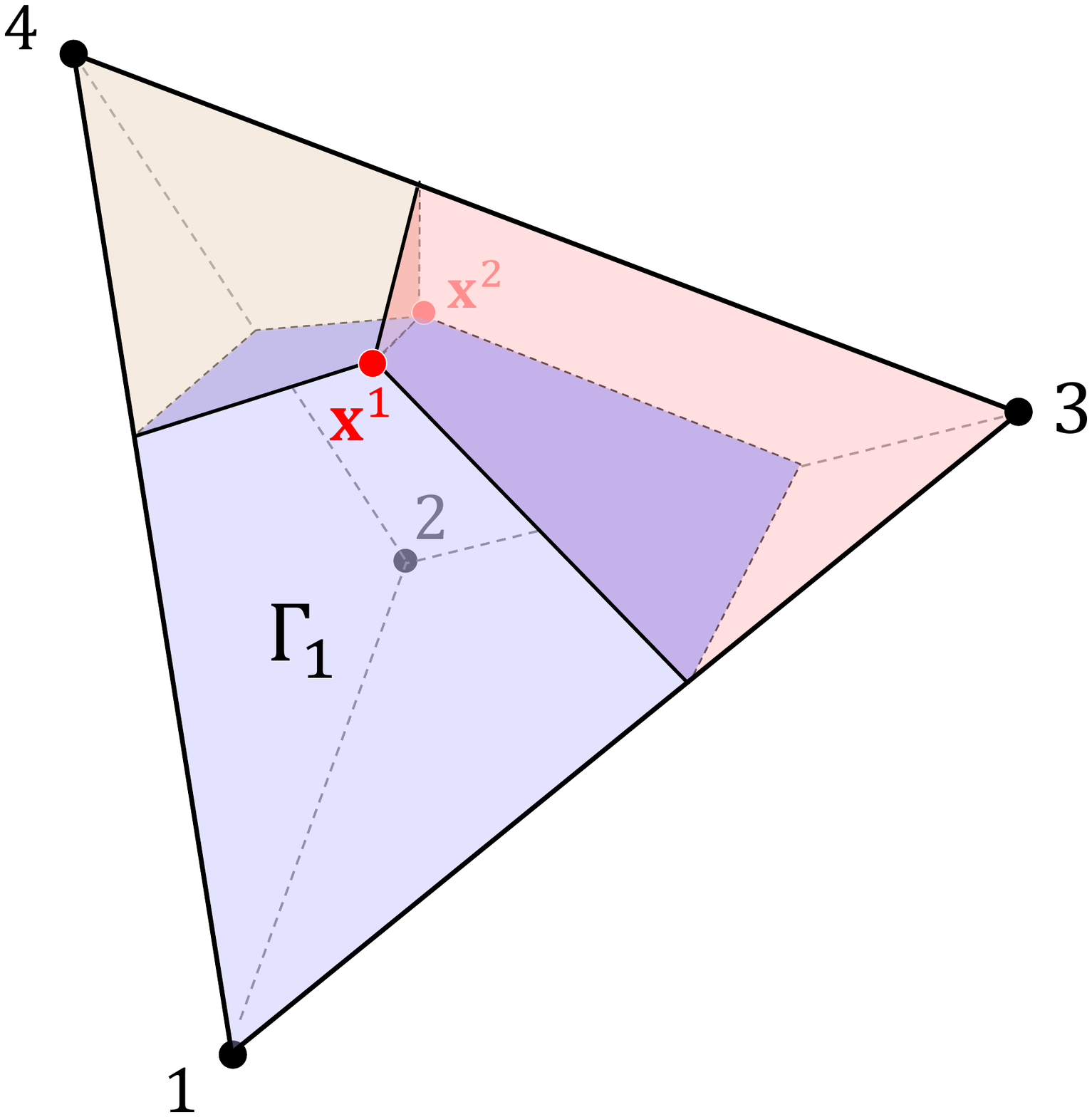}
    \caption{An illustration of $f_{\gg}(n)$ where $m=4$, $n=3$, and $m_1 = 3$ (which means $u^L(n, 3) = u^L(n, 4) = M_{\{n\}}$ according to \Cref{asn:uL-large-small}).
    The polytope represents $\Delta_m$.
    The three hyperplanes separating the regions are boundaries of $f_\gg^{-1}(j)$. 
    $\Gamma_1$ defined in \Cref{eq:Qi} is the surface at the front, where $x_1 + x_3 + x_4 = 1$ for all $\xx$.}
    \label{fig:lmm-boundary-point}
\end{figure}

We now describe how to compute $\xx^1, \dots, \xx^{m_1-1}$, given a cover $\mu$ of $\{n\}$.
For each $i\in [m_1-1]$, define the following set 
\begin{equation}
\label{eq:Qi}
\Gamma_i \coloneqq \left\{\xx \in \Delta_m: 1 - x_i = \sum_{k=m_1}^m x_k \right\},
\end{equation}
which is a facet of $\Delta_m$ that contains leader strategies satisfying Property~(c) stated in \Cref{lmm:ref-points-conditions}.  
We aim to construct a function $\tilde{u}^F$ that induces a strategy $\xx^i\in \Gamma_i$ to form an SSE with $n$, for each $i\in [m_1-1]$. 
This requires that the leader's maximum attainable payoff with respect to $\tilde{u}^F$ is achievable at every $\Gamma_i$. If $\calM_{\{n\}}$ is covered by only one action of the follower, say action $i$, this is relatively easy to achieve because the boundary separating $\widetilde{\BR}^{-1}(n)$ and $\widetilde{\BR}^{-1}(i)$ will be a hyperplane as in \Cref{fig:aj}. However, if more than one action is used to cover $\calM_{\{n\}}$, the shape of the separation surface may become irregular, possibly with vertices sticking out in its interior. We need a more sophisticated construction to ensure that these interior points do not yield higher payoffs than the best leader strategies in $\Gamma_i$.

Our construction proceeds as follows, where we let $\xx^1, \dots, \xx^{m_1-1}$ be parameterized by a vector $\gg = (g_{1}, \dots, g_{m_1 -1}) \in \mathbb{R}^{m_1 - 1}$.
The construction ensures that one optimal strategy of the leader always appears in some $\Gamma_i$, and with fine tuning of $\gg$ this can further be guaranteed for all $\Gamma_i$.
\begin{boxedtext}
\begin{itemize}[leftmargin=*,itemsep=1mm]
\item[1.] 
Given $\gg = (g_{1}, \dots, g_{m_1 -1})$, let
\begin{align}
\label{eq:tuF-g}
\tilde{u}_\gg^F(\xx,j) \coloneqq 
\begin{cases}
\sum_{i = m_1}^{m} x_i \cdot \mu(i, j) & \text{ if } j \in [n-1]; \\
\sum_{i = 1}^{m_1 - 1} x_i \cdot g_i + W \cdot \sum_{i = m_1}^{m} x_i & \text{ if } j = n.
\end{cases}
\end{align}
where $W \coloneqq \min_{i, j} \mu(i,j) - 1$.

\item[2.]
For every $i\in[m_1-1]$, pick an arbitrary leader strategy $\xx^i$ such that
\begin{equation}
\label{eq:xx-i}
\xx^i \in \argmin_{\xx \in \Gamma_i \cap f_\gg^{-1}(n)} x_i,
\end{equation}
where $f_\gg: \Delta_m \to 2^{[n]}$ denotes the BR-correspondence of $\tilde{u}^F_\gg$.
\end{itemize}
\end{boxedtext}

For simplicity, we omit the dependencies of $\xx^i$ on $\gg$ in the notation. 
We shall show soon that the arbitrary choice of $\xx^i$ in \eqref{eq:xx-i} suffices for our purpose.
\Cref{fig:lmm-boundary-point} provides an illustration of the notions defined above.

We show how to find an appropriate $\gg$, so that $\xx^1, \dots, \xx^{m_1-1}$ selected above satisfy the properties in \Cref{lmm:ref-points-conditions}.
Indeed, the following lemma indicates that Properties~(b) and (c) hold as long as we choose $\gg$ with $g_i > 0$ for all $i$.

\begin{restatable}{lemma}{lmmxigt}
\label{lmm:xi-gt-0}
$x_i^i > 0$ for any choice of $\gg$. Moreover, if $g_i > 0$ then
$\Gamma_i \cap f_\gg^{-1}(n) \neq \emptyset$.
\end{restatable}

\begin{proof}
For any $\xx \in \Gamma_i$, if $x_i = 0$ then $\xx \notin f_\gg^{-1}(n)$ because $\tilde{u}_{\gg}^F$ is a cover of $\{n\}$.
Moreover, if $g_i > 0$, then consider a special $\xx \in \Gamma_i$ such that $x_i = 1$ and $x_k = 0$ for all $k \in [n] \setminus \{i\}$. It can be verified that $\xx \in f_\gg^{-1}(n)$.
\end{proof}

\Cref{lmm:minxi-maxuL} shows that the leader's utility on $(\xx,n)$ for $\xx \in \Gamma_i$ only depends on the value of $x_i$. 

\begin{restatable}{lemma}{lmmminximaxuL}
\label{lmm:minxi-maxuL}
For all $\xx \in \Gamma_i$, it holds that $u^L(\xx, n) =  c \cdot x_i + d$ for some constants $c < 0$ and $d$.
\end{restatable}

\begin{proof}
Since $\xx \in \Gamma_i$, we have $x_k = 0$ for all $k = [m_1 - 1] \setminus \{i\}$, and hence 
$u^L(\xx, n) 
= \sum_{k=m_1}^m u^L(k,n) \cdot x_k + u^L(i,n) \cdot x_i$. Replacing $u^L(k,n)$ with $M_{\{n\}}$ for all $k = m_1, \dots, m$ (by \Cref{asn:uL-large-small}) and rearranging the terms give the desired result.
\end{proof}

To further ensure Property~(a), we define the following set for every $\gg$ with $g_i > 0$ for all $i$:
\[
I_\gg \coloneqq \left\{ i \in [m_1 - 1] : (\xx^i,n) \text{ is an SSE in game } \left(u^L, \tilde{u}_\gg^F \right) \right\}.
\]
Clearly, given $\xx^1,\dots,\xx^{m_1-1}$, $I_\gg$ can be computed using oracle $\ASSE$.
Moreover, \Cref{lmm:minxi-maxuL} implies that $\argmin_{\xx \in \Gamma_i \cap f_\gg^{-1}(n)} x_i = \argmax_{\xx \in \Gamma_i \cap f_\gg^{-1}(n)} u^L(\xx, n)$, so $I_\gg$ is well-defined---it is independent of the specific choice of $\xx^i$ in \eqref{eq:xx-i}. 
We then aim to find $\gg$ such that $I_\gg = [m_1-1]$, so that Property~(a) holds.
\Cref{thm:xxi} implies that 
this can be achieved
in polynomial time by inductively applying this result, thus leading to an efficient algorithm for computing $\xx^1, \dots,\xx^{m_1 -1}$.

\begin{restatable}{theorem}{thmxxi}
\label{thm:xxi}
Suppose that $\gg \in \mathbb{Q}_{>0}^{m_1 - 1}$.
If $I_\gg \neq [m_1 -1]$, then there exists $\gg' \in \mathbb{Q}_{>0}^{m_1 - 1}$, such that $I_\gg \subsetneq I_{\gg'}$.
Moreover, $\gg'$ can be computed in time polynomial in the bit-size of $\gg$.
\end{restatable}

\begin{proof}[Proof sketch]
To find $\gg'$, our approach is to increase $g_i$ if $(\xx^i, n)$ is not yet an SSE.
Intuitively, increasing each $g_i$ will cause $u^L(\xx^i, n)$ to increase, so the hope is that when $u^L(\xx^i, n)$ is sufficiently large,  $(\xx^i, n)$ becomes an SSE.
Attention needs to be paid to the possibility that increasing $g_i$ might also cause $\max_{\xx \in f^{-1}(n)\setminus (\cup_{i\in[m_1-1]}\Gamma_i)}u^L(\xx, n)$ to increase even faster for strategies not in any $\Gamma_i$, so $(\xx^i, n)$ never outperform some $(\xx, n)$ not in $\bigcup_{i\in[m_1-1]} \Gamma_i$.
Thanks to the way $\tilde{u}^F_{\gg}$ is designed, this possibility can be eliminated by \Cref{lemma:boundary-point}.
It indicates that $\xx^1, \dots, \xx^{m_1 - 1}$ are representatives of the leader's best choice in $f_\gg^{-1}(n)$.

The proof is then broken down into the following two cases.
\begin{itemize}
\item {\bf Case 1.} 
If $I_\gg = \emptyset$, we show that by increasing all $g_i$ to a sufficiently large number $N$, some $(\xx^i, n)$ will become an SSE. Moreover, $N$ can be bounded from above by a polynomial in the input size, so we can find a vector $\gg'$ such that $g'_i \ge N$ for all $i$ in polynomial time. See \Cref{lmm:Ig-empty}.

\item {\bf Case 2.} 
If $I_\gg \neq \emptyset$, we show that by increasing $g_i$ for an arbitrary $i \notin I_\gg$ to an appropriate number $g'_i$ (while fixing $g_k$ for all $k\neq i$), the strategy profile $(\xx^i, n)$ will become a new SSE in addition to the existing ones. Moreover, $g'_i$ can be computed in polynomial time. See \Cref{lmm:Ig-nonempty}. 
\qedhere
\end{itemize}

\begin{restatable}{lemma}{lemmaboundarypoint}
\label{lemma:boundary-point}
For any $\gg \in \mathbb{R}_{>0}^{m_1 -1}$, it holds that $\max_{\xx \in f_\gg^{-1}(n)} u^L(\xx, n) = \max_{i \in [m_1 - 1]} u^L({\xx}^i,n)$.
\end{restatable}

\end{proof}

\subsection{Computing a Cover of $\{n\}$}
\label{sec:computing-cover}

Now we describe how to compute a cover of $\{n\}$ to complete this section.
We first deal with an easy case, where $n \notin S^* \coloneqq \argmin_{j \in [n]} M_{\{j\}}$.\footnote{Recall that we can efficiently decide whether $n \in S^*$ or not using $\ASSE$: to compare $M_{\{j\}}$ and $M_{\{j'\}}$, pick $\xx \in \calM_{\{j'\}}$ according to \Cref{asn:M-j} and compare it with $M_{\{j\}}$ according to \Cref{asn:relation-jik}.}
In this case, any arbitrary $\ell \in S^*$ gives 
\[
u^L(\xx, \ell) \le M_{\{\ell\}} < M_{\{n\}} 
\]
for all $\xx \in \Delta_m$. 
Hence, any payoff matrix in which $\ell$ is a strictly dominant strategy forms a cover of $\{n\}$.
Using our result in \Cref{sec:find-xxs} and \Cref{lmm:ref-points-conditions}, we can then compute $\va_n$.
By relabeling $n$ to $j$, this approach can also be used to find a cover of $\{j\}$ and obtain $\va_j$ for any $j \notin S^*$. 
Hence, in what follows we can assume that $\va_j$ is given for all $j \in [n] \setminus S^*$. 

Next, consider the case where $n \in S^*$.
To deal with this case, we present a more general result, \Cref{thm:find-cover-given-aj}, which finds a cover for any $S \subseteq [n]$. This method will also be useful for our argument in the next sections, where we need to find a cover for size-2 subsets of $[n]$. 
To apply this method requires a base function that make the leader gain at most $M_S$, best responding only actions in $S$. See the following definition for details.
According to this definition, we can simply use a payoff function in which $n$ is a strictly dominant strategy of the follower as the base function for $\{n\}$.

\begin{definition}[Base function]
A payoff function $\tilde{u}^F$ with BR-correspondence $\wt{\BR}$ is a base function for set $S\subseteq [n]$ if (1)
$\wt{\BR}(\xx) \subseteq S$ for all $\xx \in \Delta_m$, and (2) the leader's SSE payoff in game $(u^L, \tilde{u}^F)$ is $M_S$. 

\end{definition}

\begin{theorem}
\label{thm:find-cover-given-aj}
Suppose that $S \subseteq [n]$ and let $Q = \{ j \in [n]: M_{\{j\}} = M_S \}$.
Moreover, we are given the following elements: $\va_j$ for all $j \in [n] \setminus (S \cup Q)$, a base function $\tilde{u}^F$ for $S$, and $\calM_S$ (as a polytope defined by a set of linear constraints).
Then, in polynomial time, we can either compute a cover of $S$ or decide correctly that $S$ does not admit a cover.
\end{theorem}

\begin{proof}[Proof Sketch]
We construct the following payoff function, 
\begin{boxedtext}
For every $\xx \in \Delta_m$,
\begin{equation}
\label{eq:mu-cover-S}
    \mu(\xx,j) \coloneqq \begin{cases}
    0, & \text{if } j\in S;\\
    \vb_j\cdot\xx- c_j, & \text{otherwise}.
    \end{cases}
\end{equation}
where for each $j \in [n] \setminus S$, $(\vb_j, c_j)$ is a hyperplane such that for any $\xx \in \Delta_m$:
\begin{align}
\label{eq:bbj-cj}
u^L(\xx,j) \ge M_S \quad\Longleftrightarrow\quad \vb_j\cdot \xx \le c_j.
\end{align}
\end{boxedtext}

\Cref{lmm:find-cover-mu} shows that we can efficiently check if $\mu$ is a cover of $S$, and if it is not, then $S$ does not admit any other cover, either. 
Intuitively, $\mu$ aims to bring down the leader's maximum attainable payoff in $\calM_S$ to below $M_S$, so it needs to avoid responses that lead to $u^L(\xx, j) \ge M_S$ when the leader plays $\xx$. \Cref{eq:bbj-cj} ensures this and roughly speaking it creates a ``quasi-zero-sum'' game on actions $j \notin S$. (Ideally, we could just use $\mu = - u^L$ to fulfill this task if we had the full information of $u^L$.)
It then remains to find a way to efficiently compute a set of hyperplanes $(\vb_j, c_j)$ satisfying \Cref{eq:bbj-cj} to finish the construction of $\mu$, which is further demonstrated in the proof of \Cref{lmm:find-cover-bbj-cj} (deferred to \Cref{sec:proof-of-lmm-find-cover-bbj-cj}).
\end{proof}

\begin{restatable}{lemma}{lmmfindcovermu}
\label{lmm:find-cover-mu}
It can be decided in polynomial time whether the payoff matrix $\mu$ defined in \Cref{eq:mu-cover-S} is a cover of $S$ or not.
Moreover, if it is not a cover of $S$, then $S$ does not admit any other cover, either.
\end{restatable}

\begin{restatable}{lemma}{lmmfindcoverbbjcj}
\label{lmm:find-cover-bbj-cj}
A hyperplane $(\vb_j, c_j)$ satisfying \Cref{eq:bbj-cj} can be computed in polynomial time for every $j \in [n] \setminus S$.
\end{restatable}

\section{Learning $J$ and $\xx^*$}
\label{sec:learn-j-xxs}

We show how to learn $J$ and $\xx^*$ in this section.
Recall that as outlined in \Cref{fig:main-approach}, we want to find $J \subseteq [n]$ and $\xx^* \in \Delta_m$, such that $u^L(\xx^*,j) = M_J = M_{[n]}$ for all $j \in J$.
First, we highlight several observations that we will use throughout this section.

\Crefformat{enumi}{#2#1#3}
\crefformat{enumi}{#2#1#3}

\begin{observation}
\label{asn:three-assumptions}
The following assumptions are without loss of generality:
\begin{enumerate}[itemsep=1mm,label=(\alph*),ref=\Cref{asn:three-assumptions}(\alph*)]
\item \label{asn:aj-Mj} 
$\va_j$ is given for all $j \in [n]$, and $M_{\{j\}} > M_{[n]}$.

\item \label{asn:1-min-M}
$1 \in \argmin_{j \in [n]} M_{\{j\}}$.

\item \label{asn:a1-non-uniform}
$a_{1,j} \neq a_{1,k}$ for some $j, k \in [m]$.
\end{enumerate}
\end{observation}

Specifically, \Cref{asn:aj-Mj} results directly from the learning outcome of \Cref{sec:learningaj}. \Cref{asn:1-min-M} holds by relabeling action $1$ and an arbitrary action in $\argmin_{j \in [n]} M_{\{j\}}$. 
In the case where \Cref{asn:a1-non-uniform} does not hold, we have $a_{1,1} = a_{1,2} = \dots = a_{1,m}$; hence, $u^L(\xx, 1)$ is a constant and $u^L(\xx, 1) = M_{\{1\}}$ for all $\xx \in \Delta_m$.
By \Cref{asn:aj-Mj}, it then follows that $u^L(\xx, 1) > M_{[n]}$ for all $\xx \in \Delta_m$.
According to the desired property of $J$ defined in \eqref{eq:J-M}, this means that we can essentially exclude 
all actions that do not satisfy \Cref{asn:a1-non-uniform} from our consideration completely---we consider the residue game defined on the remaining action set of the follower after excluding these actions.

In addition to the assumptions, when $\va_j$ is known, we can also assume access to the following oracle $\AER$ (\Cref{def:ER-oracle}), which determines if an action $j$ is an SSE response. 
Specifically, to determine whether $j$ is an SSE response of game $(u^L, \tilde{u}^F)$, it suffices to check whether $(\xx,j)$ is an SSE (by using $\ASSE$) for an arbitrary $\xx$ in the set $\argmax_{\xx' \in \widetilde{\BR}^{-1}(j)} u^L(\xx', j) = \argmax_{\xx' \in \widetilde{\BR}^{-1}(j)} \va_j \cdot \xx'$, so $\xx$ can be computed efficiently when $\va_j$ is known.

\begin{definition}[Equilibrium response oracle]
\label{def:ER-oracle}
Given a game $\widetilde{\calG} = (u^L, \tilde{u}^F)$ and an action $j \in [n]$ of the follower, the {\em equilibrium response (ER) oracle}, denoted $\AER$, outputs whether $j$ is an SSE response of $\widetilde{\calG}$.
\end{definition}

To learn $J$ and $\xx^*$, we will use the following strategically equivalent payoff matrix $\tilde{u}^L$ of the {\em leader}, as a surrogate for the original matrix $u^L$:
\begin{equation}
\label{eq:tilde-u-L}
\tilde{u}^L(\xx, j) = 
\begin{cases}
\frac{\gamma_j}{\gamma_1} \cdot \va_j \cdot \xx + \frac{\beta_j - \beta_1}{\gamma_1}, & \text{ if } j \in \widehat{J}; \\
W, & \text{ otherwise},
\end{cases}
\end{equation}
where $W =1 + \max_{i\in[m], \ell \in \widehat{J}} \tilde{u}^L(i, \ell)$ (a number that is sufficiently large);
$\gamma_j$ and $\beta_j$ are the parameters that give $u^L(\xx, j) \equiv \gamma_j \cdot \va_j \cdot \xx + \beta_j$; 
and
\begin{equation}
\label{eq:hat-J}
\widehat{J} \coloneqq \left\{j \in [n]: \min_{\xx \in \Delta_m} u^L(\xx,j) < M_{\{1\}} \right\}.
\end{equation}
According to Lemma~\ref{lmm:tilde-uL-J}, we can compute $J$ and $\xx^*$ based on $\tilde{u}^L$ instead of the original unknown matrix $u^L$, where we use the notation
\[
\widetilde{M}_S \coloneqq \max_{\xx \in \Delta_m} \min_{j \in S} \tilde{u}^L(\xx, j)
\]
for all $S \subseteq [n]$, defined analogously to $M_S$ (\Cref{eq:def-M}).
The problem is then equivalent to computing the maximin strategy of $\tilde{u}^L$, which is tractable if all the parameters defining $\tilde{u}^L$ are known.

\begin{restatable}{lemma}{lmmtildeuLJ}
\label{lmm:tilde-uL-J}
For any $J \subseteq [n]$ and $\xx^* \in \Delta_m$, 
\Cref{eq:J-M} holds
if and only if 
$\tilde{u}^L(\xx^*,j) = \widetilde{M}_J = \widetilde{M}_{[n]}$ for all $j \in J$.
\end{restatable}

It remains to compute the parameters needed for constructing $\tilde{u}^L$, that is, $\widehat{J}$, $\gamma_j$, and $\beta_j$.
Indeed, $\widehat{J}$ can be efficiently computed by comparing $u^L(i,j)$ with $M_{\{1\}} = \max_{i\in[n]} u^L(i,1)$ for all $i,j$; these relations are known according to \Cref{asn:relation-jik}.
Also note that $\widehat{J} \neq \emptyset$, since otherwise we would have $M_{[n]} = M_{\{1\}}$, which contradicts \Cref{asn:aj-Mj}.

As for $\gamma_j$ and $\beta_j$, we will not learn them directly, but only learn the quantities $\gamma_j/\gamma_1$ and $(\beta_j-\beta_1)/\gamma_1$. To this end, we will find two different pairs of strategies $\mathbf{x}_1, \mathbf{y}_1$ and $\mathbf{x}_2, \mathbf{y}_2$, such that $u^L(\mathbf{x}_1,1)=u^L(\mathbf{y}_1,j)$ and $u^L(\mathbf{x}_2,1)=u^L(\mathbf{y}_2,j)$.
This gives the following system of linear equations:
\begin{equation}
\label{eq:compute-gamma-beta}
    \begin{cases}
    \gamma_1 \cdot \mathbf{a}_1 \cdot \mathbf{x}_1 + \beta_1  = 
    \gamma_j \cdot \mathbf{a}_j\cdot \mathbf{y}_1 + \beta_j\\
    \gamma_1 \cdot \mathbf{a}_1\cdot \mathbf{x}_2 + \beta_1  = 
    \gamma_j \cdot \mathbf{a}_j\cdot \mathbf{y}_2 + \beta_j
    \end{cases}
\end{equation}
Rearranging the terms, it is easy to see that when $\xx_1$, $\xx_2$, $\yy_1$, and $\yy_2$ are given (in addition to $\va_1$ and $\va_j$, which we already know),
the above is a system of equations about $\gamma_j/\gamma_1$ and $(\beta_j-\beta_1)/\gamma_1$;
hence, solving the system gives the two desired quantities.
We hereafter call each of the two pairs $\mathbf{x}_1, \mathbf{y}_1$ and $\mathbf{x}_2, \mathbf{y}_2$ a {\em reference pair}.
We  in particular look for two reference pairs that make the above equation system non-degenerate, so that it has a unique solution.

As a final remark, the reason why we use $\widehat{J}$ instead of $[n]$ is that $\gamma_j/\gamma_1$ and $(\beta_j-\beta_1)/\gamma_1$ cannot be learned for actions $j \notin \widehat{J}$. Indeed, the leader's payoffs for actions not in $\widehat{J}$ is too high, so these actions will not contribute to making the set $J$ such that $M_J = M_{[n]}$.

An overview of the approach to learning $J$ and $\xx^*$ is provided in \Cref{fig:learning-J-xstar}.

\begin{figure}[t]
\begin{boxedtext*}
\begin{itemize}[leftmargin=*,itemsep=1mm]
\item[1.]
For each $j \in \widehat{J}$,
find two reference pairs $\mathbf{x}_1, \mathbf{y}_1$ and $\mathbf{x}_2, \mathbf{y}_2$, such that $u^L(\mathbf{x}_1,1)=u^L(\mathbf{y}_1,j)$ and $u^L(\mathbf{x}_2,1)=u^L(\mathbf{y}_2,j)$.
\detailsprovidedin{\Cref{sec:first-reference-pair,sec:second-ref}}

\item[2.] 
Solve \Cref{eq:compute-gamma-beta} to obtain the quantities $\gamma_j/\gamma_1$ and $(\beta_j-\beta_1)/\gamma_1$ for each $j \in \widehat{J}$. 

\item[3.]
Construct $\tilde{u}^L$ defined in \Cref{eq:tilde-u-L}.

\item[4.]
Compute the maximin value $\widetilde{M}_{[n]}$ and maximin strategy $\xx^*$ of $\tilde{u}^L$;
Output $J = \{j \in \widehat{J}: \tilde{u}^L(\xx^*, j) = \widetilde{M}_{[n]}\}$ and $\xx^*$.
\detailsprovidedin{\Cref{lmm:tilde-uL-J}}
\end{itemize}
\end{boxedtext*}
\vspace{-3mm}
\caption{Summary of the approach to learning $J$ and $\xx^*$.
\label{fig:learning-J-xstar}}
\end{figure}

\subsection{Finding the First Reference Pair}
\label{sec:first-reference-pair}

For ease of description, in what follows, we assume $2 \in \widehat{J}$ and present how to find the reference pairs for computing $\gamma_2/\gamma_1$ and $(\beta_2-\beta_1)/\gamma_1$.
The results generalize to any $j \in \widehat{J}$ by relabeling $2$ to $j$.  

\begin{observation}
\label{asn:2-hat-J}
Without loss of generality, we can assume that $2 \in \widehat{J}$.
\end{observation}

Next we show how to find a first reference pair.
Recall that we want to find a pair $\mathbf{x}, \mathbf{y}$ with $u^L(\xx,1) = u^L(\yy,2)$.
We in particular aim to find a pair such that
\begin{equation}
\label{eq:first-ref-pair}
u^L(\xx,1) = u^L(\yy,2) = M_{\{1,2\}}. \tag{$\star$}
\end{equation}

\begin{figure}[t]
	\center
	\begin{tikzpicture}
	
	\begin{ternaryaxis}[
	ternary limits relative=false,
	width= 45mm,
 	xtick=\empty,
 	ytick=\empty,
 	ztick=\empty,
	area style,
	grid=none,
	clip=false,
	xmin=0, xmax=1.0,
	ymin=0, ymax=1.0,
	zmin=0, zmax=1.0,
    axis on top,
	]
	
	\addplot3 coordinates {
	    (1, 0, 0)
		(0.32, 0.68, 0)
		(0.45, 0, 0.55)
        (1, 0, 0)
	};
	
	\addplot3 coordinates {
		(0, 1, 0)
		(0.32, 0.68, 0)
		(0.45, 0, 0.55)
		(0, 0, 1)
	};
	
 	\draw[black, thick] (0.3, 0.8, -0.1) -- (0.47, -0.1, 0.63);

	\node at (0.56, 0.26, 0.18) {\large $1$};
	\node at (0.2, 0.3, 0.5) {\large $2$};
	\node [black, rotate=12, anchor=west] at (0.47, -0.1, 0.63) {\small $\va_{2} \cdot \xx = d$};
	
	\end{ternaryaxis}
	
	\end{tikzpicture}
	\caption{The BR-correspondence of $\tilde{u}_d^F$. 
	The triangle represents the strategy space $\Delta_m$ of the leader.
	The regions labeled $1$ and $2$ are $f_d^{-1}(1)$ and $f_d^{-1}(1)$, respectively. \label{fig:brc-fd}}
\end{figure}

We use the following payoff function $\tilde{u}^F_d$ parameterized by a number $d \in \mathbb{R}$. 
\begin{boxedtext}
For every $\xx \in \Delta_m$, let
\begin{equation}
\label{eq:tilde-uF-first-ref-pair}
\tilde{u}^F_d(\xx,j) = \begin{cases}
-1, & \text{if }j \in [n]\setminus \ots;\\
0, & \text{if } j = 1;\\
d-\va_2 \cdot\xx, & \text{if } j = 2. 
\end{cases}
\end{equation}
\end{boxedtext}
Let $f_d$ be the BR-correspondence defined by $\tilde{u}^F_d$. It can be verified that 
\begin{align*}
&f_d^{-1}(1) = \{\xx \in \Delta_m : \va_2 \cdot \xx \ge d \} \\
\text{ and } \quad
&f_d^{-1}(2) = \{\xx \in \Delta_m : \va_2 \cdot \xx \le d \},
\end{align*}
as illustrated in Figure~\ref{fig:brc-fd}.
Let $\calG_d \coloneqq (u^L, \tilde{u}^F_d)$ be the game induced by $\tilde{u}^F_d$.

We then search for a $d$ that makes both actions $1$ and $2$ SSE responses.
Intuitively, when $d$ is sufficiently small, the follower responding according to $f_d$ leads to action $1$ being the sole best response against the entire strategy space $\Delta_m$ of the leader and hence being the only SSE response of $\calG_d$.
Conversely, when $d$ is sufficiently large, action $2$ becomes the only SSE response of $\calG_d$.
Therefore, the goal is to identify a point between these two extremes, where both actions $1$ and $2$ are SSE responses.
This point is exactly 
\begin{equation}
\label{eq:d-star}
d^* \coloneqq ( M_{\{1, 2\}} - \beta_{2}) /\gamma_{2}
\end{equation}
(\Cref{lmm:fd-main}).
Note that, we cannot compute $d^*$ directly using the above equation since we do not know any of the values $M_{\{1, 2\}}$, $\beta_2$, and $\gamma_2$.
Instead, we will use binary search to find it out:
according to \Cref{lmm:fd-main} presented below, $\calG_d$ has different SSE responses when $d < d^*$ and $d> d^*$, so using the SSE oracle, we can identify whether a candidate value is smaller or larger than $d^*$.

\begin{restatable}{lemma}{lmmfdmain}
\label{lmm:fd-main}
Action $1$ is an SSE response of $\calG_d$ if and only if $d \le d^*$, and action $2$ is an SSE response of $\calG_d$ if and only if $d \ge d^*$.
\end{restatable}

Once $d^*$ is identified, we can compute the first reference pair immediately. 
Intuitively, both actions $1$ and $2$ are SSE responses of $\calG_{d^*} = (u^L, \tilde{u}_{d^*}^F)$.
So picking arbitrary $\xx \in \argmax_{\xx' \in f_{d^*}^{-1}(1)} u^L(\xx', 1)$ and $\yy \in \argmax_{\yy' \in f_{d^*}^{-1}(2)} u^L(\yy', 2)$ gives $u^L(\xx,1) = u^L(\yy,2)$, as desired.
The computation of $\xx$ (and similarly $\yy$) can be handled by solving, equivalently, 
$\max_{\xx' \in f_{d^*}^{-1}(1)} \va_1 \cdot \xx'$, where the constraint $\xx' \in f_{d^*}^{-1}(1)$ is further equivalent to $\va_2 \cdot \xx' \ge d^*$.
Hence, the task reduces to solving an LP, which can be done in polynomial time. 
The following theorem presents a stronger result saying that the strategies in this reference pair yields the maximin payoff of $\ots$ for the leader.

\begin{restatable}{theorem}{thmcomputedstar}
\label{thm:compute-d-star}
A reference pair $\xx, \yy \in \Delta_m$ such that $u^L(\xx,1) = u^L(\yy,2) = M_{\{1,2\}}$ can be computed in polynomial time.
\end{restatable}

\subsection{Finding the Second Reference Pair}
\label{sec:second-ref}

Next, we search for a second reference pair. We aim to find a pair $\xx,\yy$ such that 
\begin{equation*}
\label{eq:second-ref-pair}
u^L(\mathbf{x},1)=u^L(\mathbf{y},2) < M_\ots. \tag{$\star\star$}
\end{equation*}
Given \eqref{eq:first-ref-pair}, the requirement that the payoffs are strictly smaller than $M_\ots$ ensures that \eqref{eq:compute-gamma-beta} has a unique solution. 
Similar to our approach to finding the first reference pair, we aim to construct a BR-correspondence to induce two SSEs $(\xx, 1)$ and $(\yy, 2)$. 

As illustrated in \Cref{fig:brc-fd1d2k}, the high-level idea is to pull back the boundaries of the best-response regions $f_{d^*}^{-1}(1)$ and $f_{d^*}^{-1}(2)$ (where we add a new boundary $\va_1 \cdot \xx = d_1$ to the region corresponding to action $1$).
By pulling back the boundaries by appropriate distances, we can keep the leader's maximum attainable payoffs in these two regions equal, while at the same time they become strictly smaller than $M_\ots$; hence, we obtain a pair satisfying \eqref{eq:second-ref-pair}.

Hence, the key is to maintain both actions $1$ and $2$ as SSE responses. The contraction of the best-response regions of these two actions means that a blank region will appear, which needs to be allocated to some response of the follower for the BR-correspondence to be well-defined. 
In particular, we need a follower action that gives the leader a sufficiently low payoff, so that actions $1$ and $2$ remain to be SSE responses. 
Sometimes we cannot find a single action of the follower to fulfill this task so in general we need a cover of $\ots$ which may involve multiple actions.

\begin{figure}[t]
	\center
	\begin{tikzpicture}
	
	\begin{ternaryaxis}[
	ternary limits relative=false,
	width= 45mm,
 	xtick=\empty,
 	ytick=\empty,
 	ztick=\empty,
	area style,
	grid=none,
	clip=false,
	xmin=0, xmax=1.0,
	ymin=0, ymax=1.0,
	zmin=0, zmax=1.0,
    axis on top,
	]
	
	\addplot3 coordinates {
		(0.2, 0.8, 0)
		(0.3, 0.32, 0.38)
		(0.81, 0, 0.19)
        (1, 0, 0)
	};
	
	\addplot3 coordinates {
		(0, 1, 0)
		(0.205, 0.795, 0)
		(0.36, 0, 0.64)
		(0, 0, 1)
	};
	
	\addplot3 coordinates {
        (0.81, 0, 0.19)
        (0.3, 0.32, 0.38)
        (0.36, 0, 0.64)
	};
	
 	\draw[black,thick] (0.18, 0.92, -0.1) -- (0.38, -0.1, 0.72);
	\draw[black, thick] (0.3, 0.32, 0.38) -- (0.97, -0.1, 0.13);

	\node at (0.56, 0.26, 0.18) {\large $1$};
	\node at (0.15, 0.25, 0.6) {\large $2$};
	\node at (0.5, 0.08, 0.42) {$h$};

	\node [rotate=12, anchor=west] at (0.97, -0.1, 0.13) {\small $\va_1 \cdot \xx = d_1$};
	
	\node [black, rotate=12, anchor=west] at (0.38, -0.1, 0.72) {\small $\va_{2} \cdot \xx = d_2$};
	
	\end{ternaryaxis}
	
	\end{tikzpicture}
	\caption{
 $f_{d_1, d_2}$: the BR-correspondence used for finding the second reference pair. The triangle represents the strategy space $\Delta_m$ of the leader. The regions labeled $1$ and $2$ are $f_{d_1, d_2}^{-1}(1) = P_1$ and $f_{d_1, d_2}^{-1}(1) = P_2$, respectively.
	\label{fig:brc-fd1d2k}}
\end{figure}

Our main result is stated in \Cref{thm:second-ref-pair-payoff-query}. To better illustrate the approach, we will first present an algorithm that uses {\em BR-correspondence queries}, queries that use games in the form $\widetilde{\calG} = (u^L, \widetilde{\BR})$, where $\widetilde{\BR}:\Delta_m \to 2^{[n]}$ is a BR-correspondence. We assume temporarily that $\ASSE$ can handle such queries.
Ideally, the BR-correspondences used should also be realized by valid payoff matrices, but this is not always the case with the algorithm presented next.
Hence, we also design a stronger algorithm that always uses payoff matrices to query the SSE oracle. The process is much more involved but can be better understood based on intuition conveyed from the former, and we leave the details to \Cref{sec:second-ref-payoff-query}.

\begin{restatable}{theorem}{thmsecondrefpayoffquery}
\label{thm:second-ref-pair-payoff-query}
A reference pair $\xx, \yy \in \Delta_m$ such that $u^L(\xx, 1) = u^L(\yy, 2) < M_\ots$ can be computed in polynomial time. 
\end{restatable}

\subsubsection{Querying by Using BR-correspondences}
\label{sec:2nd-ref-BR-query}

Define the following partition $P = (P_0, P_1, P_2)$ of $\Delta_m$, which is parameterized by two numbers $d_1$ and $d_2$: 
\begin{align}
&P_0(d_1, d_2) \coloneqq \{\xx \in \Delta_m: \va_2 \cdot \xx \ge d_2 \text{ and } \va_1 \cdot \xx \ge d_1\} \\
&P_1(d_1, d_2) \coloneqq \{\xx \in \Delta_m: \va_2 \cdot \xx \ge d_2 \text{ and } \va_1 \cdot \xx \le d_1\} \\
&P_2(d_1, d_2) \coloneqq \{\xx \in \Delta_m: \va_2 \cdot \xx \le d_2 \}
\end{align}
For ease of description, we will sometimes omit the dependencies on $d_1$ and $d_2$ and just write $P_0$, $P_1$, and $P_2$ when the parameters are clear from the context.
Based on the partition, we define the following BR-correspondence $f_{d_1,d_2}$, which generalizes the BR-correspondence $f_d$ defined in the previous section (see \eqref{eq:tilde-uF-first-ref-pair}).
\begin{boxedtext}
Let $f_{d_1,d_2}: \Delta_m \to 2^{[n]}$ be a BR-correspondence such that for every $\xx \in \Delta_m$:
\begin{equation}
\label{eq:brc-f-d1-d2-h}
    f_{d_1, d_2} (\xx) \supseteq 
    \begin{cases}
    \{1\} & \text{if and only if} \quad \xx \in P_1 \\
    \{2\} & \text{if and only if} \quad \xx \in P_2 \\
    h(\xx) & \text{if and only if} \quad \xx \in P_0
    \end{cases}
\end{equation}
where $h: \Delta_m \to 2^{[n] \setminus \ots}$ is the BR-correspondence of an arbitrary (proper) cover of $\{1,2\}$.
\end{boxedtext}
Hence, according to this construction, we have
\begin{align*}
&f_{d_1,d_2}^{-1}(j) = P_j \quad \text{for all } j \in \ots, \\
\text{and } \quad
&f_{d_1,d_2}^{-1}(k) \subseteq P_0 \quad \text{for all } k \in [n] \setminus \ots.
\end{align*}
Moreover, compared with $f_d$, we have $P_2(d_1, d_2) = f_{d_2}^{-1}(2)$, and $P_1(d_1, d_2) \cup P_0(d_1, d_2) = f_{d_2}^{-1}(1)$.
\Cref{fig:brc-fd1d2k} illustrates the structure of $f_{d_1, d_2}$.

Our goal is to find two numbers $d_1$ and $d_2$
such that both actions $1$ and $2$ are SSE responses of $\mathcal{G}_{d_1,d_2} \coloneqq (u^L, f_{d_1,d_2})$. 
Before we present our algorithm for computing $d_1$ and $d_2$, we define several other useful notions and make a few observations about $f_{d_1, d_2}$.

We define
\begin{equation}
\label{eq:V-d1-d2}
V_{d_1,d_2}^L(j) \coloneqq \max_{\xx \in f_{d_1,d_2}^{-1}(j)} u^L(\xx, j).
\end{equation}
Moreover, let 
\begin{equation}
\label{eq:dj-star}
d_j^* \coloneqq (M_\ots - \beta_j) / \gamma_j,
\end{equation}
for $j \in \ots$ (i.e., $d_2^* = d^*$ in \eqref{eq:d-star}).
It can be verified that when $d_1 = d_1^*$ and $d_2 = d_2^*$, $P$ coincides with $f_{d^*}$.
Moreover, we have: 
\begin{itemize}
\item 
$P_2(d_1^*, d_2^*) = f_{d^*}^{-1}(2) \neq \emptyset$, $P_1(d_1^*, d_2^*) = f_{d^*}^{-1}(1) \neq \emptyset$, and $P_0(d_1^*, d_2^*) = \calM_{\ots}$; and

\item 
for all $\odot \in \{>, <, =\}$ and $j \in \ots$,
$u^L(\xx, j) \odot M_\ots$ if and only if $\va_j \cdot \xx \odot d_j^*$.
\end{itemize}
These two values $d_1^*$ and $d_2^*$ will be useful for our algorithm.
Indeed, we search for the two numbers $d_1$ and $d_2$ we aim to find in the domains $(-\infty, d_1^*)$ and $(-\infty, d_2^*)$, respectively.
We require that $d_1 < d_1^*$ and $d_2 < d_2^*$ in order to avoid the trivial solution $d_1 = d_1^*$ and $d_2 = d_2^*$, which would yield the same pair as the first reference pair we obtained. 

\paragraph{Existence and Computation of $h$}

To obtain $f_{d_1,d_2}$ requires computing a cover $h$ of $\ots$.
We directly invoke \Cref{thm:find-cover-given-aj} presented earlier to accomplish this task.
Now that $\va_j$ is known for all $j$, to apply this theorem, we only need to supply with a base function for $\ots$, as well as $\calM_\ots$.
Indeed, the function $\tilde{u}_{d^*}^F$ we obtained in the previous section and the region $P_0(d_1^*, d_2^*)$ defined above fulfills these demands.
In the special case where $\ots$ does not admit a cover, \Cref{lmm:ots-no-cover} in the appendix demonstrates that a pair of satisfying $J$ and $\xx^*$ can be obtained without a second reference pair.

\paragraph{Pinning Down $d_1$ and $d_2$}

To find $d_1$ and $d_2$, we first restrict our search in a one-dimensional space: we aim to find a value $\epsilon > 0$ 
such that at least one of actions $1$ and $2$ are SSE responses of $\mathcal{G}_{d_1^* - \epsilon, d_2^* - \epsilon}$.
We present the following lemma.

\begin{restatable}{lemma}{lmmGdeps}
\label{lmm:G-d-eps}
There exists $\hat{\epsilon} > 0$, such that for all $d_1 \in [d_1^* - \hat{\epsilon}, d_1^*]$ and $d_2 \in [d_2^* - \hat{\epsilon}, d_2^*]$: 
\begin{itemize}
\item[(i)] 
at least one of $j \in \ots$ is an SSE response of $\mathcal{G}_{d_1, d_2}$; and

\item[(ii)] 
$V_{d_1,d_2}^L(j) = \gamma_j \cdot d_j + \beta_j$ for $j\in \ots$.
\end{itemize}
Moreover, assuming $\ASSE$ can handle BR-correspondence queries, $\hat{\epsilon}$ is computable in polynomial time.
\end{restatable}

The proof of \Cref{lmm:G-d-eps} is deferred to \Cref{sec:proof-lmm:G-d-eps}.
Using \Cref{lmm:G-d-eps}, we obtain $d_1 = d_1^* - \hat{\epsilon}$ and $d_2 = d_2^* - \hat{\epsilon}$, such that $\calG_{d_1, d_2}$ has at least one of actions $1$ and $2$ as an SSE response.
If it happens that both actions $1$ and $2$ are SSE responses, we are done with a second reference pair.
Otherwise, e.g., suppose that action $2$ is not an SSE response, it means that we need to expand $P_2$ to increase the leader's maximum attainable payoff for strategies in this region; we do this simply by increasing $d_2$, and we search for a value of $d_2$ that makes both actions $1$ and $2$ an SSE response.
This leads to \Cref{thm:second-ref-pair-BR-correspondence} as a weaker version of \Cref{thm:second-ref-pair-payoff-query}.
The detailed proofs of \Cref{thm:second-ref-pair-BR-correspondence,thm:second-ref-pair-payoff-query} can be found in the appendix.

\begin{restatable}{theorem}{thmsecondrefpairBRcorrespondence}
\label{thm:second-ref-pair-BR-correspondence}

Assume that $\ASSE$ can handle BR-correspondence queries.
A reference pair $\xx, \yy \in \Delta_m$ such that $u^L(\xx, 1) = u^L(\yy, 2) < M_\ots$ can be computed in polynomial time. 
\end{restatable}

\section{Conclusion}
\label{sec:conclusion}

Whereas the polynomial-time tractability of the follower's optimal manipulation in \citet{birmpas2021optimally} assumes full information on the leader's payoffs, 
we showed that such strict information advantage is not quite necessary: in the learning setting, polynomial number of information exchanges with the leader suffices for the follower to find out the optimal payoff function to report. 

Our result reveals an interesting difference in the effects of information asymmetry on the two players in a Stackelberg game and points out potential risks of applying strategy commitment under a lack of information: while a leader needs an exponential number of queries to learn a follower's best response correspondence in the worst case~\citep{peng2019learning}, it only takes the follower a polynomial number of queries to collect enough payoff information of the leader and achieve optimal deception. Thus, special care is needed, when one wants to apply strategy commitment to gain advantage in games but without sufficient information. 

We note that despite of the SSE oracle defined in our paper, there can be other definitions of oracles, possibly weaker than ours. For example, it may be worth studying  oracles that only take a fake payoff function as input and output an SSE according to some underlying tie-breaking rules. This might be more realistic as it is the leader who chooses which equilibrium strategy to play. However, our results do not directly extend to this setting. One insight into this is that, under the weaker oracle, we may lose the ability to check whether the leader utilities under two strategy profiles are equal. Hence, we may not be able to learn the precise quantities necessary for our algorithms. We leave it as an interesting open problem whether there can be an optimal, or a near-optimal algorithm for the follower to learn to deceive when they have access to such weaker oracles.

\bibliographystyle{plainnat}

\clearpage

\appendix

\section{Proof of \NoCaseChange{\Cref{thm:equ-game}}}
\label{sec:proof-of-thm:equ-game}

\thmequgame*

\begin{proof}
Given \Cref{lmm:opt-to-inducibility}, LP~\eqref{lp:maximin-inducibility} reduces to  deciding the inducibility of any given strategy profile $(\yy, k)$.
We demonstrate that this inducibility problem can be solved in polynomial time.
Specifically, we will argue that the following algorithm produces an {\em inducibility witness} $\tilde{u}^F$ for $(\yy, k)$ in polynomial time:
that is, $(\yy, k)$ is inducible (with respect to $u^L$) if and only if it is an SSE in $(u^L, \tilde{u}^F)$.
Hence, by querying $\ASSE$ to check whether $(\yy, k)$ is an SSE of $(u^L, \tilde{u}^F)$, we can decide the inducibility of  $(\yy, k)$. The stated result then follows.

\begin{boxedtext}
Construct an inducibility witness of $(\yy, k)$:
\begin{itemize}[leftmargin=*]
\item[1.] 
For every $j$ such that $M_{\{j\}} = M_{[n]}$, since $\va_j$ is not given, let $\va_j$ be a vector in $\mathbb{R}^{m}$ such that: $a_{ji} = 1$ if $i \in \calM_{\{j\}}$ and $a_{ji} = 0$ otherwise.

\item[2.] 
Construct a payoff matrix $\tilde{u}^L \in \mathbb{R}^{m \times n}$ corresponding to the following payoff function:
\begin{equation}
\tilde{u}^L(\xx, j) = 
\begin{cases}
\va_j \cdot \xx - \va_j \cdot \xx^*, & \text{ if } j \in J \setminus\{k\}; \\
\va_k \cdot \xx - \va_k \cdot \yy, & \text{ if } j = k; \\
1, & \text{ otherwise. }
\end{cases}
\tag{\ref{eq:equ-uL}}
\end{equation}

\item[3.] 

Decide if $(\yy, k)$ is inducible with respect to $\tilde{u}^L$:
\begin{itemize}
\item 
If it is inducible, compute a payoff matrix $\tilde{u}^F \in \mathbb{R}^{m \times n}$ such that $(\yy, k)$ is an SSE in $(\tilde{u}^L, \tilde{u}^F)$;

\item
Otherwise, pick an arbitrary $\tilde{u}^F \in \mathbb{R}^{m \times n}$.
\end{itemize}

\item[4.]
Output $\tilde{u}^F$.
\end{itemize}
\end{boxedtext}

Namely, the above algorithm constructs a ``surrogate'' matrix $\tilde{u}^L$ using the given information $\va_j$, $J$, and $\xx^*$.
Then a witness $\tilde{u}^F$ is produced using $\tilde{u}^L$. 
The polynomial run-time of the algorithm is readily seen. In particular, Step~3 can be done in polynomial time according to \Cref{thm:maximin-characterization}. (Note that all parameters of $\tilde{u}^L$ is known, so this is a full-information setting.) 
We next prove that $(\yy, k)$ is inducible if and only if it is an SSE in $(u^L, \tilde{u}^F)$, where $\tilde{u}^F$ is the matrix produced by the above algorithm.

Indeed, the ``if'' direction is trivial, so we prove that if $(\yy, k)$ is inducible then it must be an SSE in $(u^L, \tilde{u}^F)$.
Suppose that $(\yy, k)$ is inducible in the sequel.
According to \Cref{thm:maximin-characterization}, it must be that $u^L(\yy, k) \ge M_{[n]}$
in this case.

In what follows, we first argue that the following statements hold for any $\xx \in \Delta_m$ (\Cref{clm:chain-tilde-uL}).
\begin{align}
\forall j \in J \setminus \{k\}: \quad & \tilde{u}^L(\xx, j) > 0 
\quad\Longleftrightarrow\quad \va_j \cdot \xx > \va_j \cdot \xx^*
\quad\Longleftrightarrow\quad u^L(\xx, j) > u^L(\xx^*, j); \label{eq:chain-jnotinJmk}\\
\text{ and } \quad & \tilde{u}^L(\xx, k) > 0
\quad\Longleftrightarrow\quad \va_k \cdot \xx > \va_k \cdot \yy
~\quad\Longleftrightarrow\quad u^L(\xx, k) > u^L(\yy, k). \label{eq:chain-jeqk}
\end{align}
With this result, we can then prove \Cref{clm:thm:equ-game:claim-1} and \Cref{clm:thm:equ-game:claim-2} to complete the proof: \Cref{clm:thm:equ-game:claim-1} implies that Step~3 of the above algorithm outputs a $\tilde{u}^F$ that makes $(\yy, k)$ an SSE of $(\tilde{u}^L, \tilde{u}^F)$, instead of an arbitrary matrix.
Given this, \Cref{clm:thm:equ-game:claim-2} further confirms that $(\yy, k)$ is also an SSE of $(u^L, \tilde{u}^F)$.

\begin{claim}
\label{clm:chain-tilde-uL}
\eqref{eq:chain-jnotinJmk} and \eqref{eq:chain-jeqk} hold for any $\xx \in \Delta_m$.
\end{claim}

\begin{proof}
Note that $\tilde{u}^L(\xx, j) > 0 \Longleftrightarrow \va_j \cdot \xx > \va_j \cdot \xx^*$ and $\tilde{u}^L(\xx, k) > 0
\Longleftrightarrow \va_k \cdot \xx > \va_k \cdot \yy$ follow directly by \eqref{eq:equ-uL}, so it remains to prove the second part of each statement. Moreover, if $M_{\{j\}} \neq M_{[n]}$ and $M_{\{k\}} \neq M_{[n]}$, $\va_j$ and $\va_k$ satisfy \eqref{eq:aj}, so it is easy to see that the second part of each statement also holds by expanding $u^L$ according to \eqref{eq:aj}.

Now consider the case where $M_{\{j\}} = M_{[n]}$.
We have 
\[u^L(\xx^*, j) = M_{[n]} = M_{\{j\}} \ge u^L(\xx, j),\]
where $u^L(\xx^*, j) = M_{[n]}$ follows by the definition of $\xx^*$ in \eqref{eq:J-M} (and the assumption that $j \in J$), and $M_{\{j\}} \ge u^L(\xx, j)$ follows by the definition that $M_{\{j\}} = \max_{\xx' \in \Delta_m} u^L(\xx', j)$.
At the same time, $u^L(\xx^*, j) = M_{\{j\}}$ also implies that $\xx^* \in \calM_{\{j\}}$.
Hence, according to the way $\va_j$ is defined in Step~1 of the above algorithm, we have $\xx^* \in \argmax_{\xx' \in \Delta_m} \va_j \cdot \xx' \ge \va_j \cdot \xx$.
Consequently, both $\va_j \cdot \xx > \va_j \cdot \xx^*$ and $u^L(\xx, j) > u^L(\xx^*, j)$ are always false in this case, so \eqref{eq:chain-jnotinJmk} holds.

Finally, consider the case where $M_{\{k\}} = M_{[n]}$.
In this case, we can establish
$M_{\{k\}} \ge u^L(\yy, k) \ge M_{[n]}$. 
Consequently, $u^L(\yy, k) \ge M_{[n]}$, so we have $\yy \in \calM_{\{k\}}$. Similarly to the above case where $M_{\{j\}} = M_{[n]}$, we can show that both $\va_k \cdot \xx > \va_k \cdot \yy$ and $u^L(\xx, k) > u^L(\yy, k)$ must always be false in this case (i.e., by putting $\yy$ in place of $\xx^*$ and $k$ in place of $j$). 
Hence, \eqref{eq:chain-jeqk} holds.
\end{proof}

\begin{claim}
\label{clm:thm:equ-game:claim-1}
$(\yy, k)$ is inducible with respect to $\tilde{u}^L$ (assuming that $(\yy, k)$ is inducible with respect to $u^L$).
\end{claim}

\begin{proof}
By construction, $\tilde{u}^L(\yy, k) = 0$. 
Hence, by \Cref{thm:maximin-characterization}, it suffices to show that the maximin value of $\tilde{u}^L$ is at most $0$, i.e., $\widetilde{M}_{[n]} \coloneqq \max_{\xx \in \Delta_m} \min_{j \in [n]} \tilde{u}^L(\xx, j) \le 0$.

Suppose for the sake of contradiction that $\widetilde{M}_{[n]} > 0$. By definition, this means that there exists $\xx \in \Delta_m$ such that $\tilde{u}^L(\xx, j) > 0$ for all $j \in [n]$.
Applying \eqref{eq:chain-jnotinJmk} and \eqref{eq:chain-jeqk} gives 
$u^L(\xx, j) > u^L(\xx^*, j) = M_{[n]}$ for all $j \in J \setminus \{k\}$ (where $u^L(\xx^*, j) = M_{[n]}$ according to \eqref{eq:J-M}), and $u^L(\xx, k) > u^L(\yy, k) \ge M_{[n]}$.
It then follows that $\min_{j \in J} u^L(\xx, j) > M_{[n]}$, which contradicts the condition in the definition of $J$, i.e., $M_J = M_{[n]}$.
\end{proof}

\begin{claim}
\label{clm:thm:equ-game:claim-2}
$(\yy, k)$ is an SSE of $(u^L, \tilde{u}^F)$ (assuming that it is an SSE of $(\tilde{u}^L, \tilde{u}^F)$).
\end{claim}

\begin{proof}
Suppose that $(\yy, k)$ is an SSE in $(\tilde{u}^L, \tilde{u}^F)$.
Let $\widetilde{\BR}$ be the BR-correspondence of $\tilde{u}^F$.
Pick arbitrary $\xx \in \Delta_m$ and $j \in \widetilde{\BR}(\xx)$.
Since $(\yy, k)$ is an SSE in $(\tilde{u}^L, \tilde{u}^F)$, by definition, we have $\tilde{u}^L(\xx, j) \le \tilde{u}^L(\yy, k)$.
Substituting the right hand side with \eqref{eq:equ-uL}, we get that $\tilde{u}^L(\xx, j) \le \va_k \cdot \yy - \va_k \cdot \yy = 0$. According to \eqref{eq:equ-uL}, this means that $j \in J \cup \{k\}$.

If $j \in J \setminus \{k\}$, then applying \eqref{eq:chain-jnotinJmk}, we get that
\[
u^L(\xx, j) \le u^L(\xx^*, j) = M_{[n]} \le u^L(\yy, k).
\]
(Recall that $u^L(\yy, k) \ge M_{[n]}$ as $(\yy, k)$ is inducible).
If $j=k$, then applying \eqref{eq:chain-jeqk}, we get that
\[
u^L(\xx, j) = u^L(\xx, k) \le u^L(\yy, k).
\]
Hence, $u^L(\xx, j) \le u^L(\yy, k)$ holds in both cases. Since the choice of $\xx$ and $j$ is arbitrary, we have
\[
\max_{\xx \in \Delta_m} \max_{j \in \widetilde{\BR}(\xx)} u^L(\xx, j) \le u^L(\yy, k).
\]

It remains to show that $k \in \widetilde{\BR}(\yy)$. Indeed, according to \Cref{clm:thm:equ-game:claim-1} and Step~3 of the above algorithm, $(\yy, k)$ must be an SSE in $(\tilde{u}^L, \tilde{u}^F)$, which means that $k \in \widetilde{\BR}(\yy)$. 
\end{proof}

The proof is complete.
\end{proof}

\section{Omitted Proofs in \NoCaseChange{\Cref{sec:learningaj}}}

\subsection{Proof of \NoCaseChange{\Cref{lmm:proper-cover}}}

\lmmpropercover*

\begin{proof}
Suppose that $\mu$ is a cover of $S$.
Consider the following matrix $\mu'$:
\begin{align*}
\mu'(i,j)=
\begin{cases}
\mu(i,j) & \text{ if } j \in [n] \setminus S \\
W & \text{ if } j \in S
\end{cases}
\end{align*}
where $W \coloneqq \min_{i,j \in [m] \times [n] \setminus S} \mu(i,j) - 1$.
We argue that $\mu'$ is a proper cover of $S$. Indeed, it is proper since by the above construction no action in $S$ can be a best response to any $\xx \in \Delta_m$. It suffices to prove that $\mu'$ is a cover.

Let $f$ and $f'$ be the BR-correspondences of $\mu$ and $\mu'$, respectively.
Note that for all $\xx \in \calM_S$, it must be that $S \cap f(\xx) = \emptyset$.
Indeed, if $k \in S \cap f(\yy)$ and $\yy \in \calM_S$, we would then have 
\[
\max_{\xx \in \calM_S} \max_{j \in f(\xx)} u^L(\xx, j)
\ge \max_{j \in f(\yy)} u^L(\yy, j)
\ge u^L(\yy, k)
\ge \min_{j \in S} u^L(\yy, j)
= M_S,
\]
contradicting the assumption that $\mu$ is a cover.

By construction, $f'(\xx) = f(\xx)$ if $S \cap f(\xx) = \emptyset$. Hence, $f'(\xx) = f(\xx)$ for all $\xx \in \calM_S$. It follows that
\[
\max_{\xx \in \calM_S} \max_{j \in f'(\xx)} u^L(\xx, j)
= \max_{\xx \in \calM_S} \max_{j \in f(\xx)} u^L(\xx, j) < M_S.
\]
Hence, $\mu'$ is a cover of $S$.
\end{proof}

\subsection{Proof of \Cref{lmm:cover-iff}}

\lmmcoveriff*

\begin{proof}

We first prove the necessity. Suppose a cover $\mu$ of $S$ exists. According to the definition of a cover, for any $\xx \in \calM_{S}$, there exists $j \in [n]\setminus S$, such that $u^L(\xx,j) < M_S$. Thus, for all $\xx \in \calM_S$, we have
\[
\min_{k \in [n]}u^L(\xx,k) \leq u^L(\xx,j)< M_S.
\]
Meanwhile, for all $\xx \notin \calM_S$, 
\[
\min_{k \in [n]}u^L(\xx,k) \leq \min_{k \in S}u^L(\xx,k) < M_S.
\]
Thus, $M_{[n]} = \max_{\xx \in \Delta_m}\min_{j\in [n]}u^L(\xx,j)<M_S$.

Next, consider the sufficiency. Suppose that $M_S > M_{[n]}$. We show that a function $\mu$ such that $\mu(i,j)=-u^L(i,j)$ is a cover of $S$. 
Let $f$ be the BR-correspondence of $\mu$.
By construction, $\argmax_{k \in [n]} \mu(\xx, k) = \argmin_{k \in [n]} u^L(\xx, k)$ for any $\xx \in \Delta_m$, so $j \in f(\xx)$ implies that $u^L(\xx,j) = \min_{k \in [n]}u^L(\xx,k)$.
It follows that
\[
u^L(\xx,j) = \min_{k \in [n]}u^L(\xx,k) \le M_{[n]} < M_S
\]
for all $\xx \in \Delta_m$.
Thus, 
$\max_{\xx \in \calM_S} \min_{j \in f(\xx)}u^L(\xx,j) < M_S$, 
which completes the proof. 
\end{proof}

\subsection{Formal Proof of \Cref{thm:xxi}}

We first present the proof of \Cref{lemma:boundary-point}, then the proof is completed by \Cref{lmm:Ig-empty} and \Cref{lmm:Ig-nonempty}, which solves the cases of $I_\gg = \emptyset$ and $I_\gg\ne \emptyset$ respectively. 

\lemmaboundarypoint*

\begin{proof}
We prove that every vertex $\bv$ of $f_\gg^{-1}(n)$ lies in some $\Gamma_i$, $i \in [m_1 -1]$; in other words, $\bv \in \Gamma_i \cap f_\gg^{-1}(n)$.
Recall that $\xx^i \in \argmin_{\xx \in \Gamma_i \cap f_\gg^{-1}(n)} x_i$.
Hence, $\bv \in \Gamma_i \cap f_\gg^{-1}(n)$ implies that $x_i^i \ge v_i$, and in turn, $u^L(\xx, n) \le u^L(\xx^i, n)$ according to \Cref{lmm:minxi-maxuL}. Since $\max_{\xx \in f_\gg^{-1}(n)} u^L(\xx, n)$ is always attained at some vertex of $f_\gg^{-1}(n)$, the stated result then follows.

Pick an arbitrary vertex $\bv$ of $f_\gg^{-1}(n)$.
Since $\bv \in \mathbb{R}^m$, it lies in $m$ hyperplanes of $f_\gg^{-1}(n)$ and is uniquely defined by them.
By definition, $f_\gg^{-1}(n)$ is defined by the following linear constraints, each corresponding to a hyperplane:
\begin{subequations}
\label{eq:fg-hyperplane}
\begin{align}
& x_1 + \dots + x_m = 1 \label{eq:fg-hyperplane-1} \\
& x_i \ge 0 && \text{ for all } i \in [m] \label{eq:fg-hyperplane-2} \\
& \tilde{u}^F_{\gg}(\xx, n) \geq \tilde{u}^F_{\gg}(\xx, j) && \text{ for all } j \in [n-1] \label{eq:fg-hyperplane-3}
\end{align}
\end{subequations}
Hence, without loss of generality, we can assume that the $m$ hyperplanes defining $\bv$ corresponds to the following coefficient matrix $A$:
\newcommand\leftbrace[2]{%
#1\left\{\vphantom{\begin{array}{c} #2 \end{array}}\right.}
\newcommand\rightbrace[2]{%
\left.\vphantom{\begin{array}{c} #2 \end{array}}\right\}#1}
\begin{equation*}
\small
\renewcommand{\arraystretch}{1.3}
\begin{array}{r}
\vphantom{1}\\
\leftbrace{\ell_1}{1 \\ 1 \\ 1} \\
\\[-2mm]
\leftbrace{\ell_2}{1 \\ 1 \\ 1} \\
\\[-2mm]
\leftbrace{k}{1 \\ 1 \\ 1}
\end{array}
\left(
\begin{array}{ccccc:ccccc}
1 & & \cdots & & 1 & 1 & & \cdots & & 1 \\[1mm]
  \hdashline
1 & & & & & & & & & \\
  & \ddots & & & & & & & & \\
  & & 1 & & & & & & & \\[1mm]
  \hdashline
  & & & & &1 & & & & \\
  & & & &  & & \ddots & & & \\
  & & & & & & & 1& & \\[1mm]
  \hdashline
g_1 & & \cdots & & g_{m_1-1} & \tilde{\mu}(m_1, 1) & & \cdots & & \tilde{\mu}(m, 1) \\
\vdots & & & & \vdots & \vdots & & & & \vdots \\
g_1 & & \cdots & & g_{m_1-1} & \tilde{\mu}(m_1, k) & & \cdots & & \tilde{\mu}(m, k) 
\end{array}
\right)
\begin{array}{l}
\rightbrace{\eqref{eq:fg-hyperplane-1}}{1} \\
\\[-3mm]
\rightbrace{\eqref{eq:fg-hyperplane-2}}{1 \\ 1 \\ 1 \\ \\ 1 \\ 1 \\ 1} \\
\\[-3mm]
\rightbrace{\eqref{eq:fg-hyperplane-3}}{1 \\ 1 \\ 1}
\end{array}
\end{equation*}
where we let $\tilde{\mu}(i,j) = W - \mu(i, j)$ for all $i,j$.
Namely, $A \cdot \bv = \vb$, where $\vb$ is the vector of the constant terms in \eqref{eq:fg-hyperplane}.
In what follows, we let $A_{i:j}$ denote the submatrix formed by the $i$-th to the $j$-th rows of $A$; and let $\overline{A}_{i:j}$ denote the submatrix formed by the $i$-th to the $j$-th columns.
Since $\bv$ is uniquely defined by $A$, we have $\rankp{A} = m$.

We next argue that $\ell_1 \ge m_1 - 2$ to complete the proof. 
This indicates that for at least $m_1 - 2$ actions $i \in [m_1 - 1]$, we have that $v_i = 0$. In other words, $v_i > 0$ for at most one $i \in [m_1 - 1]$, so by definition this means that $\bv \in \Gamma_i$ for some $i \in [m_1 - 1]$.

Suppose for the sake of contradiction that $\ell_1 \le m_1 - 3$.
Hence, $\ell_2 + k = m - \ell_1 - 1 \ge m - m_1 + 2$.
Let 
\[
U = 
\begin{pmatrix}
\tilde{\mu}(m_1 + \ell_2, 1) & \cdots & \tilde{\mu}(m, 1) \\
\vdots & & \vdots \\
\tilde{\mu}(m_1 + \ell_2, k) & \cdots & \tilde{\mu}(m, k) 
\end{pmatrix}
\]
be a submatrix of $A$, and let 
$\bv' = (v_{m_1 + \ell_2 - 1}, \dots, v_m)$.
Hence, $U$ is a $k$-by-$k'$ matrix, with $k' \le k-1$.

We have $U \cdot \bv' = \mathbf{1} \cdot \lambda$, where $\lambda = \sum_{i=1}^{m_1 -1} g_i \cdot v_i$.
Note that $\bv'$ is a size-$k'$ vector, so it must be that $\rankp{U \mid \mathbf{1}^\top} \le k' \le k - 1$ (otherwise, the system of linear equations $U \cdot \xx = \mathbf{1} \cdot \lambda$ would have no solution). 
Note that via linear transformation, the submatrix $A_{m-k+1, m}$ can be transformed into 
$\begin{pmatrix}
O & \mathbf{1}^\top & U
\end{pmatrix}$,
where $O$ denotes a $k$-by-$(m_1 + \ell_2-2)$ matrix with all entries being $0$.
Hence, we get that
\[
\rankp{A_{m-k+1, m}} = \rankp{O\quad \mathbf{1}^\top \quad U} \le k - 1.
\]
Consequently, 
\[
\rankp{A} \le \rankp{A_{1:m-k}} + \rankp{A_{m-k+1:m}} \le (m - k) + (k - 1) < m,
\]
which contradicts the fact that $\rankp{A} = m$. 
\end{proof}

\begin{restatable}{lemma}{lmmIgempty}
\label{lmm:Ig-empty}
There exists $N > 0$, such that $I_\gg \neq \emptyset$ if $g_i \ge N$ for all $i \in [m_1 - 1]$.
Moreover, $N$ can be computed in polynomial time.
\end{restatable}

\begin{proof}
Given \Cref{lemma:boundary-point}, it suffices to find a value $N$ and prove that action $n$ is an SSE response of game $(u^L, \tilde{u}_\gg^F)$ if $g_i \ge N$ for all $i \in [m_1 - 1]$.

Let 
$D =\max_{i,j,i',j'} |u^L(i,j)-u^L(i',j')|$.
For every $\xx\in\Delta_m$, we have
\begin{align*}
u^L(\xx, n) 
= \sum_{i=1}^{m_1-1} x_i\cdot u^L(i,n) +\left(1-\sum_{i=1}^{m_1-1}x_i\right)\cdot M_{\{n\}}
\ge M_{\{n\}}-\left(\sum_{i=1}^{m_1-1}x_i\right)\cdot D,
\end{align*}
where we used the fact that $u^L(i,n) \ge M_{\{n\}}-D$ for all $i$.

Moreover, let 
\[
U=\max_{i,j} u^L(i,j)
\quad \text{ and } \quad
T = \max_{\xx\in\calM_{\{n\}}, j\in h (\xx)}u^L(\xx,j),
\]
where $h$ is the BR-correspondence of $\mu$.
For every $j\in[n-1]$,
\begin{align*}
u^L(\xx, j) 
=& \sum_{i=1}^{m_1-1} x_i\cdot u^L(i,j) +\sum_{i=m_1}^{m}x_i\cdot u^L(i,j)\\
\leq& \left(\sum_{i=1}^{m_1-1}x_i\right)\cdot U+\left(1-\sum_{i=1}^{m_1-1}x_i\right)\cdot T
= \left(\sum_{i=1}^{m_1-1}x_i\right)\cdot \left(U- T \right) + T.
\end{align*}

Recall that since $\mu$ is a cover of $\{n\}$, by definition $T < M_{\{n\}}$. 
So if it holds that $f_\gg(\xx) = \{n\}$ for every $\xx$ such that 
\[
c\coloneqq \frac{M_{\{n\}}-T}{D+U-T} < \sum_{i=1}^{m_1-1}x_i,
\] 
we know that $n$ must be an SSE response of $(u^L,\tilde{u}^F_{\gg})$;
indeed, for all $\yy$ such that $\sum_{i=1}^{m_1-1} y_i \le c$,  this leads to
\[
u^L(\yy, j) \le T + c \cdot (U - T) = M_{\{n\}} - c \cdot D \le \max_{\xx \in f_\gg^{-1}(n)} u^L(\xx, n)
\]
for all $j \in [n-1]$.
It suffices to let 
\[
g_i \geq \overline{N} \coloneqq \frac{1-c}{c}\cdot \left( \max_{i,j} \mu(i,j) -W \right)
\] 
for every $i\in[m_1-1]$ to ensure this, so any $N \ge \overline{N}$ satisfies the condition in the statement of this lemma.

To compute $N$, we can start from an arbitrary value $N > 0$ and $\gg = (N, \dots, N)$.
We query oracle $\ASSE$ to check if $n$ is an SSE response of $(u^L, \tilde{u}_\gg^F)$. Set $N \leftarrow 2 \cdot N$ and $\gg \leftarrow (N, \dots, N)$, and repeat the step if it is not.
This way, we can find a value $N \ge \overline{N}$ in time $\log(\overline{N})$.
\end{proof}

\begin{restatable}{lemma}{lmmIgnonempty}
\label{lmm:Ig-nonempty}
Suppose that $I_\gg \neq \emptyset$ and $i \in [m_1 - 1] \setminus I_\gg$.
There exists $g^* > 0$, such that $I_{\gg'} = I_\gg \cup \{i\}$, where $\gg' = (g_1, \dots, g_{i-1}, g^*, g_{i+1}, \dots, g_{m_1 -1})$.
Moreover, $g'$ can be computed in polynomial time.
\end{restatable}

\begin{proof}[Proof sketch]
Let
\begin{equation}
\label{eq:g-star}
{g^*} = \frac{u^* - (1-y_i^*) \cdot W}{y_i^*},
\end{equation}
where $u^*$ and $\yy^*$ are the optimal solution to the following LP, where $V \coloneqq \max_{j \in [n]} \max_{\xx \in f_\gg^{-1} (j)} u^L(\xx, j)$ is the leader's SSE payoff in $(u^L, \tilde{u}_\gg^F)$. 
\begin{subequations}
\label{lp:xi-bar}
\begin{align}
\min_{u, \yy \in \Delta_m} \quad 
& u \tag{\ref{lp:xi-bar}}\\
\text{subject to} \quad
& \yy \in \Gamma_i \label{lp:xi-bar-1}\\
& \tilde{u}^F_\gg(\yy, j) \le u && \text{ for all } j \in [n-1] \label{lp:xi-bar-2}\\
& u^L(\yy, n) = V \label{lp:xi-bar-3}
\end{align}
\end{subequations}

We demonstrate the following results to complete the proof; the omitted details can be found in the appendix. 
\begin{itemize}
\item $g^*$ is well-defined (\Cref{lmm:lp:xi-bar-feasible}), i.e., LP \eqref{lp:xi-bar} is feasible and bounded, and $y^*_i >0$.

\item $g^* > g_i$, so that by setting $g_i$ to $g^*$ we indeed increases its value (\Cref{lmm:uL-xxstar-n-hatu}). This is important because it means that $g_i > 0$ still holds after the value change, so that we can apply \Cref{thm:xxi} inductively. 

\item Finally, let $(g', \gg_{-i}) = (g_1, \dots, g_{i-1}, g', g_{i+1}, \dots, g_{m_1 -1})$. 
Then $g' = g^*$ is the only point where $I_{\gg'} = I_\gg \cup \{i\}$ (\Cref{lmm:Ixxip-Ixxi}). 
This means that we can use binary search to pin down $g^*$. 
In particular, since ${g^*}$ is derived immediately from the solution to LP~\eqref{lp:xi-bar}, its bit-size is bounded by a polynomial in the size of the LP. We can find ${g^*}$ in polynomial time.
Note that we cannot compute $g^*$ directly using the expression \eqref{eq:g-star} because LP~\eqref{lp:xi-bar} involves unknown parameters.
\qedhere
\end{itemize}
\end{proof}

\subsection{Proof of \Cref{lmm:Ig-nonempty}: Omitted Results}


\begin{lemma}
\label{lmm:lp:xi-bar-feasible}
LP \eqref{lp:xi-bar} is feasible and bounded, and $y^*_i >0$. 
\end{lemma}

\begin{proof}
Notice that  \eqref{lp:xi-bar} is equivalent to computing the minimal value of $\max_{j\in [n-1]} \tilde{u}^F_{\gg}(\yy,j)$ with $\yy$ satisfying \eqref{lp:xi-bar-1} and \eqref{lp:xi-bar-3}. It is feasible and bounded  if there exists 
$\yy \in \Gamma_i$ such that $u^L(\yy, n) = V$.
Note that for all $\yy \in \Gamma_i$, we have 
\[
u^L(\yy, n) = y_i \cdot u^L(i,n) + (1 - y_i) \cdot M_{\{n\}}.
\]
Let $\zz$ be an arbitrary point such that $\sum_{k = m_1}^m z_k = 1$ (hence $\zz \in \Gamma_i$).
\begin{itemize}
\item  
When $\yy = \zz$, we have $u^L(\yy, n) = M_{\{n\}}$.
\item
When $\yy = \xx^i$, we have $u^L(\yy, n) = u^L(\xx^i, n) < V$ as $(\xx^i, n)$ is not an SSE by assumption.
\end{itemize}
By definition (i.e., \eqref{eq:xx-i}), $\xx^i \in \Gamma_i$. 
By assumption $I_\gg \neq \emptyset$, so $(\xx^k, n)$ is an SSE response for some $k \in [m_1 - 1]$.
Hence, by definition, 
\[
V = u^L(\xx^k, n) < \max_{\xx \in \Gamma_i} u^L(\xx, n) \le M_{\{n\}},
\]
where $u^L(\xx^k, n) < \max_{\xx \in \Gamma_k} u^L(\xx, n)$ follows by the fact that $u^L(\xx, n)$ decreases with $x_k$ in the space $\Gamma_k$ (\Cref{lmm:minxi-maxuL}), and $x_k^k > 0$ (\Cref{lmm:xi-gt-0}).
Therefore, by continuity, there must exist $\yy \in \Gamma_i$ such that $u^L(\yy, n) = V$. Moreover, for all such $\yy$, $y_i > 0$ as otherwise $u^L(\yy, n) = M_{\{n\}} > V$.
\end{proof}

\begin{lemma}
\label{lmm:uL-increase}
$u^L(\xx^i, n)$ increases strictly with $g_i$.
\end{lemma}
\begin{proof}
Pick two arbitrary values $g', g'' \in \mathbb{R}$ such that $g' < g''$.
Let
\begin{align*}
\gg' &= (g_1, \dots, g_{i-1}, g',\, g_{i+1}, g_{m_1 -1}) \\
\text{ and }\quad
\gg'' &= (g_1, \dots, g_{i-1}, g'' , g_{i+1}, g_{m_1 -1}).
\end{align*}
Let ${\xx'}^k$ and ${\xx''}^k$ denote the critical points defined with respect to $\gg'$ and $\gg''$, respectively (i.e., \eqref{eq:xx-i}). 
Suppose for the sake of contradiction that $u^L({\xx'}^i, n) \ge u^L({\xx''}^i, n)$.

By \Cref{lmm:minxi-maxuL}, $u^L({\xx'}^i, n) \ge u^L({\xx''}^i, n)$ implies that ${x'_i}^i \le {x''_i}^i$.
Moreover,
\begin{align}
\tilde{u}_{\gg'}^F(\xx'^i, n) 
&= g' \cdot {x_i'}^i + \sum_{k =m_1}^m {x_k'}^i \cdot W \nonumber \\ 
&< g'' \cdot {x_i'}^i + \sum_{k =m_1}^m {x_k'}^i \cdot W 
= \tilde{u}_{\gg''}^F({\xx'}^i, n), \label{eq:tu-xxipn-tu-xxipp-n}
\end{align}
where we used the fact that ${x_i'}^i > 0$ (\Cref{lmm:xi-gt-0}).

Since $n \in f_{\gg'}({\xx'}^i)$, by definition $\tilde{u}_{\gg'}^F({\xx'}^i, n) \ge \tilde{u}_{\gg'}^F({\xx'}^i, j)$ for all $j \in [n-1]$. Hence, using \eqref{eq:tu-xxipn-tu-xxipp-n}, we have
\[
\tilde{u}_{\gg''}^F(\xx'^i, n) >
\tilde{u}_{\gg'}^F(\xx'^i, n) \ge
\tilde{u}_{\gg'}^F(\xx'^i, j) =
\tilde{u}_{\gg''}^F(\xx'^i, j),
\]
where the last equality holds since $\tilde{u}_{\gg}^F(\cdot, j)$ does not depend on $\gg$ by construction.
This means that $\tilde{u}_{\gg''}^F(\xx, n) > \tilde{u}_{\gg''}^F(\xx, j)$ for $\xx$ in a neighborhood $\calN$ of ${\xx'}^i$.
Since ${x_i'}^i > 0$, ${\xx'}^i$ must lie in the relative interior of $\Gamma_i$.  So $\calN \cap \Gamma_i$ must contain a point $\xx$ such that $x_i < {x_i'}^i$.
We then establish the following contradictory transitions:
\[
\min_{\yy \in \Gamma_i \cap f_{\gg'}^{-1}(n)} y_i 
\le x_i < {x_i'}^i = 
\min_{\yy \in \Gamma_i \cap f_{\gg'}^{-1}(n)} y_i.
\]
This completes the proof.
\end{proof}

Hereafter, we let $\gg^* = (g_1, \dots, g_{i-1}, g^*, g_{i+1}, g_{m_1 -1})$, and let  ${\xx^*}^k$ be the critical points defined with respect to $\gg^*$ (i.e., \eqref{eq:xx-i}).

\begin{lemma}
\label{lmm:uL-xxstar-n-hatu}
$u^L({\xx^*}^i, n) = V$, and thus ${g^*} > g_i$.
\end{lemma}

\begin{proof}
According to \Cref{lmm:uL-increase}, since $u^L(\xx^i,n)$ increases strictly with $g_i$, if $u^L({\xx^*}^i,n) = V > u^L({\xx}^i,n)$, then $g^* > g_i$. 

Since $\yy^* \in \Gamma_i$ (by \eqref{lp:xi-bar-1}), we have 
\[
\tilde{u}_{\gg^*}^F(\yy^*, n) = y_i^* \cdot g^* + (1 - y_i^*) \cdot W = u^*.
\]
Note that $\tilde{u}_{\gg'}^F(\cdot, j)$ is 
not dependent on $\gg'$ for all $j \in [n-1]$. Hence, since $\yy^*$ and $u^*$ satisfy \eqref{lp:xi-bar-2}, we have 
\[
\tilde{u}_{\gg^*}^F(\yy^*, j) \le u^*, \text{ for all }j \in [n-1].
\]
As a result, $\tilde{u}_{\gg^*}^F(\yy^*, n) \ge \tilde{u}_{\gg^*}^F(\yy^*, j)$ for all $j \in [n]$, which means that $\yy^* \in f_{\gg^*}^{-1}(n)$.
Using \Cref{lmm:minxi-maxuL}, we have $u^L({\xx^*}^i, n) \ge u^L(\yy^*, n) = V$, where $u^L(\yy^*, n) = V$ follows by \eqref{lp:xi-bar-3}.

To see that $u^L({\xx^*}^i, n) \le V$, suppose for the sake of contradiction that $u^L({\xx^*}^i, n) > V$. 
Since $u^L(\yy^*, n) = V$, using Lemma~\ref{lmm:minxi-maxuL} again, we have that ${x_i^*}^i < y_i^*$.
Consider a point 
\[
\zz = \lambda \cdot {\xx^*}^i + (1 - \lambda) \cdot \ee^i,
\]
where $\lambda \in [0,1]$, and $\ee^i = (e_1^i, \dots, e_m^i)$ such that $e_i^i = 1$ and $e_k^i = 0$ for all $k \neq i$.
Since both ${\xx^*}^i$ and $\ee^i$ are in $\Gamma_i$, we have $\zz \in \Gamma_i$.
Moreover, since $e_i^i = 1 \ge y_i^*$, there exists $\lambda \in [0,1)$ with which we have $z_i = y_i^*$, and in turn 
\begin{equation}
\label{eq:tu-zz-ustar}
\tilde{u}_{\gg^*}^F(\zz, n) = \tilde{u}_{\gg^*}^F(\yy^*, n) = u^*,
\end{equation}
and 
\begin{equation}
\label{eq:uL-zz-hatu}
u^L(\zz, n) = u^L(\yy^*, n) = V.
\end{equation}

Meanwhile, note that for all $j \in [n-1]$, we have
\[
\tilde{u}_{\gg^*}^F(\ee^i, n) = g^* > 0 = \tilde{u}_{\gg^*}^F(\ee^i, j),
\]
In addition, since ${\xx^*}^i \in f_{\gg^*}^{-1}(n)$, by definition we have
\[
\tilde{u}_{\gg^*}^F({\xx^*}^i, n) \ge \tilde{u}_{\gg^*}^F({\xx^*}^i, j).
\]
Consequently, since $\lambda < 1$, we must have
\[
\tilde{u}_{\gg^*}^F(\zz, n) > \tilde{u}_{\gg^*}^F(\zz, j).
\]
Plugging into \eqref{eq:tu-zz-ustar} gives $\tilde{u}_{\gg^*}^F(\zz, j) < u^*$ for all $j \in [n-1]$.
Pick 
$u \in [\max_{j \in [n-1]} \tilde{u}_{\gg^*}^F(\zz, j), u^*)$.
The fact that $\zz \in \Gamma_i$ and $u^L(\zz, n) = V$ (i.e., \eqref{eq:uL-zz-hatu}) we showed above implies that $u$ and $\zz$ form a feasible solution to LP~\eqref{lp:xi-bar}. The objective value of this solution, i.e., $u$, is strictly smaller than $u^*$, which is a contradiction.
This completes the proof of \Cref{lmm:uL-xxstar-n-hatu}.
\end{proof}

\begin{lemma}
\label{lmm:Ixxip-Ixxi}
$I_{\gg'} = I_\gg$ if $g' \in [g_i, {g^*})$, $I_{\gg'} = \{i\}$ if $g' > {g^*}$, and $I_{\gg'} = I_\gg \cup \{i\}$ if $g' = {g^*}$.
\end{lemma}

\begin{proof}
We first argue that the following results hold for any $g' \ge g_i$.
\begin{itemize}
\item[(i)]
$\max_{\xx \in f_{\gg'}^{-1}(j)} u^L(\xx, j) \le \max_{\xx \in f_{\gg}^{-1}(j)} u^L(\xx, j)$, for all $j \in [n-1]$; and 
\item[(ii)]
$u^L({\xx'}^k,n) = u^L(\xx^k,n)$, for all $k \in [m_1 - 1] \setminus \{i\}$.
\end{itemize}
In words, (i) says that the leader's payoffs for inducing a best response $j \in [n-1]$ does not increase with $g_i$;
(ii) says that the leader's payoffs for $(\xx^k,n)$, $k \neq i$, does not change with $g_i$.

To see (i), 
note that by construction, $\tilde{u}_{\gg}^F(\xx, j)$ does not depend on $g_i$ for all $j \in [n-1]$.
Moreover, now that $g' \ge g_i$, $\tilde{u}_{\gg'}^F(\xx, n) \ge \tilde{u}_{\gg}^F(\xx, n)$ for all $\xx \in \Delta_m$. 
Therefore, any leader strategy $\xx$ that does not induce $j$ under $\tilde{u}_{\gg}^F$ does not induce $j$ under $\tilde{u}_{\gg'}^F$, either. We have $f_{\gg'}^{-1}(j) \subseteq f_\gg^{-1}(j)$, so (i) follows immediately.

To see (ii), 
note that $\tilde{u}_\gg^F$ does not depend on $g_i$ in the space $\Gamma_k$, $k \neq i$. 
Hence, $\Gamma_k \cap f_{\gg'}^{-1}(n) = \Gamma_k \cap f_\gg^{-1}(n)$. 
By Lemma~\ref{lmm:minxi-maxuL}, the statement then follows.

We proceed with the proof.
By assumption, $I_{\gg} \neq \emptyset$ and we can assume that $1 \in I_{\gg}$.
Hence, $u^L(\xx^1, n) = V$. 
The statements above then imply that 
\begin{align*}
\max_{j \in [n-1]} \max_{\xx \in f_{\gg'}^{-1}(j)} u^L(\xx, j) 
&\le \max_{j \in [n-1]} \max_{\xx \in f_{\gg}^{-1}(j)} u^L(\xx, j) \\  
&\le V \\
&= u^L(\xx^1, n) 
= u^L({\xx'}^1, n)
\le \max_{\xx \in f_{\gg'}^{-1}(n)} u^L(\xx, n),
\end{align*}
where $\max_{j \in [n-1]} \max_{\xx \in f_{\gg}^{-1}(j)} u^L(\xx, j) \le V$ follows by the definition of the SSE, and $u^L({\xx'}^1, n) \le \max_{\xx \in f_{\gg'}^{-1}(n)} u^L(\xx, n)$ since $\xx'^{-1} \in f_{\gg'}^{-1}(n)$ by definition.
This further means that 
\[
\max_{j \in [n]} \ \max_{\xx \in f_{\gg'}^{-1}(j)} u^L(\xx, j) =
\max_{\xx \in f_{\gg'}^{-1}(n)} u^L(\xx, n),
\]
and applying \Cref{lemma:boundary-point}, we get that 
\[
\max_{k \in [m_1 - 1]} u^L({\xx'}^k,n) = \max_{j \in [n]} \ \max_{\xx \in f_{\gg'}^{-1}(j)} u^L(\xx, j).
\]

Consequently,
\begin{itemize}
\item when $g' \in [g_i, g^*)$, we have $\max_{k \in [m_1 - 1]} u^L({\xx'}^k,n) = V > u^L({\xx'}^i,n)$;
\item when $g' = g^*$, we have $\max_{k \in [m_1 - 1]} u^L({\xx'}^k,n) = V = u^L({\xx'}^i,n)$; and 
\item when $g' > g^*$, we have $\max_{k \in [m_1 - 1]} u^L({\xx'}^k,n) = u^L({\xx'}^i,n) > V$.
\end{itemize}
\Cref{lmm:Ixxip-Ixxi} then follows.
\end{proof}

\subsection{Proof of \Cref{lmm:find-cover-mu}}

\lmmfindcovermu*

\begin{proof}
To decide whether $\mu$ is a cover of $S$, we check the satisfiability of the following linear constraints (by assumption of \Cref{thm:find-cover-given-aj}, $\calM_S$ is given as a set of linear constraints): 
\begin{align*}
& \xx \in \calM_S \\
& \mu(\xx,j) \leq 0  && \text{ for all } j \in [n]
\end{align*}
which can be done in polynomial time.
Specifically: 
\begin{itemize}
\item 
If the constraints are satisfiable, then there exists $\yy \in \calM_S$ such that 
$\mu(\yy, j) \le 0 = \mu(\yy, k)$ for all $j \in [n], k \in S$; hence, $k \in h(\yy)$, where $h$ denotes the BR-correspondence of $\mu$.
We have
\[
\max_{\xx \in \calM_S, j \in h(\xx)} u^L(\xx, j) \ge u^L(\yy, k) \ge \min_{j \in S} u^L(\xx, j) = M_S.
\]
By definition, this means that $\mu$ is {\em not} a cover of $S$.

\item 
Conversely, suppose that the constraints are not satisfiable.
Pick arbitrary $\xx \in \calM_S$ and $j \in h(\xx)$. Hence, $\mu(\xx, j) > 0$, and according to the definition of $\mu$, it must be that $j \notin S$.
We have $\mu(\xx, j)= \vb_j \cdot \xx - c_j > 0$, which implies that $u^L(\xx, j) < M_S$ by \eqref{eq:bbj-cj}. 
Since the choice of $\xx$ and $k$ is arbitrary, we then have
\[
\max_{\xx \in \calM_S, j \in h(\xx)} u^L(\xx, j) < M_S,
\]
so $\mu$ is a cover of $S$.
\end{itemize}

Next, consider the second part of the statement.
Suppose that $S$ admits a cover. We show that $\mu$ must be a cover of it.
According to the necessary condition demonstrate in \Cref{lmm:cover-iff}, we have $M_S > M_{[n]}$.
Hence, for any $\xx \in \calM_S$, $\min_{j \in [n]} u^L(\xx,j) \le M_{[n]} < M_S$, which means that there exists $j \in [n]\setminus S$ such that $u^L(\xx,j) < M_S$. 
By \eqref{eq:bbj-cj}, we then have $\vb_j \cdot \xx > c_j$, which means
\[
\mu(\xx, j) = \vb_j \cdot \xx - c_j > 0.
\]
Namely, the linear constraints we presented above are not satisfiable.
Therefore, $\mu$ is a cover of $S$.
\end{proof}

\subsection{Proof of \Cref{lmm:find-cover-bbj-cj}}
\label{sec:proof-of-lmm-find-cover-bbj-cj}
\lmmfindcoverbbjcj*

\begin{proof}
Without loss of generality, we show how to learn $\vb_n$ and $c_n$, and we assume that $n \notin S$.
We first define the following parameters:
\begin{align*}
\check{d}\coloneqq \min_{\xx \in \calM_S}\va_j\cdot\xx, 
\quad
\hat{d} \coloneqq \max_{\xx \in \calM_S}\va_j\cdot\xx, 
\quad
\text{and }\quad
d^*\coloneqq (M_{S}-\beta_n)/\gamma_n. 
\end{align*}
Then specify $\vb_n$ and $c_n$ as follows:
\begin{itemize}
\item If $n \notin Q$, we let $\vb_n = \va_n$ and $c_n = 
\begin{cases}
\check{d} - 1,  & \text{ if } d^* < \check{d}; \\
d^*, & \text{ if } d^* \in [\check{d}, \hat{d}]; \\
\hat{d}, & \text{ if } d^* > \hat{d}.
\end{cases}$. 
(By assumption of \Cref{thm:find-cover-given-aj}, $\va_n$ is given if $n \notin S \cup Q$.) 

\item If $n \in Q$, we let 
$b_{n,i} = \begin{cases}
1, & \text{if } u^L(i,n) = M_{\{n\}} \\
0, & \text{otherwise}
\end{cases}$, and $c_n = \check{d}$.
\end{itemize}

Note that $\check{d}$ and $\hat{d}$ can be computed directly given that $\calM_S$ is known (by assumption of \Cref{thm:find-cover-given-aj}). 
Nevertheless, $d^*$ cannot be computed directly since $\beta_n$, $\gamma_n$, and $M_S$ are unknown. To complete the proof, we show how to learn $d^*$ next.

Define the following payoff function parameterized by a number $d$.
\begin{boxedtext}
For every $\xx \in \Delta_m$,
\begin{equation}
\label{eq:tilde-uF-find-cover}
    \tilde{u}^F_d(\xx,j) \coloneqq \begin{cases}
    \tilde{u}^F(\xx, j), & \text{if } j \in [n-1];\\
    \tilde{u}^F(\xx, k) + (d- \vb_n \cdot\xx), & \text{if } j = n; 
    \end{cases}
\end{equation}
where we use arbitrary $k \in \argmax_{\ell \in [n-1]} \tilde{u}^F(\zz, \ell)$, defined with arbitrarily selected
\[
\zz \in Z(d)\coloneqq \left\{\xx \in \calM_S: \vb_n \cdot \xx = d \right\}.
\]
$\zz$ is well-defined if $d \in [\check{d}, \hat{d}]$.
\end{boxedtext}
We argue that $n$ is an SSE response of the game $(u^L, \tilde{u}_d^F)$ if and only if $d \ge d^*$, so that we can use binary search to find out $d^*$ (or find out if $d^* < \check{d}$ or $d^* \ge \hat{d}$, in which case we only need $\check{d}$ and $\hat{d}$ to compute $\vb_n$ and $c_n$ as defined above).

Denote the BR-correspondences of $\tilde{u}^F_d$ and $\tilde{u}^F$ as $f_d$ and $f$, respectively.
Note the following facts:
\begin{enumerate}[label=(\alph*)]
\item $f_d(\zz) = f(\zz) \cup \{n\}$.

\item $f_d^{-1}(j) \subseteq f^{-1}(j)$ for all $j \in [n-1]$.

\item $\max_{\xx \in f_d^{-1}(n)} \vb_n \cdot \xx = d$ if $d \in [\check{d}, \hat{d}]$.

\item $u^L(\xx, j) = M_S$ for all $\xx \in \calM_S$ and $j \in f(\xx)$.
\end{enumerate}

Indeed, (a) can be verified by comparing $\tilde{u}_d^F(\zz, j)$, $j\in[n-1]$, with $\tilde{u}^F(\zz, n)$.
(b) is due to the fact that $f^{-1}(n) = \emptyset$ according to the property of $\tilde{u}^F$ as a base function.

To see (c), note that $n \in f_d(\zz)$ implies that 
\[
\max_{\xx \in f_d^{-1}(n)} \vb_n \cdot \xx \ge \vb_n \cdot \zz = d.
\]
Moreover, for all $\xx \in f_d^{-1}(n)$, 
\[
\tilde{u}^F_d(\xx, n) 
\geq \max_{j \in [n-1]} \tilde{u}^F_d(\xx,j) 
= \max_{j \in [n-1]} \tilde{u}^F(\xx,j)
\ge \tilde{u}^F(\xx,k),
\]
which implies $\vb_n \cdot\xx \leq d$ according to \eqref{eq:tilde-uF-find-cover}. Hence, $\max_{\xx \in f_d^{-1}} \vb_n \cdot \xx = d$.

Finally, if (d) did not hold, then $u^L(\xx,j)\ne M_S$ for some $\xx \in \calM_S$ and $j \in f(\xx)$. Since $\tilde{u}^F$ is a base function, by definition, $j \in f(\xx) \subseteq S$. Moreover, $\xx \in \calM_S$ means that $\min_{k \in S} u^L(\xx,k) = M_S$, so we must have $u^L(\xx,j) > M_S$, which contradicts the definition of a base function. 

\smallskip

For each $j \in [n]$, define
\[
V_d^L(j) \coloneqq \max_{\xx \in f_d^{-1}(j)} u^L(\xx,j),
\]
which is the leader's maximum attainable payoff for inducing a follower response $j$. 
According to (b), for all $j \in [n-1]$, we have 
\[
V_d^L(j) \le \max_{j'\in [n]} \max_{\xx \in f^{-1}(j')} u^L(\xx, j') = M_S,
\]
where the second transition follows by the property of $\tilde{u}^F$ as a base function. 
Moreover, pick arbitrary $j \in f(\zz)$; by (a), (d), and the fact that $\zz \in \calM_S$, we get that 
\[
\max_{j' \in [n-1]} V_d^L(j') \ge u^L(\xx, j) = M_S.
\]

Therefore, $\max_{j' \in [n-1]} V_d^L(j') = M_S$.
By (c), we also have $\coloneqq V_{d}^L(n) = \gamma_n \cdot d + \beta_n$ if $d \in [\check{d},\hat{d}]$.
It follows that 
\[
\Phi(d) \coloneqq V_{d}^L(n) - \max_{j \in [n-1]} V_{d}^L(j) = \gamma_n \cdot d + \beta_n - M_S
\]
is continuous and strictly increasing with respect to $d$.
We can then use binary search and oracle $\ASSE$ to pin down $d^*$ (or decide if $d^* < \check{d}$ or  $d^* > \hat{d}$):
$n$ is an SSE response of $(u^L, \tilde{u}_d^F)$ if and only if $d \geq d^*$;
meanwhile, there exists an SSE response $j \in [n-1]$ if and only if $d \le d^*$.
\end{proof}

\section{Omitted Proofs in \NoCaseChange{\Cref{sec:learn-j-xxs}}}

\subsection{Proof of \Cref{lmm:tilde-uL-J}}

\lmmtildeuLJ*

\begin{proof}
Note that for all $j \in [n]$, we have
\begin{align}
& u^L(\xx^*, j) = M_{[n]} \quad \Longrightarrow \quad j \in \widehat{J}, \label{eq:lmm-J-xstar-eq1} \\
\text{ and } \quad
& \tilde{u}^L(\xx^*, j) = \widetilde{M}_{[n]} \quad \Longrightarrow \quad j \in \widehat{J}.
\label{eq:lmm-J-xstar-eq2}
\end{align}
Specifically, \eqref{eq:lmm-J-xstar-eq1} 
follows by the definition of $\widehat{J}$: if $j \notin \widehat{J}$, then by definition we have $M_{\{1\}} \le \min_{\xx \in \Delta_m} u^L(\xx, j)$; hence,
\begin{align*}
M_{[n]}
< M_{\{1\}} 
\le \min_{\xx \in \Delta_m} u^L(\xx, j)
\le u^L(\xx^*, j)
\end{align*}
(where $M_{[n]} < M_{\{1\}}$ is due to \Cref{asn:aj-Mj}).
Moreover, if $j \notin \widehat{J}$, according to \eqref{eq:tilde-u-L}, we have 
\[
\tilde{u}^L(\xx^*, j) 
= W 
> \max_{\xx \in \Delta_m} \min_{k \in [n]} \tilde{u}^L(\xx, k) 
= \widetilde{M}_{[n]}.
\]
Hence, \eqref{eq:lmm-J-xstar-eq2} holds.

For all $j \in \widehat{J}$, we have
\[
\tilde{u}^L(\xx, j) 
= \frac{\gamma_j}{\gamma_1} \cdot \va_j \cdot \xx + \frac{\beta_j - \beta_1}{\gamma_1}
= u^L(\xx, j)/\gamma_1 - \beta_1/\gamma_1,
\]
which means that the following two statements are equivalent:
\begin{align}
& u^L(\xx^*, j) = \max_{\xx \in \Delta_m} \min_{k \in J} u^L(\xx, k) = \max_{\xx \in \Delta_m} \min_{k \in [n]} u^L(\xx, k), \label{eq:lmm-J-xstar-eq3}\\
\text{ and }\quad
& 
\tilde{u}^L(\xx^*, j) = \max_{\xx \in \Delta_m} \min_{k \in J} \tilde{u}^L(\xx, k) = \max_{\xx \in \Delta_m} \min_{k \in [n]} \tilde{u}^L(\xx, k).
\label{eq:lmm-J-xstar-eq4}
\end{align}
That is, for all $j \in \widehat{J}$, \eqref{eq:lmm-J-xstar-eq3} $\Longleftrightarrow$ \eqref{eq:lmm-J-xstar-eq4}.
Using this equivalence and \eqref{eq:lmm-J-xstar-eq1} and \eqref{eq:lmm-J-xstar-eq2}, 
we establish the following transitions to complete the proof:
\begin{align*}
& \forall j \in J: u^L(\xx^*,j) = M_J = M_{[n]} \\
\quad \Longleftrightarrow \quad 
& 
J \subseteq \widehat{J},
\text{ and }
\forall j \in J:
u^L(\xx^*, j) = \max_{\xx \in \Delta_m} \min_{k \in J} u^L(\xx, k) = \max_{\xx \in \Delta_m} \min_{k \in [n]} u^L(\xx, k)
\\
\quad \Longleftrightarrow \quad 
& 
J \subseteq \widehat{J},
\text{ and }
\forall j \in J:
\tilde{u}^L(\xx^*, j) = \max_{\xx \in \Delta_m} \min_{k \in J} \tilde{u}^L(\xx, k) = \max_{\xx \in \Delta_m} \min_{k \in [n]} \tilde{u}^L(\xx, k)
\\
\quad \Longleftrightarrow \quad 
& \forall j \in J: \tilde{u}^L(\xx^*,j) = \widetilde{M}_J = \widetilde{M}_{[n]}. 
&&\qedhere
\end{align*}
\end{proof}

\section{Omitted Proofs in \NoCaseChange{\Cref{sec:first-reference-pair}}}

\subsection{An Useful Observation Used in All proofs}

In all the proofs in \Cref{sec:first-reference-pair}, we make the following key assumption to ease our presentation. 

\begin{lemma}
\label{asn:2-non-dominiated}
Without loss of generality, we can assume that 
there exists $i \in I_2$ such that $u^L(i, 2) > u^L(i, 1)$. (Recall that $I_j \coloneqq \argmax_{i' \in [m]} u^L(i', j)$ for all $j \in [n]$.) 
\end{lemma}

\begin{proof}
To see the rationale of this assumption, consider the case where it does not hold, i.e., 
$u^L(i, 1) \ge u^L(i, 2)$ for all $i \in I_2$.
By \Cref{asn:1-min-M}, $M_{\{2\}} \ge M_{\{1\}}$; hence, we have
\[
M_{\{1\}} \ge u^L(i, 1) \ge u^L(i, 2) = M_{\{2\}} \ge M_{\{1\}},
\]
where $u^L(i, 2) = M_{\{2\}}$ since $i \in I_2$.
This means $M_{\{1\}} = M_{\{2\}}$ and $u^L(i, 1) = u^L(i, 2)$ for all $i \in I_2$. So $I_2 \subseteq I_1$.

Hence, if there exists $i \in I_1$ such that $u^L(i, 1) > u^L(i, 2)$, we can exchange the roles of actions $1$ and $2$ so that all the assumptions made will hold (in particular, $M_{\{1\}} = M_{\{2\}}$ means that \Cref{asn:1-min-M} still holds as $2 \in \argmin_{j\in[m]} M_{\{j\}}$).
We can then proceed with the subsequent algorithm and eventually learn $\gamma_1/\gamma_2$ and $(\beta_1-\beta_2)/\gamma_2$---from which the original target quantities $\gamma_2/\gamma_1$ and $(\beta_2-\beta_1)/\gamma_1$ can be derived readily.

If otherwise $u^L(i, 1) \le u^L(i, 2)$ for all $i \in I_1$, then we have $M_{\{1\}} = u^L(i, 1) \le u^L(i, 2) \le M_{\{2\}}$, so it must be that $u^L(i, 1) = u^L(i, 2)$ for all $i \in I_1$. Given that $I_2 \subseteq I_1$ as we argued, we have $I_2 = I_1$ in this case. 
We further consider the following possibilities: 
\begin{itemize}
\item 
$I_j = I_1$ for all $j \in \widehat{J}$. 
Pick arbitrary $i \in I_1$, we then have $u^L(i, j) = M_{\{j\}} \ge M_{\{1\}}$ for all $j \in \widehat{J}$; 
moreover, according to the definition of $\widehat{J}$, $u^L(i,j) \ge M_{\{1\}}$ for all $j \notin \widehat{J}$.
We get that $\min_{j \in [n]} u^L(i, j) \ge M_{\{1\}}$ .
It follows that
$M_{[n]} = \max_{\xx \in \Delta_m} \min_{j \in [n]} u^L(\xx,j) \ge \min_{j \in [n]} u^L(i, j) \ge M_{\{1\}}$,
which contradicts \Cref{asn:aj-Mj}.

\item
$I_j \neq I_1$ for some $j \in \widehat{J}$ (and hence, $I_j \neq I_2$).
This means that we are able to learn $\gamma_j/\gamma_i$ and $(\beta_j-\beta_i)/\gamma_i$ for both $i \in \ots$ with our subsequent algorithm (i.e., by putting $j$ in place of action $2$ and $i$ in place of action $1$).
We can then derive $\gamma_2/\gamma_1$ and $(\beta_2-\beta_1)/\gamma_1$ as follows:
\begin{align*}
\gamma_2/\gamma_1 = \frac{\gamma_j/\gamma_1} {\gamma_j/\gamma_2}, 
\quad\text{ and }\quad
\frac{\beta_2-\beta_1}{\gamma_1} =
\frac{\beta_j-\beta_1}{\gamma_1} -
\frac{\beta_j-\beta_2}{\gamma_2} \cdot \frac{\gamma_2}{\gamma_1}.
&
\hfill \qedhere
\end{align*}
\end{itemize}
\end{proof}

\subsection{Proof of \Cref{lmm:fd-main}}

\lmmfdmain*

To prove \Cref{lmm:fd-main}, we define
\begin{equation}
\label{eq:V-d}
V_d^L(j) \coloneqq \max_{\xx \in f_d^{-1}(j)} u^L(\xx, j)
\end{equation}
for each $j \in [n]$, 
which is the maximum payoff the leader can obtain by inducing the follower to respond with $j$.
We also let $V_d^L(j) = -\infty$ if $f_d^{-1}(j) = \emptyset$.
By construction $f_d^{-1}(j) = \emptyset$ for all $j \in [n] \setminus \ots$, so only actions $1$ and $2$ can be SSE responses. 
We  prove that $V_d^L(1) \ge V_d^L(2)$ if and only if $d \le d^*$ and $V_d^L(1) \le V_d^L(2)$ if and only if $d \ge d^*$.
Indeed, in the case where $d = d^*$ the following lemma shows that $V_{d^*}^L(1) = V_{d^*}^L(2) = M_\ots$.

\begin{lemma}
\label{lmm:hat-V1-V2-M-12}
$V_{d^*}^L(1) = V_{d^*}^L(2) = M_\ots$.
\end{lemma}

We defer the proof of \Cref{lmm:hat-V1-V2-M-12} to \Cref{sec:proof-lmm:hat-V1-V2-M-12} and proceed with the proof of \Cref{lmm:fd-main}.

\begin{proof}[Proof of \Cref{lmm:fd-main}]
Clearly, $V_d^L(1)$ is non-increasing with respect to $d$. Moreover, we have $V_{d^*}^L(1) = V_{d^*}^L(2)$ according to \Cref{lmm:hat-V1-V2-M-12}.
Hence, it suffices to prove that $V_d^L(1) > V_d^L(2)$ for all $d < d^*$, and $V_d^L(1) < V_d^L(2)$ for all $d > d^*$.

If $d < d^*$, then by construction we have $\va_2 \cdot \xx \le d < d^*$ for all $\xx \in f_{d}^{-1}(2)$.
Hence, 
$\max_{\xx \in f_{d}^{-1}(2)} \va_2 \cdot \xx < d^*$, and 
\[
V_d^L(2) 
< \gamma_2 \cdot d^* + \beta_2 
= M_\ots
= V_{d^*}^L(1)
\le V_{d}^L(1),
\]
where $V_{d^*}^L(1) \le V_{d}^L(1)$ as $V_{d}^L(1)$ is non-increasing with respect to $d$.

If $d > d^*$, consider the following two cases.

\begin{itemize}
\item $f_{d}^{-1}(1) = \emptyset$. 
Then it follows immediately that 
$V_{d}^F(2) \ge V_{d^*}^F(2) = M_\ots > -\infty = V_{d}^F(1)$.

\item $f_{d}^{-1}(1) \neq \emptyset$.
Pick arbitrary $\xx \in f_{d}^{-1}(1)$.
By definition, we have $\va_2 \cdot \xx \ge d$.
Since $V_{d^*}^L(2) = M_\ots$, there exists $\yy \in \Delta_m$ such that $u^L(\yy, 2) = M_\ots$, which means that 
$\va_2 \cdot \yy = (M_\ots - \beta_2)/\gamma_2 = d^* < d$.
Hence, there exists a convex combination $\zz$ of $\xx$ and $\yy$ such that $\va_2 \cdot \zz = d$.
We have $\zz \in f_{d}^{-1}(2)$, which implies that
\begin{align*}
V_d^L(2) 
\ge u^L(\zz, 2) 
&= \gamma_2 \cdot d + \beta_2 \\
&> \gamma_2 \cdot d^* + \beta_2
= M_\ots
\geq V_d^L(1).
\end{align*}
\end{itemize}
This completes the proof.
\end{proof}

\subsection{Proof of \Cref{lmm:hat-V1-V2-M-12}}
\label{sec:proof-lmm:hat-V1-V2-M-12}

We prove the lemma in two parts: 
(1) $V_{d^*}^L(1) = M_\ots$ (\Cref{lmm:hat-u1-M-12}), and 
(2) $V_{d^*}^L(2) = M_\ots$ (\Cref{lmm:hat-u2-M-12}).
En route, we also prove a result (stated in \Cref{lmm:hat-u2-M-12}) that will be useful in the next section.

\begin{lemma}
\label{lmm:hat-u1-M-12}
$V_{d^*}^L(1) = M_\ots$.
\end{lemma}

\begin{proof}
Suppose for the sake of contradiction that $V_{d^*}^L(1) \neq M_\ots$ and consider the following cases.

\paragraph{Case 1.} $V_{d^*}^L(1) < M_\ots$.
Since $V_{d^*}^L(1) < M_\ots$, then for all $\xx \in f_{d^*}^{-1}(1)$ we have
\[
u^L(\xx, 1) \le \max_{\xx' \in f_{d^*}^{-1}(1)} u^L(\xx', 1) = V_{d^*}^L(1) < M_\ots.
\]

For all $\xx \in \Delta_m \setminus  f_{d^*}^{-1}(1)$, by construction we have $\va_2 \cdot \xx < d^*$; hence,
\begin{equation}
\label{eq:lmm:hat-u1-M-12:eq-0}
u^L(\xx, 2) 
= \gamma_2 \cdot \va_2 \cdot \xx + \beta_2 
< \gamma_2 \cdot d^* + \beta_2
= M_\ots.
\end{equation}

Therefore, for all $\xx \in \Delta_m$, we have $\min_{j \in \ots} u^L(\xx, j) < M_\ots$. This leads to a contradiction: 
$M_\ots = \max_{\xx \in \Delta_m} \min_{j \in \ots} u^L(\xx, j) < M_\ots$.

\paragraph{Case 2.} $V_{d^*}^L(1) > M_\ots$.
We further consider the following two cases.
\begin{itemize}
\item[(i)] 
There exists $\yy_2 \in f_{d^*}^{-1}(1)$ such that \begin{equation}
\label{eq:lmm:hat-u1-M-12:eq-1}
u^L(\yy_2, 2) > M_\ots.
\end{equation}
Pick arbitrary $\yy_1 \in \argmax_{\xx \in f_{d^*}^{-1}(1)} u^L(\xx, 1)$.
By assumption, $V_{d^*}^L(1) > M_\ots$, so 
\begin{equation}
\label{eq:lmm:hat-u1-M-12:eq-2}
u^L(\yy_1, 1) = V_{d^*}^L(1) > M_\ots.
\end{equation}
The fact that $\yy_1 \in f_{d^*}^{-1}(1)$ also implies that $\va_2 \cdot \yy_1 \ge d^*$ by construction. Hence,
\begin{equation}
\label{eq:lmm:hat-u1-M-12:eq-3}
u^L(\yy_1, 2) 
= \gamma_2 \cdot \va_2 \cdot \yy_1 + \beta_2
\ge \gamma_2 \cdot d^* + \beta_2
= M_\ots.
\end{equation}
Let $\zz = \lambda \cdot \yy_1 + (1 - \lambda) \cdot \yy_2$, where $\lambda \in (0,1)$.
By continuity, \eqref{eq:lmm:hat-u1-M-12:eq-2} implies that when $\lambda$ is sufficiently close to $1$, we have
$u^L(\zz, 1) > M_\ots$.
Moreover, \eqref{eq:lmm:hat-u1-M-12:eq-1}, \eqref{eq:lmm:hat-u1-M-12:eq-3} and the fact that $\lambda < 1$ imply that 
$u^L(\zz, 2) > M_\ots$.
Thus, $\min_{j \in \ots} u^L(\zz, j) > M_\ots$, which contradicts the definition of $M_\ots$, i.e., $M_\ots = \max_{\xx \in \Delta_m} \min_{j \in \ots} u^L(\xx, j) \ge \min_{j \in \ots} u^L(\zz, j)$.

\item[(ii)] 
$u^L(\xx, 2) \le M_\ots$ for all $\xx \in f_{d^*}^{-1}(1)$.
Note that for all $\xx \in \Delta_m \setminus f_{d^*}^{-1}(1)$, it holds that $u^L(\xx, 2) < M_\ots$ (see \eqref{eq:lmm:hat-u1-M-12:eq-0}).
Hence, now we have 
$u^L(\xx, 2) \le M_\ots$, 
for all $\xx \in \Delta_m$, which implies that
\[
M_{\{2\}} 
= \max_{\xx \in \Delta_m} u^L(\xx, 2) 
\le M_\ots.
\]
According to the assumption of Case~2, we have $V_{d^*}^L(1) > M_\ots$.
It follows that
\[
M_{\{2\}} \le M_\ots < V_{d^*}^L(1)
\le \max_{\xx \in \Delta_m} u^L(\xx, 1)
= M_{\{1\}}.
\]
This contradicts \Cref{asn:1-min-M}.
\end{itemize}

Therefore, both cases lead to contradictions. We have $V_{d^*}^L(1) = M_\ots$.
\end{proof}

\begin{lemma}
\label{lmm:hat-u2-M-12}
$V_{d^*}^L(2) = M_\ots$.
Moreover, there exists $\xx^- \in \Delta_m$ such that $\va_2 \cdot \xx^- < d^*$.
\end{lemma}

\begin{proof}
According to \Cref{lmm:hat-u1-M-12}, 
$f_{d^*}^{-1} (1) \neq \emptyset$,
which means that there exists $\xx^+ \in \Delta_m$ such that $\va_2 \cdot \xx^+ \ge d^*$.

We next show that there exists $\xx^- \in \Delta_m$ such that $\va_2 \cdot \xx^- < d^*$.
Suppose for the sake of contradiction that $\va_2 \cdot \xx \ge d^*$ for all $\xx \in \Delta_m$.
By definition we then have $f_{d^*}^{-1}(1) = \Delta_m$, so we get that
\[
V_{d^*}^L(1) 
= \max_{\xx \in f_{d^*}^{-1}(1)} u^L(\xx, 1)
= \max_{\xx \in \Delta_m} u^L(\xx, 1)
= M_{\{1\}}.
\]
By \Cref{lmm:hat-u1-M-12}, $V_{d^*}^L(1) = M_\ots$, so we have
\[
M_{\{1\}} = M_\ots.
\]
Now that $\va_2 \cdot \xx \ge d^*$ for all $\xx \in \Delta_m$ by assumption,
we have 
\[
\min_{\xx \in \Delta_m} u^L(\xx, 2) 
= \min_{\xx \in \Delta_m} \gamma_2 \cdot \va_2 \cdot \xx + \beta_2
\ge \gamma_2 \cdot d^* + \beta_2
= M_\ots.
\]
Hence, $\min_{\xx \in \Delta_m} u^L(\xx, 2) \ge M_{\{1\}}$, which contradicts \Cref{asn:2-hat-J} and the definition of $\widehat{J}$.

Therefore, we obtain two points $\xx^+, \xx^- \in \Delta_m$ such that $\va_2 \cdot \xx^+ \ge d^*$ and $\va_2 \cdot \xx^- < d^*$. 
There must be a convex combination $\yy$ of $\xx^-$ and $\xx^+$ such that $\yy \in \Delta_m$ and $\va_2 \cdot \yy = d^*$.
Hence, $\yy \in f_{d^*}^{-1}(2)$ and
\[
d^*
\ge \max_{\xx \in f_{d^*}^{-1}(2)} \va_2 \cdot \xx
\ge \va_2 \cdot \yy 
= d^*,
\]
where $d^*
\ge \max_{\xx \in f_{d^*}^{-1}(2)} \va_2 \cdot \xx$ holds as $f_{d^*}^{-1}(2) = \{\xx \in \Delta_m : \va_2 \cdot \xx \le d^*\}$ by construction.
This means that $\max_{\xx \in f_{d^*}^{-1}(2)} \va_2 \cdot \xx = d^*$.
Consequently, $V_{d^*}^L(2) = \gamma_2 \cdot d^* + \beta_2 = M_\ots$.
\end{proof}

\subsection{Proof of \Cref{thm:compute-d-star}}

\thmcomputedstar*

\begin{proof}
\Cref{lmm:fd-main} implies immediately that we can use binary search and oracle $\AER$ to compute $d^*$ in polynomial time. 
Moreover, both actions $1$ and $2$ are SSE responses of $\calG_{d^*} = (u^L, \tilde{u}_{d^*}^F)$.
Pick arbitrary $\xx \in \argmax_{\xx' \in f_{d^*}^{-1}(1)} u^L(\xx', 1)$ and $\yy \in \argmax_{\yy' \in f_{d^*}^{-1}(2)} u^L(\yy', 2)$.
Then $(\xx, 1)$ and $(\yy, 2)$ are both SSEs; hence, $u^L(\xx,1) = u^L(\yy,2)$.
Moreover, since $u^L(\xx,1) = V_{d^*}^L(1)$ and $u^L(\yy,2) = V_{d^*}^L(2)$,
according to \Cref{lmm:hat-V1-V2-M-12},
we have
$u^L(\xx,1) = u^L(\yy,2) = M_{\{1,2\}}$.

Indeed, since $u^L(\xx, j) = \gamma_j \cdot \xx + \beta_j$, to compute $\xx$ and $\yy$ amounts to solving 
$\max_{\xx' \in f_{d^*}^{-1}(1)} \va_1 \cdot \xx'$ and 
$\max_{\yy' \in f_{d^*}^{-1}(2)} \va_2 \cdot \yy'$, respectively.
Moreover, by definition, the constraints $\xx' \in f_{d^*}^{-1}(1)$ and $\yy' \in f_{d^*}^{-1}(2)$ are further equivalent to $\va_2 \cdot \xx' \ge d^*$ and $\va_2 \cdot \xx' \le d^*$, respectively.
Hence, the task reduces to solving two LPs, which can be done in polynomial time. 
This completes the proof.
\end{proof}

\section{Omitted Proofs in \NoCaseChange{\Cref{sec:second-ref}}}

\subsection{Results When $\ots$ does not Admit a Cover}

\begin{lemma}
\label{lmm:ots-no-cover}
If $\ots$ does not admit a cover, then $J \subseteq [n]$ and $\xx^* \in \Delta_m$ that satisfy \eqref{eq:J-M} can be computed in polynomial time.
\end{lemma}

\begin{proof}
By \Cref{lmm:cover-iff}, now that $\ots$ does not admit a cover, it must be that $M_\ots = M_{[n]}$. Let $J = \ots$.

Next, pick arbitrary 
\begin{align*}
\yy \in \argmax_{\xx \in \Delta_m} u^L(\xx, 1) 
\quad \text{ and } \quad 
\zz \in \argmax_{\xx \in f_{d^*}^{-1}(1)} u^L(\xx, 1),
\end{align*}
where $f_{d^*}$ is the BR-correspondence of \eqref{eq:tilde-uF-first-ref-pair}.
By definition, we have $u^L(\zz, 1) = M_{\{1\}}$, and according to \Cref{lmm:hat-V1-V2-M-12}, $u^L(\yy, 1) = M_\ots$.
Now that $M_\ots = M_{[n]}$, by \Cref{asn:aj-Mj}, $M_{\{1\}} > M_\ots$.
This further implies that $\zz \notin f_{d^*}^{-1}(1)$.
Hence, we have  
$\va_2 \cdot \yy \le d^* = d_2^*$ and 
$\va_2 \cdot \zz \le d^* = d_2^*$,
which implies that there exists $\xx^*$ in the line segment between $\yy$ and $\zz$ such that $\va_2 \cdot \xx^* = d_2^*$, or equivalently
\[
u^L(\xx^*, 2) = M_\ots.
\]
Since $u^L(\xx, 1) \ge M_\ots$ for $\xx \in \{\yy, \zz\}$, by linearity of $u^L(\cdot, 1)$, we also get that $u^L(\xx^*, 1) \ge M_\ots$.
At the same time, since $\xx^* \in f_{d^*}^{-1} (1)$, it must be that $u^L(\xx^*, 1) \le u^L(\yy, 1) = M_\ots$, according to the definition of $\yy$.
As a result, 
\[
u^L(\xx^*, 1) = M_\ots.
\]

Hence, $J$ and $\xx^*$ satisfy \eqref{eq:J-M}.
To compute $\xx^*$ amounts to finding an $\xx\in \Delta_m$ such that $u^L(\xx, 1) = u^L(\xx, 2) = M_\ots$, which translates to the linear constraints $\va_1 \cdot \xx = d_1^*$ and $\va_2 \cdot \xx = d_2^*$.
Hence, $\xx^*$ can be computed in polynomial time.
\end{proof}

\subsection{Proof of \Cref{lmm:G-d-eps}}
\label{sec:proof-lmm:G-d-eps}

First, the following lemma shows that Condition~(ii) is equivalent to that $P_1(d_1,d_2) \neq \emptyset$ and $P_2(d_1,d_2) \neq \emptyset$.

\begin{restatable}{lemma}{lmmPjnonempty}
\label{lmm:P-j-non-empty}
Suppose that $d_1 \le d_1^*$, and $d_2 \le d_2^*$.
Then $V_{d_1,d_2}^L(j) = \gamma_j \cdot d_j + \beta_j$ if $P_j(d_1, d_2) \neq \emptyset$ for $j\in \ots$.
\end{restatable}

\begin{proof}
We prove the equivalent statement: $\max_{\xx \in P_j(d_1, d_2)} \va_j \cdot \xx = d_j$ if $P_j(d_1, d_2) \neq \emptyset$.
Consider the case where $j = 1$.
Indeed, by definition we have $\va_1 \cdot \xx \le d_1$ for all $\xx \in P_1(d_1, d_2)$, so it suffices to show that $\max_{\xx \in P_1(d_1, d_2)} \va_1 \cdot \xx \ge d_1$.

Choose arbitrary $\xx \in P_1(d_1, d_2)$.
By definition, we have
\[
\va_1 \cdot \xx \le d_1,
\quad \text{and }\quad
\va_2 \cdot \xx \ge d_2.
\]
By \Cref{lmm:hat-u1-M-12}, there exists $\yy \in \Delta_m$ such that 
\[
\va_1 \cdot \yy = d^*_1 \geq d_1,
\quad \text{and }\quad
\va_2\cdot\yy \geq d^*_2\geq d_2.
\]
Hence, there exists a convex combination $\zz$ of $\xx$ and $\yy$, such that $\va_1 \cdot \zz = d_1$ and $\va_2\cdot\zz \geq d_2$, which means $\zz \in P_1(d_1,d_2)$ and hence, 
$\max_{\xx \in P_1(d_1, d_2)} \va_1 \cdot \xx \ge \va_1 \cdot \zz = d_1$.

The case where $j=2$ can be proven analogously.
\end{proof}

Next, we identify a boundary value $\epsilon'$ that makes $P_j(d_1,d_2) \neq \emptyset$ for all $d_j \in [d_1^* - \epsilon', d_j^*]$.
This value is characterized as the optimal solution to the following LP according to \Cref{lmm:epsilon-p-gt-0}:
\begin{subequations}
\label{lp:eps-Pj}
\begin{align}
\max_{\epsilon, \xx_1, \xx_2} \quad 
& \epsilon  \tag{\ref{lp:eps-Pj}} \\
\text{subject to} \quad
& \xx_1 \in \Delta_m \cap P_1(d_1^* - \epsilon, d_2^*) \\
& \xx_2 \in \Delta_m \cap P_1(d_1^*, d_2^* - \epsilon) 
\end{align}
\end{subequations}
The LP characterization is important as it helps us bound the bit-size of $\hat{\epsilon}$.

\begin{restatable}{lemma}{lmmepsilonpgt}
\label{lmm:epsilon-p-gt-0}
Let $\epsilon'$ be the optimal value of LP~\eqref{lp:eps-Pj}.
Then $\epsilon' > 0$. 
Moreover, $P_1(d_1,d_2) \neq \emptyset$ and $P_2(d_1,d_2) \neq \emptyset$ for all $d_1 \in [d_1^* - \epsilon', d_1^*]$ and $d_2 \in [d_2^* - \epsilon', d_2^*]$.
\end{restatable}

\begin{proof}
Indeed, if $\epsilon' > 0$, then according to the constraints of LP~\eqref{lp:eps-Pj}, we have
$P_1(d_1^* - \epsilon', d_2^*) \neq \emptyset$ and $P_2(d_1^*, d_2^* - \epsilon') \neq \emptyset$.
The second part of the statement of this lemma then follows readily: for all $d_1 \in [d_1^* - \epsilon', d_1^*]$ and $d_2 \in [d_2^* - \epsilon', d_2^*]$, we have
\[
P_1(d_1,d_2) \supseteq P_1(d_1^* - \epsilon', d_2^*) \neq \emptyset,
\]
and 
\[
P_2(d_1,d_2) \supseteq P_2(d_1^*, d_2^* - \epsilon') \neq \emptyset.
\]
Hence, it suffices to prove that $\epsilon' > 0$.

According to \Cref{lmm:hat-u2-M-12}, there exists $\xx^- \in \Delta_m$ such that $\va_2 \cdot \xx^- < d^* = d_2^*$.
Let $\epsilon_2' = d_2^* - \va_2 \cdot \xx^-$. 
We have $\epsilon_2' > 0$. Moreover, $\xx^- \in P_2(d_1^*, d_2^* - \epsilon_2')$, which means $P_2(d_1^*, d_2^* - \epsilon_2') \neq \emptyset$.

We next argue that there also exists $\epsilon_1' > 0$ such that $P_1(d_1^* - \epsilon_1', d_2^*)
\neq \emptyset$.
Once this holds, we have $\varepsilon = \min\{\epsilon_1', \epsilon_2'\} > 0$.
Moreover,
\[
P_1(d_1^* - \varepsilon, d_2^*) 
\supseteq P_1(d_1^* - \epsilon_1', d_2^*)
\neq \emptyset,
\]
and 
\[
P_2(d_1^*, d_2^* - \varepsilon)
\supseteq P_2(d_1^*, d_2^* - \epsilon_2')
\neq \emptyset.
\]
Hence, $\varepsilon > 0$, along with two arbitrarily chosen points $\xx_1$ and $\xx_2$ from the above sets, constitutes a feasible solution to \eqref{lp:eps-Pj}, which implies the claimed result.

We next demonstrate the existence of $\epsilon_1'$. Note that it suffices to show that: \begin{equation}
\label{eq:epsilon-p-gt-0-case-2}
\exists\,\xx^* \in P_1(d_1^*, d_2^*): \va_1 \cdot \xx^* < d_1^*.
\end{equation}
Indeed, letting $\epsilon_1' = d_1^* - \va_1 \cdot \xx^*$, we then have $\xx^* \in P_1(d_1^* - \epsilon_1', d_2^*)$, which means $P_1(d_1^* - \epsilon_1', d_2^*) \neq \emptyset$. 
Consider the following two cases.

\paragraph{Case 1.}
$\va_2 \cdot \xx \le d_2^*$ for all $\xx \in \Delta_m$.
In this case, we can prove that 
\[
P_1(d_1^*, d_2^*) = \argmax_{\xx \in \Delta_m} u^L(\xx, 2)
\]
by noting the following facts:
\begin{itemize}
\item     
For all $\xx \in P_1(d_1^*, d_2^*)$, by definition we have $\va_2 \cdot \xx \ge d_2^*$.
Now that $\va_2 \cdot \xx \le d_2^*$, it must be that
$\va_2 \cdot \xx = d_2^*$, which further means 
\[
u^L(\xx, 2) 
= \gamma_2 \cdot d_2^* + \beta_2
= M_\ots
= V_{d_1^*,d_2^*}^L(2)
\]
for all $\xx \in P_1(d_1^*, d_2^*)$,
where $M_\ots = V_{d_1^*,d_2^*}^L(2)$ follows by \Cref{lmm:hat-u2-M-12}.

\item 
For all $\xx \notin P_1(d_1^*, d_2^*)$, we have $\va_2 \cdot \xx < d^*$, so
$u^L(\xx, 2) < M_\ots = V_{d_1^*,d_2^*}^L(2)$.
\end{itemize}

Hence, by \Cref{asn:2-non-dominiated}, there exists $\xx \in \argmax_{\xx' \in \Delta_m} u^L(\xx', 2) = P_1(d_1^*, d_2^*)$ such that 
$u^L(\xx, 1) < u^L(\xx, 2)$.
Now that $\va_2 \cdot \xx \le d^*$ for all $\xx \in \Delta_m$ according to the assumption of Case~1, we have $u^L(\xx, 2) \le M_\ots$, and in turn $u^L(\xx, 1) < M_\ots$.
Therefore, letting $\xx^*$ be the aforementioned $\xx$ gives
$\va_1 \cdot \xx^* < (M_\ots - \beta_1)/\gamma_1 = d_1^*$.

\paragraph{Case 2.}
There exists $\xx \in \Delta_m$ such that $\va_2 \cdot \xx > d_2^*$.
Let 
$\hat{d}_2 = \max_{\xx \in \Delta_m} \va_2 \cdot \xx$; we have $\hat{d}_2 > d_2^*$.
Hence, let $\hat{d}_2' = (\hat{d}_2 + d_2^*)/2$, we have $\hat{d}_2 > \hat{d}_2' > d_2^*$.
By continuity, there exists $\yy \in \Delta_m$ such that $\va_2 \cdot \yy = \hat{d}_2'$.

Now consider a vector $\tilde{\yy} \in \mathbb{R}^m$ such that 
\[
\tilde{y}_i = 
\begin{cases}
y_i + \delta / m' & \text{ if } y_i = 0; \\
y_i - \delta / (m - m') & \text{ if } y_i > 0,
\end{cases}
\]
where $\delta > 0$ and $m' = \left| i \in [m] : y_i = 0 \right|$.
Clearly, $\sum_{i \in [m]} \tilde{y}_i = \sum_{i \in [m]} y_i = 1$, and when $\delta$ is sufficiently close to $0$, we can ensure that: $\tilde{y}_i \in (0, 1)$ for all $i \in [m]$, so $\tilde{\yy}$ is in the interior of $\Delta_m$;
moreover, $\hat{d}_2 > \va_2 \cdot \tilde{\yy} > d_2^*$ by continuity, so $\tilde{\yy} \in P_1(d_1^*, d_2^*)$.

If it happens that $\va_1 \cdot \tilde{\yy} < d_1^*$, then we are done with $\tilde{\yy}$ being a point satisfying \eqref{eq:epsilon-p-gt-0-case-2}.
Hence, in what follows, we assume that $\va_1 \cdot \tilde{\yy} \ge d_1^*$.
According to \Cref{lmm:hat-u1-M-12}, we have 
$\va_1 \cdot \xx \le d_1^*$ for all $\xx \in f_{d^*}^{-1}(1) = P_1(d_1^*, d_2^*)$ (where $f$ is the BR-correspondence of $\tilde{u}^F_d$ defined in \eqref{eq:tilde-uF-first-ref-pair}),
so it must be that $\va_1 \cdot \tilde{\yy} = d_1^*$.

We proceed by defining a set of vectors $\zz^1, \dots, \zz^m \in \mathbb{R}^m$ such that
\[
z_j^i = 
\begin{cases}
\tilde{y}_j - \delta' & \text{ if } j \in [m] \setminus \{i\}; \\
\tilde{y}_j + (m - 1) \cdot \delta' & \text{ if } j = i,
\end{cases}
\]
where $\delta' > 0$.
When $\delta'$ is sufficiently close to $0$, we can ensure that $\hat{d}_2 > \va_2 \cdot \zz^i > d_2^*$.
Moreover, we have $\sum_{j \in [m]} z_j^i = \sum_{j \in [m]} \tilde{y}_j = 1$, and since $\tilde{\yy}$ is in the interior of $\Delta_m$, a sufficiently small $\delta'$ also ensures that $\zz^i \in \Delta_m$.
Hence, $\zz^i \in P_1(d_1^*, d_2^*)$.
Observe that
\begin{align*}
\va_1 \cdot \zz^i 
&= \va_1 \cdot \tilde{\yy} + (m-1) \cdot \delta' \cdot a_{1,i} + \sum_{j \in [m] \setminus \{i\}} \delta' \cdot a_{1,j} \\
&= d_1^* + \delta' \cdot \left[ (m-1) \cdot a_{1,i} - \sum_{j \in [m]} a_{1,j} + a_{1,i} \right] \\
&=d_1^* + \delta' \cdot \left[ m \cdot a_{1,i} - \sum_{j \in [m]} a_{1,j} \right].
\end{align*}
We claim that at least $\va_1 \cdot \zz^i < d_1^*$ for at least one $i \in [m]$.
Indeed, if $\va_1 \cdot \zz^i \ge d_1^*$ for all $i \in [m]$, according to the same argument above via \Cref{lmm:hat-u1-M-12}, we get that $\va_1 \cdot \zz^i = d_1^*$ for all $i \in [m]$.
It follows that 
$a_{1,1} = a_{1,2} = \dots = a_{1,m} = \sum_{j \in [m]} a_{1,j}/m$,
which contradicts \Cref{asn:a1-non-uniform}.
\end{proof}

Using the above results, we now prove \Cref{lmm:G-d-eps}.

\begin{proof}[Proof of \Cref{lmm:G-d-eps}]

We first solve LP~\eqref{lp:eps-Pj} and let the optimal value be $\epsilon'$.
If some $j \in \{1,2\}$ is an SSE response of $\calG_{d_1^*-{\epsilon'},d_2^*-{\epsilon'}}$, then we are done with $\hat{\epsilon} = \epsilon'$ satisfying both conditions in \Cref{lmm:G-d-eps}.
Specifically:
\begin{itemize}
\item 
Condition~(ii) holds according to \Cref{lmm:P-j-non-empty} and \Cref{lmm:epsilon-p-gt-0}.

\item
Further applying \Cref{lmm:P-j-non-empty}, we get that for any $d_j \in [d_j^* - \epsilon', d_j^*]$,
\begin{align}
V_{d_1, d_2}^L(j) 
&= \gamma_j \cdot d_j + \beta_j \nonumber \\
&\ge \gamma_j \cdot (d_j^* - \hat{\epsilon}) + \beta_j 
= V_{d_1^* - \hat{\epsilon}, d_2^* - \hat{\epsilon}}^L(j)
\label{eq:lmm:G-d-eps:eq1}
\end{align}
for both $j \in \ots$, whereas
\begin{equation}
\label{eq:lmm:G-d-eps:eq2}
V_{d_1, d_2}^L(k)
\le V_{d_1^* - \hat{\epsilon}, d_2^* - \hat{\epsilon}}^L(k)
\end{equation}
for all $k \in [n] \setminus \ots$ as $f_{d_1, d_2}^{-1}(k) \subseteq f_{d_1^* - \hat{\epsilon}, d_2^* - \hat{\epsilon}}^{-1}(k)$.
In words, the leader's payoff for playing a strategy that induced actions $1$ or $2$ will not decrease, while that for strategies inducing actions $k \notin \ots$ will not increase. 
Hence, at least one of $j \in \ots$ remains to be an SSE response, and Condition~(i) also holds.
\end{itemize}

Hence, in what follows, we assume that neither action $1$ nor $2$ is an SSE response of $\calG_{d_1^*-{\epsilon'},d_2^*-{\epsilon'}}$.

For each $k \in [n] \setminus\{1,2\}$, we define the following LP, where $\mu$ is the payoff function of $h$.
\begin{subequations}
\label{lp:eps-hat}
\begin{align}
\min_{\epsilon, \xx \in \Delta_m} \quad 
& \epsilon \tag{\ref{lp:eps-hat}} \\
\text{subject to} \quad
& \epsilon \le \epsilon' \label{lp:eps-hat-cons0}\\
& \va_j \cdot \xx \ge d_j^* - \epsilon  && \text{ for all } j \in \{1,2\}  \label{lp:eps-hat-cons1}\\
& \mu(\xx, k) \ge \mu(\xx, k')  && \text{ for all } k' \in [n] \setminus\{1,2\} \label{lp:eps-hat-cons2} \\
& u^L(\xx, k) \ge \gamma_j \cdot (d_j^* - \epsilon) + \beta_j && \text{ for all } j \in \{1,2\} \label{lp:eps-hat-cons3}
\end{align}
\end{subequations}
Hence, we obtain a class of LPs.
Let $\varepsilon$ be the minimum optimal value of LPs in this class that indeed has a feasible solution.
We show that any $\hat{\epsilon} < \varepsilon$ satisfies the conditions in the statement of \Cref{lmm:G-d-eps} via the following claims.

\begin{claim}
If neither action $1$ nor $2$ is an SSE response of $\calG_{d_1^*-{\epsilon'},d_2^*-{\epsilon'}}$, then LP~\eqref{lp:eps-hat} must be feasible for at least one $k \in [n] \setminus \{1,2\}$.
\end{claim}
\begin{proof}
Indeed, pick an arbitrary SSE $(\xx, k)$ of $\calG_{d_1^*-{\epsilon'},d_2^*-{\epsilon'}}$.
Since neither action $1$ nor $2$ is an SSE response, it must be that $k \in [n] \setminus \ots$.
Consider the LP~\eqref{lp:eps-hat}.
$\epsilon'$ and $\xx$ form a feasible solution to this LP as all the constraints are satisfied:
Constraints~\eqref{lp:eps-hat-cons1} and \eqref{lp:eps-hat-cons2} holds because $k$ is a best response against $\xx$;
and Constraint~\eqref{lp:eps-hat-cons3} says that $(\xx,k)$ gives the leader a higher payoff than the best payoff attainable by committing to any strategy in $P_1$ or $P_2$ (when $P_1$ and $P_2$ are non-empty, the right hand side of this constraint is the best attainable payoff in the corresponding region as argued in \Cref{lmm:P-j-non-empty}). 
\end{proof}

\begin{claim}
\label{clm:vareps-ge-0}
$\varepsilon > 0$.
\end{claim}

\begin{proof}

Suppose for the sake of contradiction that $\varepsilon \le 0$.

If $\varepsilon < 0$, then by Constraint~\eqref{lp:eps-hat-cons1} there exists $\xx \in \Delta_m$ such that $\va_j \cdot \xx > d_j^*$ for $j \in \ots$ (i.e., $u^L(\xx,j) > M_\ots$), which would imply that $\max_{\xx' \in \Delta_m} \min_{j\in \ots} u^L(\xx', j) > M_\ots$---a contradiction.

If $\varepsilon = 0$, then $P_0(d_1^*-\varepsilon, d_2^* - \varepsilon) = \calM_{\{1,2\}}$, in which case all $\xx$ satisfying Constraints~\eqref{lp:eps-hat-cons1} and \eqref{lp:eps-hat-cons2} would violate \eqref{lp:eps-hat-cons3}: Constraints~\eqref{lp:eps-hat-cons1} and \eqref{lp:eps-hat-cons2} require that $\xx \in P_0 = \calM_{\{1,2\}}$ and $k \in \widetilde{BR}(\xx)$, respectively, so 
\[
u^L(\xx, k) \le \max_{\xx' \in \calM_\ots} \max_{k' \in \widetilde{\BR}(\xx')} u^L(\xx', k') < M_\ots = \gamma_j \cdot (d_j^* - \varepsilon) + \beta_j,
\]
where the second transition follows by the definition of a cover (see \eqref{eq:cover}).
\end{proof}

\begin{claim}
Any $\hat{\epsilon} < \varepsilon$ satisfies the conditions in the statement of \Cref{lmm:G-d-eps}.
\end{claim}

\begin{proof}
According to \Cref{clm:vareps-ge-0}, $\varepsilon > 0$.
Hence, $\hat{\epsilon} > 0$.
By \Cref{lp:eps-hat-cons0}, we have $\hat{\epsilon} < \varepsilon \le \epsilon'$, so Condition~(ii) holds according to \Cref{lmm:P-j-non-empty} and \Cref{lmm:epsilon-p-gt-0}.

By definition, $\varepsilon$ is the minimum of the optimal solutions of the LPs.
Hence, if we fix $\epsilon$ to $\hat{\epsilon} = \varepsilon / 2$, no $\xx \in \Delta_m$ satisfies the constraints in the above LP for any $k \in [n]\setminus \{1,2\}$. This means that
either $k$ is not a best response against $\xx$ (Constraints~\eqref{lp:eps-hat-cons1} and \eqref{lp:eps-hat-cons2}), or $k$ is a best response but the payoff of $(\xx,k)$ is lower than the highest payoff attainable by committing to a strategy in $P_1$ or $P_2$ (Constraint~\eqref{lp:eps-hat-cons3}).
In summary, $k$ cannot an SSE response in $\calG_{d_1^*-\hat{\epsilon},d_2^*-\hat{\epsilon}}$, so only action $1$ or $2$ can be SSE responses.
The same argument showing \eqref{eq:lmm:G-d-eps:eq1} and \eqref{eq:lmm:G-d-eps:eq2} also shows that only action $1$ or $2$ can be SSE responses of $\calG_{d_1,d_2}$ for all $d_1 \in [d_1^* - \hat{\epsilon}, d_1^*]$ and $d_2 \in [d_1^* - \hat{\epsilon}, d_2^*]$, so Condition~(i) follows.
\end{proof}


We conclude the proof by showing that a desired $\hat{\epsilon}$ can be computed in polynomial time.
Note that we cannot hope to solve the LPs defined in \eqref{lp:eps-hat} to obtain $\hat{\epsilon}$ since we do not know $u^L$. 
Nevertheless, since Condition~(i) holds for $\hat{\epsilon} < \varepsilon$ as we argued above, we can use binary search to find $\hat{\epsilon}$ from the interval $[0, \epsilon']$:
we maintain a candidate value $\epsilon$, query the oracle $\ASSE$ to check if action $1$ or $2$ is an SSE response of $\mathcal{G}_{d_1^*-\epsilon,d_2^*-\epsilon}$, and halve $\epsilon$ if they are not.
The algorithm will terminate when $\epsilon < \varepsilon$, which takes polynomial time as
the bit-size of $\varepsilon$ is bounded from above by a polynomial in the size of the representation of the LPs defined in \eqref{lp:eps-hat}.
\end{proof}

\subsection{Proof of \Cref{thm:second-ref-pair-BR-correspondence}}

\thmsecondrefpairBRcorrespondence*

\begin{proof}
Using \Cref{lmm:G-d-eps}, we first find $\hat{\epsilon}$ satisfying the conditions stated in the lemma.
Let $d_1 = d_1^* - \hat{\epsilon}$ and $d_2 = d_2^* - \hat{\epsilon}$.
Since $d_1 < d_1^*$ and $d_2 < d_2^*$, we have $V_{d_1,d_2}^L(1) < M_\ots$ and $V_{d_1,d_2}^L(2) < M_\ots$.

If both actions $1$ and $2$ are SSE responses of $\calG_{d_1, d_2}$, then we are done:
picking arbitrary $\xx \in \max_{\xx' \in P_1(d_1, d_2)} \va_1 \cdot \xx'$ and $\yy \in \max_{\yy' \in P_2(d_1, d_2)} \va_2 \cdot \yy'$ gives two SSEs $(\xx, 1)$ and $(\yy, 2)$;
we have $u^L(\xx, 1) = u^L(\yy, 2) = V_{d_1,d_2}^L(2) < M_\ots$.

If action $1$ is not an SSE response of $\calG_{d_1, d_2}$, we have $V_{d_1,d_2}^L(1) < V_{d_1,d_2}^L(2)$.
Let 
\[
d_1' = (V_{d_1,d_2}^L(2) - \beta_1) / \gamma_1.
\]

According to Condition~(ii) of \Cref{lmm:G-d-eps}, we have $V_{d,d_2}^L(1) \ge V_{d,d_2}^L(2)$ if and only if $d \ge d_1'$.
This means that action $1$ is an SSE response of $\calG_{d, d_2}$ if and only if $d \ge d_1'$, and we can use oracle $\AER$ and binary search to pin down $d_1'$. 
Meanwhile $V_{d_1',d_2}^L(1) = V_{d_1',d_2}^L(2)$, so both actions $1$ and $2$ are SSE responses when $d = d'_1$. Similarly to the first case, where both actions are SSE responses of $\calG_{d_1,d_2}$, we can obtain a desired reference pair.

Similarly, if action $2$ is not an SSE response of $\calG_{d_1, d_2}$, then we search for a number
\[
d_2' = (V_{d_1,d_2}^L(1) - \beta_2) / \gamma_2.
\]
Both actions $1$ and $2$ are SSE responses of $\calG_{d_1, d_2'}$, and a desired reference pair can be obtained accordingly.
\end{proof}

\subsection{Querying by Using Payoff Matrices}
\label{sec:second-ref-payoff-query}

The above algorithm uses BR-correspondences in the queries.
We next show how to transform the algorithm to one that uses payoff-based queries, thereby solving our problem completely.
Using a similar approach, we first define the following payoff matrix parameterized by two numbers $d_1 < d_1^*$ and $d_2 < d_2^*$. We will show that the BR-correspondence of this payoff matrix functions equivalently to $f_{d_1, d_2}$ defined in \eqref{eq:brc-f-d1-d2-h}.

\begin{boxedtext}
For every $\xx \in \Delta_m$, let
\begin{align}
\label{eq:tilde-uF-d1-d2}
\tilde{u}_{d_1, d_2}^F(\xx,j) \coloneqq 
\begin{cases}
\vb_j \cdot  (\xx - \zz_j) + c_j, & \text{ if } j \in \ots;\\
\mu(\xx,j), & \text{ if } j \in [n] \setminus \ots.
\end{cases}
\end{align}
where $\mu$ is an arbitrary proper cover of $\ots$; and 
$\vb_j, c_j \in \mathbb{R}$ and $\zz_j \in \Delta_m$ are functions of $d_1$ and $d_2$, defined as follows.\footnotemark

\begin{itemize}
\item
First, for each $j \in \ots$, define
\begin{align}
&Z_1(d_1, d_2) \coloneqq \left\{\xx \in \Delta_m : \va_1 \cdot \xx = \bar{d}_1 \text{ and }  \va_2 \cdot \xx \ge \bar{d}_2  \right\}, \label{eq:Z-1} \\
\text{ and }\quad
&Z_2(d_2, d_2) \coloneqq \left\{\xx \in \Delta_m : \va_2 \cdot \xx = \bar{d}_2 \right\} \label{eq:Z-2},
\end{align}
where for $j \in \ots$:
\[
\bar{d}_j \coloneqq (d_j + d_j^*)/2.
\]

\item
Let $K_1 = [n] \setminus \ots$ and $K_2 = [n] \setminus \{2\}$.
We define
\begin{align}
\label{eq:c-1}
c_1 = \max_{\xx \in Z_1, k \in K_1} \tilde{u}_{d_1, d_2}^F(\xx,k),
\end{align}
and let $\zz_1$ and $k_1$ be an arbitrary solution corresponding to the above maximum value, i.e.,
\begin{align}
\label{eq:zz-1}
\zz_1 &\in \argmax_{\xx \in Z_1} \max_{k \in K_1} \tilde{u}_{d_1, d_2}^F(\xx,k), \\
\label{eq:k-1}
\text{and} \quad
k_1 &\in \argmax_{k \in K_1} \tilde{u}_{d_1, d_2}^F(\zz_1,k).
\end{align}
Note that according to \eqref{eq:tilde-uF-d1-d2}, $\tilde{u}_{d_1, d_2}^F(\zz_1,k) = \mu(\zz_1, k)$ for all $k \in K_1$, so the above definitions in fact only rely on $\mu$.
Furthermore, let
\begin{align}
\vb_1 \coloneqq \gg_{k_1} - \frac{6 \cdot W_1}{d_1^* - d_1} \cdot \va_1,
\end{align}

where we use
\[
W_j \coloneqq 1 + \max_{i\in[m], k\in K_j} \tilde{u}_{d_1,d_2}^F(i,k) - \min_{i\in[m], k\in K_j} \tilde{u}_{d_1,d_2}^F(i,k)
\]
for $j \in \ots$, i.e., it is a value strictly larger than the maximum gap between the payoff parameters;
in addition, for all $k \in [n]$, we use
\[
\gg_k \coloneqq \left( \tilde{u}_{d_1,d_2}^F(1,k), \dots, \tilde{u}_{d_1,d_2}^F(m,k) \right).
\]

\item
$c_2$, $\zz_2$, and $k_2$ are defined analogously by changing the labels in the above definitions from $1$ to $2$.
Note that since $K_2 = K_1 \cup \{1\}$, these parameters also depend on $\tilde{u}_{d_1, d_2}^F(\xx,1) = \vb_1 \cdot (\xx - \zz_1) + c_1$, with $c_1$ and $\zz_1$ defined above.

\end{itemize}
\end{boxedtext}

\footnotetext{We omit the dependencies of these parameters on $d_1$ and $d_2$ in the notation to simplify the presentation.}

With the above definitions, we show several properties of $\tilde{u}_{d_1,d_2}^F$ in the subsequent lemmas.
We extend the notation in \Cref{sec:2nd-ref-BR-query}: hereafter, we let $\tilde{f}_{d_1,d_2}$ be the BR-correspondence of $\tilde{u}_{d_1,d_2}^F$, 
let $\widetilde{\calG}_{d_1,d_2} \coloneqq (u^L, \tilde{u}_{d_1,d_2}^F)$,
and let
$\widetilde{V}_{d_1,d_2}^L(j) = \max_{\xx \in \tilde{f}_{d_1,d_2}^{-1}(j)} u^L(\xx, j)$.

\begin{figure}[t]
	\center
	\begin{tikzpicture}
	
	\begin{ternaryaxis}[
	ternary limits relative=false,
	width= 50mm,
 	xtick=\empty,
 	ytick=\empty,
 	ztick=\empty,
	area style,
	grid=none,
	clip=false,
	xmin=0, xmax=1.0,
	ymin=0, ymax=1.0,
	zmin=0, zmax=1.0,
    axis on top,
	]
	
	\addplot3 coordinates {
		(0.12, 0.88, 0)
		(0.32, 0.38, 0.3)
		(0.68, 0.08, 0.24)
        (0.95, 0, 0.05)
        (1, 0, 0)
	};
	
	\addplot3 coordinates {
		(0, 1, 0)
		(0.12, 0.88, 0)
		(0.43, 0.1, 0.47)
		(0.3, 0, 0.7)
		(0, 0, 1)
	};
	
	\addplot3 coordinates {
        (0.95, 0, 0.05)
		(0.68, 0.08, 0.24)
		(0.32, 0.38, 0.3)
		(0.432, 0.1, 0.468)
		(0.302, 0, 0.698)
	};
	
    \draw[black,thick] (0.12, 0.88, 0) -- (0.43, 0.1, 0.47) -- (0.3, 0, 0.7);
	\draw[black,thick] (0.32, 0.38, 0.3) -- (0.68, 0.08, 0.24) -- (0.95, 0, 0.05);
	
 	\draw[black,thick, dotted] (0.05, 1.05, -0.1) -- (0.3, -0.1, 0.8);
 	\draw[gray, thick, dotted] (0.25, 0.85, -0.1) -- (0.47, -0.1, 0.63);
 	
 	\draw[black,thick, dotted] (0.16, 0.6, 0.24) -- (1.1, 0, -0.1);
	\draw[gray, thick, dotted] (0.39, 0.26, 0.35) -- (0.97, -0.1, 0.13);

	\node at (0.56, 0.26, 0.18) {\large $1$};
	\node at (0.15, 0.25, 0.6) {\large $2$};
	\node at (0.54, 0.08, 0.38) {\large $\mu$};
	
	\node [rotate=12.5, anchor=west] at (0.3, -0.1, 0.8) {\small $\va_{2} \cdot \xx = d_2$};
	\node [rotate=12.5, anchor=west] at (1.1, 0, -0.1) {\small $\va_{1} \cdot \xx = d_1$};
	
	\node [gray, rotate=12.5, anchor=west] at (0.47, -0.1, 0.63) {\small $\va_{2} \cdot \xx = \bar{d}_2$};
	\node [gray, rotate=12.5, anchor=west] at (0.97, -0.1, 0.13) {\small $\va_{1} \cdot \xx = \bar{d}_1$};

	\node[circle,inner sep=1.3pt,fill=BrickRed] at (0.43, 0.1, 0.47) {};
	\node[BrickRed] at (0.37, 0.12, 0.51) {\small $\zz_2$};
	
	\node[circle,inner sep=1.3pt,fill=BlueViolet] at (0.68, 0.08, 0.24) {};
	\node[BlueViolet] at (0.68, 0.13, 0.19) {\small $\zz_1$};

	\end{ternaryaxis}
	
	\end{tikzpicture}
	\caption{
	Illustration of $\tilde{f}_{d_1,d_2}$. The triangle represents the strategy space $\Delta_m$ of the leader. The regions labeled $1$ and $2$ are $\tilde{f}_{d_1, d_2}^{-1}(1)$ and $\tilde{f}_{d_1, d_2}^{-1}(1)$, respectively. The region labeled $\mu$ is $\bigcup_{k\in [n] \setminus \ots} \tilde{f}_{d_1, d_2}^{-1}(k)$.
	The dotted lines depicts the boundaries of 
$f_{d_1, d_2}^{-1}(j)$ and $f_{\bar{d}_1, \bar{d}_2}^{-1}(j)$. Recall that $f_{d_1, d_2}^{-1}(j) = P_j(d_1,d_2)$ and $f_{\bar{d}_1, \bar{d}_2}^{-1}(j) = P_j(\bar{d}_1,\bar{d}_2)$. 
	\label{fig:brc-payoff}}
\end{figure}

We make several observations about $\tilde{f}_{d_1,d_2}$ in the following lemmas.
The first lemma \Cref{lmm:2nd-ref-payoff-f-supset} characterizes the structure of $\tilde{f}_{d_1,d_2}$ in comparison to $f_{d_1,d_2}$ defined in \Cref{sec:2nd-ref-BR-query}. 
\Cref{fig:brc-payoff} provides an illustration of this characterization.

\begin{restatable}{lemma}{lmmtndrefpayofffsupset}
\label{lmm:2nd-ref-payoff-f-supset}
Suppose that $d_j \le d_j^*$ and $d'_j \in [\bar{d}_j, d_j^*]$ for each $j \in \ots$.
Then the following statements hold for all $j \in \ots$ and $k \in [n] \setminus \ots$:
\begin{itemize}
\item[(i)] 
$\va_j\cdot\xx \leq d'_j$ for all $\xx \in \tilde{f}_{d_1,d_2}^{-1}(j)$;

\item[(ii)] 
$\bigcup_{j \in \ots} f_{d_1, d_2}^{-1}(j) \subseteq \bigcup_{j \in \ots}\tilde{f}_{d_1,d_2}^{-1}(j)$; and

\item[(iii)]
$f_{d'_1,d'_2}^{-1}(k) \subseteq \tilde{f}_{d_1,d_2}^{-1}(k) \subseteq f_{d_1,d_2}^{-1}(k)$.
\end{itemize}
\end{restatable}

\begin{proof}
We prove each of the statements.

\paragraph{Claim (i).}
Suppose that $\va_j \cdot \xx > d'_j$. We argue that $\xx \notin  \tilde{f}_{d_1,d_2}^{-1}(j)$ to prove the claim.

Indeed, we have $\va_j \cdot (\xx - \zz_j) > 0$ and hence,

\begin{align}
\tilde{u}_{d_1,d_2}^F(\xx, j) 
&= \vb_j \cdot (\xx - \zz_j) + c_j \nonumber \\
&= \gg_{k_j} \cdot (\xx - \zz_j) - \frac{6 \cdot W_j}{d_j^* - d_j} \cdot \va_j \cdot (\xx - \zz_j) + c_j \nonumber \\
&< \gg_{k_j} \cdot (\xx - \zz_j) + c_j \nonumber \\
&= \tilde{u}_{d_1,d_2}^F(\xx, k_j) - \tilde{u}_{d_1,d_2}^F(\zz_j, k_j)  + c_j \nonumber \\
&= \tilde{u}_{d_1,d_2}^F(\xx, k_j), \label{eq:payoff-claim-i-eq1}
\end{align}
where we used the fact that $\tilde{u}_{d_1,d_2}^F(\zz_j, k_j)  = c_j$ according to the definitions of $\zz_j$, $k_j$ and $c_j$.
Hence, $\xx \notin \tilde{f}_{d_1,d_2}^{-1}(j)$, and Claim~(i) holds.

\paragraph{Claim (ii).}
Pick arbitrary $j \in \ots$ and $\xx \in P_j(d_1, d_2)$.
We argue that 
$\tilde{u}_{d_1,d_2}^F (\xx, j) > \tilde{u}_{d_1,d_2}^F (\xx, k)$, which implies Claim~(ii) immediately.

Indeed, since $\xx \in P_j(d_1, d_2)$, by definition, we have $\va_j \cdot \xx \le d_j$. Hence,
\begin{align}
\tilde{u}_{d_1,d_2}^F(\xx, j) 
&= \vb_j \cdot (\xx - \zz_j) + c_j \nonumber \\
&= \gg_{k_j} \cdot (\xx - \zz_j) - \frac{4 \cdot W_j}{d_j^* - d_j} \cdot \va_j \cdot (\xx - \zz_j) + c_j \nonumber \\
&> - 2 \cdot W_j - \frac{6 \cdot W_j}{d_j^* - d_j} \cdot (d_j - \bar{d}_j) + c_j \nonumber \\
&= W_j + c_j \nonumber \\
&\ge \tilde{u}_{d_1,d_2}^F(\xx, k) \label{eq:payoff-claim-ii-eq1}
\end{align}
for all $k \in [n] \setminus \ots$,
where we used the fact that
\[
\gg_{k_j} \cdot (\xx - \zz_j) \ge \min_{\xx' \in \Delta_m} \gg_{k_j} \cdot \xx' - \max_{\xx' \in \Delta_m} \gg_{k_j} \cdot \xx' > W_j.
\]

\paragraph{Claim (iii).}
Consider the BR-correspondence $h: \Delta_m \to 2^{[n] \setminus \ots}$ of $\mu$.
For any $d_1$ and $d_2$, according to the construction of $f_{d_1,d_2}$ and $\tilde{f}_{d_1,d_2}$, we have 
\begin{align*}
& f_{d_1,d_2}^{-1}(k) = h^{-1}(k) \cap  P_0(d_1, d_2) \\
\text{and }\quad
& \tilde{f}_{d_1,d_2}^{-1}(k) = h^{-1}(k) \cap \widetilde{P}_0(d_1, d_2),
\end{align*}
where 
\[
\widetilde{P}_0(d_1, d_2) 
\coloneqq \left\{\xx \in \Delta_m : \max_{k \in [n] \setminus \ots} \tilde{u}_{d_1, d_2}^F(\xx, k) \ge \max_{j \in \ots} \tilde{u}_{d_1, d_2}^F(\xx, j) \right\}.
\]
We argue that 
\begin{equation}
\label{eq:P0-wtP0}
P_0(d'_1, d'_2) \subseteq \widetilde{P}_0(d_1, d_2) \subseteq P_0(d_1, d_2)
\end{equation}
to complete the proof.

Indeed, we have 
\begin{align*}
&P_0(d'_1, d'_2) = \closure\left(\Delta_m \setminus \bigcup_{j\in \ots} P_j(d'_1, d'_2) \right) \\
\text{and }\quad
&P_0(d_1, d_2) = \closure\left(\Delta_m \setminus \bigcup_{j\in \ots} P_j(d_1, d_2) \right),
\end{align*}
where $\closure(\cdot)$ denotes the closure of a set.
Since \eqref{eq:payoff-claim-i-eq1} and \eqref{eq:payoff-claim-ii-eq1} are strictly satisfied,
we also have 
\[
\widetilde{P}_0(d_1, d_2) = \closure \left(
\Delta_m \setminus \bigcup_{j \in \ots} \tilde{f}_{d_1,d_2}^{-1}(j)
\right).
\]
According to Claims~(i) and (ii) we proved above,
\[
\bigcup_{j\in \ots} P_j(d'_1, d'_2) 
\supseteq
\bigcup_{j \in \ots} \tilde{f}_{d_1,d_2}^{-1}(j)
\supseteq
\bigcup_{j\in \ots} P_j(d_1, d_2).
\]
Hence, \eqref{eq:P0-wtP0} holds, and this proves Claim~(iii).
\end{proof}

The next lemma presents a result similar to \Cref{lmm:P-j-non-empty}.
Regarding \Cref{fig:brc-payoff}, the result indicates that $\zz_j$ indeed hits the line $\va_j \cdot \xx = \bar{d}_j$ whenever $Z_j \neq \emptyset$.

\begin{restatable}{lemma}{lmmtndrefhatuL}
\label{lmm:2nd-ref-hat-u-L}
Suppose that $j \in \ots$ and $Z_j \neq \emptyset$.
Then $\zz_j \in \tilde{f}_{d_1,d_2}^{-1}(j)$ and  $\widetilde{V}_{d_1,d_2}^L(j) = \gamma_j \cdot \bar{d}_j + \beta_j$.
\end{restatable}

\begin{proof}
It suffices to prove that $\zz_j \in \tilde{f}_{d_1,d_2}^{-1}(j)$.
Indeed, this implies that 
\[
\widetilde{V}_{d_1,d_2}^L(j) \ge u^L(\zz_j, j) = \gamma_j \cdot \bar{d}_j + \beta_j
\]
given that $\va_j \cdot \zz_j = \bar{d}_j$ by definition; 
Moreover, according to \Cref{lmm:2nd-ref-payoff-f-supset}, 
$\va_j \cdot \xx \le \bar{d}_j$ for all $\xx \in \tilde{f}_{d_1,d_2}^{-1}(j)$; hence,
\[
\widetilde{V}_{d_1,d_2}^L(j) = \max_{\xx \in \tilde{f}_{d_1,d_2}^{-1}(j)} u^L(\xx, j) \le \gamma_j \cdot \bar{d}_j + \beta_j.
\]
As a result, $\widetilde{V}_{d_1,d_2}^L(j) = \gamma_j \cdot \bar{d}_j + \beta_j$.

Next, we prove that $\zz_j \in \tilde{f}_{d_1,d_2}^{-1}(j)$ to complete the proof.
By definition (i.e., \eqref{eq:tilde-uF-d1-d2}),
\begin{align}
\label{eq:payoff-claim-i}
\tilde{u}_{d_1,d_2}^F(\zz_j, j) 
= \vb_j \cdot (\zz_j - \zz_j) + c_j 
= c_j 
\ge \tilde{u}_{d_1,d_2}^F(\zz_j, k) 
\end{align}
for all $k \in K_j$.
Since $K_2 = [n] \setminus \{1\}$, this immediately implies that action $2$ is a best response to $\zz_2$, so
$\zz_2 \in \tilde{f}_{d_1,d_2}^{-1}(2)$.

To argue that $\zz_1 \in \tilde{f}_{d_1,d_2}^{-1}(1)$, since $K_1 = [n] \setminus \ots$, we need to prove in addition that $\tilde{u}_{d_1,d_2}^F(\zz_1, 1) \ge \tilde{u}_{d_1,d_2}^F(\zz_1, 2)$.
Indeed, by definition, we have $\va_2 \cdot \zz_1 \ge \bar{d}_2$ (i.e., see \eqref{eq:Z-1} and \eqref{eq:zz-1}), so similar to \eqref{eq:payoff-claim-i-eq1}, we have that $\tilde{u}_{d_1,d_2}^F(\zz_1, 2) 
\le \tilde{u}_{d_1,d_2}^F(\zz_1, k_2)$. Since $k_2 \in K_1 \cup \{1\}$, applying \eqref{eq:payoff-claim-i} then gives
\[
\tilde{u}_{d_1,d_2}^F(\zz_1, 2) 
\leq \tilde{u}_{d_1,d_2}^F(\zz_1, k_2)
\le \tilde{u}_{d_1,d_2}^F(\zz_1, 1).
\]
Consequently, $\tilde{u}_{d_1,d_2}^F(\zz_1, 1) \ge \tilde{u}_{d_1,d_2}^F(\zz_1, k)$ for all $k \in [n]$, so $\zz_1 \in \tilde{f}_{d_1,d_2}^{-1}(1)$.
This completes the proof.
\end{proof}

Finally, the following lemma presents a result similar to \Cref{lmm:G-d-eps}. The result is key to the first step of our approach, where we restrict the search space in a one-dimensional interval.

\begin{restatable}{lemma}{lmmtndrefpayoffotSSEresponse}
\label{lmm:2nd-ref-payoff-1-2-SSE-response}
Suppose that $\epsilon' > 0$ is the optimal value of LP~\eqref{lp:eps-Pj}.
It holds for all $\delta_1 \in [d_1^* - \epsilon', d_1^*]$ and $\delta_2 \in [d_2^* - \epsilon', d_2^*]$ that: 
\begin{itemize}
\item[(i)] 
$Z_j(\delta_1,\delta_2) \neq \emptyset$ for both $j \in \ots$; and

\item[(ii)] 
if at least one of $j \in \ots$ is an SSE response of $\widetilde{\calG}_{\delta_1, \delta_2}$, then the same holds for $\widetilde{\calG}_{\delta'_1, \delta'_2}$, with arbitrary $\delta'_1 \in [\bar{\delta}_1, d_1^*]$ and $\delta'_2 \in [\bar{\delta}_2, d_2^*]$, where $\bar{\delta}_j = (\delta_j + d_j^*)/2$.
\end{itemize}
\end{restatable}

\begin{proof}
First, consider Condition~(i).
If $\epsilon \le \epsilon'$, we have $\bar{\delta}_j \in [\delta_j, d_j^*] \subseteq [d_j^* - \epsilon', d_j^*]$.
Hence, according to \Cref{lmm:G-d-eps}, we have $P_j(\bar{\delta}_1, \bar{\delta}_2) \neq \emptyset$.
By definition, 
\[
Z_j(\delta_1, \delta_2) = P_j(\bar{\delta}_1, \bar{\delta}_2) \cap \{\xx \in \Delta_m : \va_j \cdot \xx = \bar{\delta}_j \}.
\]
It then suffices to argue that there exists $\xx \in P_j(\bar{\delta}_1, \bar{\delta}_2)$ such that $\va_j \cdot \xx = \bar{\delta}_j$.
Indeed, now that $P_j(\bar{\delta}_1, \bar{\delta}_2) \neq \emptyset$, applying \Cref{lmm:P-j-non-empty}, we get that there exists $\xx \in P_j(\bar{\delta}_1, \bar{\delta}_2)$ such that $u^L(\xx, j) = \gamma_j \cdot \bar{\delta}_j + \beta_j$, which immediately implies that $\va_j \cdot \xx = \bar{\delta}_j$.

Now consider Condition~(ii).
Now that Condition~(i) holds, we have $Z_j(\delta_1, \delta_2) \neq \emptyset$ and $Z_j(\delta'_1, \delta'_2) \neq \emptyset$ for both $j \in \ots$. 
According to \Cref{lmm:2nd-ref-hat-u-L}, we then have
\begin{align}
\label{eq:lmm:2nd-ref-payoff-1-2-SSE-response-eq1}
\widetilde{V}_{\delta_1, \delta_2}^L(j)
&= \gamma_j \cdot \bar{\delta}_j + \beta_j \nonumber \\
&\le \gamma_j \cdot \bar{\delta}'_j + \beta_j
= \widetilde{V}_{\delta'_1, \delta'_2}^L(j),
\end{align}
where $\bar{\delta}'_j = (\delta'_j + d_j^*)/2$.

Moreover, since some $j \in \ots$ is an SSE response of $\widetilde{\calG}_{\delta_1, \delta_2}$, by definition,
this means that
\begin{align}
\label{eq:lmm:2nd-ref-payoff-1-2-SSE-response-eq2}
\widetilde{V}_{\delta_1, \delta_2}^L(j) \ge \widetilde{V}_{\delta_1, \delta_2}^L(k)
\end{align}
for all $k \in [n] \setminus \ots$.
According to Claim~(iii) of \Cref{lmm:2nd-ref-payoff-f-supset} (and note that $\delta'_j \in [\bar{\delta}_j, d_j^*]$), we have $\tilde{f}_{\delta'_1, \delta'_2}^{-1}(k) \subseteq f_{\delta'_1, \delta'_2}^{-1}(k) \subseteq \tilde{f}_{\delta_1, \delta_2}^{-1}(k)$.
As a result, $\tilde{f}_{\delta'_1, \delta'_2}^{-1}(k) \subseteq \tilde{f}_{\delta_1, \delta_2}^{-1}(k)$, which implies that
\begin{align}
\label{eq:lmm:2nd-ref-payoff-1-2-SSE-response-eq3}
\widetilde{V}_{\delta'_1, \delta'_2}^L(k) \le 
\widetilde{V}_{\delta_1, \delta_2}^L(k).
\end{align}
Combining \eqref{eq:lmm:2nd-ref-payoff-1-2-SSE-response-eq1}--\eqref{eq:lmm:2nd-ref-payoff-1-2-SSE-response-eq3},
we get that
\[
\widetilde{V}_{\delta'_1, \delta'_2}^L(j)
\ge
\widetilde{V}_{\delta'_1, \delta'_2}^L(k)
\]
for all $k \in [n] \setminus \ots$.
Therefore, only actions $1$ and $2$ can be SSE responses of $\widetilde{\calG}_{\delta_1, \delta_2}$.
\end{proof}

With the above results in hand, we summarize our approach in the following theorem, which completes our task.
The approach is similar to that used to prove \Cref{thm:second-ref-pair-BR-correspondence},
where we first search for values for $d_1$ and $d_2$ that make at least one of action $1$ and $2$ an SSE response. Then we further expand the best response region (see \Cref{fig:brc-payoff}) corresponding to the action that is not yet an SSE response while fixing the region of the other action, until both actions are SSE responses. 

We remark that despite the similarity, one difference between the approaches in \Cref{thm:second-ref-pair-payoff-query} and  \Cref{thm:second-ref-pair-BR-correspondence} is that once we find a position where one of $j\in \ots$, say action $1$, is an SSE response, we cannot simply fix $d_1$ and and increase $d_2$ in the hope of finding another position where both of them are SSE responses.
This is because the way $\tilde{u}_{d_1,d_2}^F$ is constructed does {\em not} ensure that 
$\tilde{f}_{d_1, d_2}^{-1}(2) \subseteq \tilde{f}_{d_1, d'_2}^{-1}(2)$ for any $d_2 \le d'_2$ as our choice of $\zz_j$ is arbitrary from the set $Z_j$.\footnote{It might be possible to ensure this by using more carefully selected $\zz_j$. However, this introduces other complexities in the approach, as well as the presentation of it.}
Consequently, we can not ensure that $\tilde{f}_{d_1, d_2}^{-1}(k) \supseteq \tilde{f}_{d_1, d'_2}^{-1}(k)$ for the other actions $k \in [n]\setminus$, which means $\widetilde{V}_{d_1,d_2}^L(k)$ may grow when we increase $d_2$, potentially exceeding $\widetilde{V}_{d_1,d_2}^L(1)$ and $\widetilde{V}_{d_1,d_2}^L(2)$, in which case $1$ and $2$ cannot be SSE responses anymore.
To resolve this issue, once we find a position $(d_1,d_2)$ such that one of $j\in \ots$ is an SSE response, our approach is to first ``jump''  to $(\bar{d}_1, \bar{d}_2)$.
This ensures that $\tilde{f}_{d_1, d_2}^{-1}(k) \supseteq \tilde{f}_{d_1, d'_2}^{-1}(k)$ 
for all $k \in [n] \setminus \ots$ according to the structure of $\tilde{f}_{d_1, d_2}$ we demonstrated in \Cref{lmm:2nd-ref-payoff-f-supset} (and \Cref{fig:brc-payoff}).

\thmsecondrefpayoffquery*

\begin{proof}
The algorithm is similar to the approach used in \Cref{sec:2nd-ref-BR-query} and proceeds as follows.

\smallskip
\begin{itemize}
\item Step 1. 
Let $\epsilon'$ be the optimal value of LP~\eqref{lp:eps-Pj}.
We first search for an $\epsilon \in (0, \epsilon']$ such that one of $j \in \ots$ is an SSE response of $\widetilde{\calG}_{d_1^* - \epsilon, d_2^* - \epsilon}$.
\end{itemize}
\smallskip

Specifically, we use binary search to find such an $\epsilon$:
we start with the candidate value $\varepsilon = \epsilon'$ and halve $\varepsilon$ if none of $j \in \ots$ is an SSE response of $\widetilde{\calG}_{d_1^* - {\varepsilon}, d_2^* - {\varepsilon}}$.
Note that by \Cref{lmm:epsilon-p-gt-0}, $\epsilon' > 0$, so applying \Cref{lmm:2nd-ref-payoff-1-2-SSE-response} we get that $Z_j(d_1^* - \varepsilon, d_2^* - \varepsilon) \neq \emptyset$, which means that $\widetilde{\calG}_{d_1^* - {\varepsilon}, d_2^* - {\varepsilon}}$ is well defined.

Moreover, the binary search procedure will terminate in polynomial time.
This is due to the following fact.
Let $\hat{\epsilon}$ be a number satisfying the conditions stated in \Cref{lmm:G-d-eps}, which according to the same lemma exists and has a polynomial bit-size.
For any $\varepsilon \le \hat{\epsilon}$, we can show that at least one of $j \in \ots$ must be an SSE response of $\widetilde{\calG}_{d_1^* - {\varepsilon}, d_2^* - {\varepsilon}}$, so it follows immediately that the binary search procedure terminates in polynomial time.
Indeed, according to \Cref{lmm:G-d-eps}, some $j \in \ots$ is an SSE response of ${\calG}_{d_1^* - {\varepsilon}, d_2^* - {\varepsilon}}$, which means that
\[
\max_{j \in \ots} V_{d_1^* - \varepsilon, d_2^* - \varepsilon}^L(j)
\ge
\max_{k \in [n] \setminus \ots} V_{d_1^* - \varepsilon, d_2^* - \varepsilon}^L(k).
\]
We can show that the same holds with respect to $\widetilde{V}_{d_1^* - \varepsilon, d_2^* - \varepsilon}^L$.
Specifically, it suffices to show that
\begin{align}
& \widetilde{V}_{d_1^* - \varepsilon, d_2^* - \varepsilon}^L(j) 
\ge V_{d_1^* - \varepsilon, d_2^* - \varepsilon}^L(j) \quad \text{ for all } j \in \ots, 
\label{eq:thm:second-ref-pair-payoff:eq1} \\
\text{ and }\quad
& \widetilde{V}_{d_1^* - \varepsilon, d_2^* - \varepsilon}^L(k) 
\le V_{d_1^* - \varepsilon, d_2^* - \varepsilon}^L(k) \quad \text{ for all } k \in [n] \setminus \ots.
\label{eq:thm:second-ref-pair-payoff:eq2}
\end{align}
\eqref{eq:thm:second-ref-pair-payoff:eq2} follows directly by \Cref{lmm:2nd-ref-payoff-f-supset}, i.e., $\tilde{f}_{d_1^* - \varepsilon, d_2^* - \varepsilon}^{-1}(k) \subseteq f_{d_1^* - \varepsilon, d_2^* - \varepsilon}^{-1}(k)$.

To see that \eqref{eq:thm:second-ref-pair-payoff:eq1} holds, note that since $Z_j(d_1^* - {\varepsilon}, d_2^* - {\varepsilon}) \neq \emptyset$, 
by \Cref{lmm:2nd-ref-hat-u-L}, we have
\[
\widetilde{V}_{d_1^* - \varepsilon, d_2^* - \varepsilon}^L(j) 
= \gamma_j \cdot \frac{d_j^* + (d_j^* - \varepsilon)}{2} + \beta_j.
\]
Since $\varepsilon > 0$, it follows that 
\[
\widetilde{V}_{d_1^* - \varepsilon, d_2^* - \varepsilon}^L(j)
> \gamma_j \cdot (d_j^* - \varepsilon) + \beta_j
\ge V_{d_1^* - \varepsilon, d_2^* - \varepsilon}^L(j),
\]
where the second inequality follows by the fact that $\va_j \cdot \xx \le d_j^* - \varepsilon$ for all $\xx \in f_{d_1^* - \varepsilon, d_2^* - \varepsilon}^{-1}(j) = P_j(d_1^* - \varepsilon, d_2^* - \varepsilon)$.

\smallskip
\begin{itemize}
\item Step 2.
Let $\epsilon$ be the outcome of Step~1.
Let $d_j = d_j^* - {\epsilon}$ and $\bar{d}_j = (d^* + d_j)/2$ for $j \in \ots$.
We search for $\delta_j \in [\bar{d}_j, d_j^*)$ such that both actions $1$ and $2$ are SSE responses of $\widetilde{\calG}_{\delta_1, \delta_2}$.
\end{itemize}
\smallskip

Indeed, if both $1$ and $2$ are SSE responses of $\widetilde{\calG}_{\bar{d}_1, \bar{d}_2}$, then we are done.
Otherwise, suppose that only action $1$ is an SSE response.
Then we fix $\delta_1 = \bar{d}_1$ and search for a value $\delta_2$ such that
\[
(\delta_2+ d^*_2)/2 = (\widetilde{V}_{\bar{d}_1, \bar{d}_2}^L(1) - \beta_2)/\gamma_2.
\]
Given \Cref{lmm:2nd-ref-hat-u-L}, and now that $Z_j(\delta_1, \delta_2) \neq \emptyset$ according to \Cref{lmm:2nd-ref-payoff-1-2-SSE-response},
only when $\delta_2' = \delta_2$, we have $\widetilde{V}_{\delta_1, \delta_2'}^L(1) = \widetilde{V}_{\delta_1, \delta_2'}^L(2)$, in which case both actions are SSE responses of $\widetilde{\calG}_{\delta_1, \delta_2'}$.
Hence, with access to $\AER$, we can use binary search to find out $\delta_2$: whenever action $2$ is not an SSE response of $\widetilde{\calG}_{\delta_1, \delta_2'}$, we know that $\delta'_2 < \delta_2$; and whenever action $1$ is not, we know that $\delta'_2 > \delta_2$.
Similarly, if only action $2$ is an SSE response of $\widetilde{\calG}_{\bar{d}_1, \bar{d}_2}$, we fix $\delta_2 = \bar{d}_2$ and search for $\delta_1 = (\widetilde{V}_{\bar{d}_1, \bar{d}_2}^L(2) - \beta_1)/\gamma_1$. 

\smallskip
\begin{itemize}
\item Step 3.
Finally, consider $\tilde{u}_{\delta_1, \delta_2}^F$ and let $\zz_1$ and $\zz_2$ be the parameters corresponding to this payoff matrix.
It can be verified that $(\zz_1, 1)$ and $(\zz_2, 2)$ are a reference pair such that $u^L(\zz_1, 1) = u^L(\zz_2, 2) < M_\ots$.
\end{itemize}
\smallskip

Indeed, according to \Cref{lmm:2nd-ref-hat-u-L}, $j$ is a best response of $\zz_j$ for both $j \in \ots$; moreover, 
\begin{align*}
u^L(\zz_j, j) 
= \widetilde{V}_{\delta_1,\delta_2}^L(j) 
&= \gamma_j \cdot \bar{\delta}_j + \beta_j \\
&< \gamma_j \cdot d_j^* + \beta_j 
= M_\ots,
\end{align*}
where $\bar{\delta}_j \coloneqq (d_j^* + \delta_j)/2 < d_j^*$.
Now that both actions are SSE responses, we have $\widetilde{V}_{\delta_1,\delta_2}^L(1) = \widetilde{V}_{\delta_1,\delta_2}^L(2)$.
Hence, $u^L(\zz_1, 1) = u^L(\zz_2, 2) < M_\ots$. 
\end{proof}

\end{document}